\documentclass{article}
\usepackage{arxiv}
\usepackage[utf8]{inputenc} 
\usepackage[T1]{fontenc}    
\usepackage{hyperref}       
\usepackage{url}            
\usepackage{booktabs}       
\usepackage{amsfonts}       
\usepackage{nicefrac}       
\usepackage{microtype}      
\usepackage{lipsum}

\usepackage[small]{caption}
\usepackage{algorithm}
\usepackage[noend]{algpseudocode}
\usepackage{amssymb}
\usepackage[round]{natbib}
\usepackage{mathtools}
\usepackage{amsmath}
\usepackage{tikz}
\usetikzlibrary{calc}
\usepackage{todonotes}
\usepackage{dsfont}
\usepackage{enumitem}
\usepackage{tikz}
\usepackage{amsthm}
\usepackage{thm-restate}
\usepackage{booktabs}
\usepackage{multirow}
\usepackage{newfloat}
\usepackage{wrapfig}
\usepackage{dsfont}
\usepackage{comment}
\usepackage{multirow}
\usetikzlibrary{decorations.markings}

\newtheorem{theorem}{Theorem}

\newtheorem{corollary}{Corollary}
\newtheorem{proposition}{Proposition}
\newtheorem{definition}{Definition}

\newtheorem{assumption}{Assumption}

\newcommand{\NPHard}{$\mathsf{NP}$-hard}
\newcommand{\NP}{$\mathsf{NP}$}

\newcommand{\Poly}{$\mathsf{P}$}
\newcommand{\APX}{$\mathsf{APX}$}

\newcommand{\supp}{\textnormal{supp}}

\newcommand{\poly}{\textnormal{poly}}

\newcommand{\SUP}{\textsf{SUP}}

\newcommand{\defeq}{\vcentcolon=}

\newcommand{\pa}{\mathcal{P}}
\newcommand{\avec}{{\boldsymbol{a}}}

\DeclareMathOperator*{\argmax}{arg\,max}

\title{Designing Menus of Contracts Efficiently:\\ The Power of Randomization \thanks{In the previous version of the paper~\cite{CastiglioniEC22}, we incorrectly claim that the problem of finding an optimal menu of randomized contracts admits a maximum and that an optimal menu can be computed in polynomial time. In this version, we show that the problem does not admit a maximum but only a supremum, and we provide a polynomial-time algorithm that provides a solution whose value is arbitrarily close to that provided by the supremum.}}

\author{
	Matteo Castiglioni\\
	Politecnico di Milano\\
	\texttt{matteo.castiglioni@polimi.it}
	\And
	Alberto Marchesi\\
	Politecnico di Milano\\
	\texttt{alberto.marchesi@polimi.it}
	\And
	Nicola Gatti\\
	Politecnico di Milano\\
	\texttt{nicola.gatti@polimi.it}
}

\begin{document}


\maketitle

We study \emph{hidden-action principal-agent problems} in which a principal commits to an outcome-dependent payment scheme (called \emph{contract}) so as to incentivize the agent to take a costly, unobservable action leading to favorable outcomes.
In particular, we focus on \emph{Bayesian} settings where the agent has \emph{private information}.
This is collectively encoded by the agent's \emph{type}, which is unknown to the principal, but randomly drawn according to a finitely-supported, commonly-known probability distribution.
In our model, the agent's type determines both the probability distribution over outcomes and the cost associated with each agent's action.
In Bayesian principal-agent problems, the principal may be better off by committing to a \emph{menu of contracts} specifying a contract for each agent's type, rather than committing to a single contract.
This induces a two-stage process that resembles interactions studied in classical \emph{mechanism design}: after the principal has committed to a menu, the agent first \emph{reports a type} to the principal, and, then, the latter puts in place the contract in the menu that corresponds to the reported type.
Thus, the principal's computational problem boils down to designing a menu of contracts that incentivizes the agent to report their true type and maximizes expected utility.

Previous works showed that, in Bayesian principal-agent problems, computing an optimal menu of contracts or an optimal (single) contract is \APX-hard, which is in sharp contrast from what happens in non-Bayesian settings, where an optimal contract can be computed efficiently.
Crucially, previous works focus on menus of \emph{deterministic} contracts.
Surprisingly, in this paper we show that, if one instead considers menus of \emph{randomized} contracts defined as probability distributions over payment vectors, then an ``almost-optimal'' menu can be computed in polynomial time.
Indeed, the problem of computing a principal-optimal menu of \emph{randomized} contracts may \emph{not} admit a maximum, but only a supremum.
Nevertheless, we show how to design a polynomial-time algorithm that guarantees the principal with an expected utility arbitrarily close to the supremum.
Besides this main result, we also close several gaps in the computational complexity analysis of the problem of computing menus of deterministic contracts.
In particular, we prove that the problem cannot be approximated up to within any multiplicative factor and it does \emph{not} admit an additive FPTAS unless \Poly~$=$~\NP, even in basic instances with a constant number of actions and only four outcomes.
This considerably extends previously-known negative results.
Then, we show that our hardness result is tight, by providing an additive PTAS that works in instances with a constant number of outcomes.
We complete our analysis by showing that an optimal menu of deterministic contracts can be computed in polynomial time when either there are only two outcomes or there is a constant number of types.


\section{Introduction}

\emph{Principal-agent problems} have recently received a growing attention from the economics and computation community.
These problems model the interaction between a \emph{principal} and an \emph{agent}, where the latter chooses an action determining some externalities on the former.
In this work, we focus on \emph{hidden-action} problems---also known as models with \emph{moral hazard}---in which the principal cannot observe the action adopted by the agent, but only an outcome that is stochastically determined as an effect of such an action.
The agent incurs in a cost for performing the action, while the principal perceives a reward associated with the realized outcome.
Thus, the goal of the principal is to incentivize the agent to take an action resulting in favorable outcomes.
This is accomplished by the principal by committing to a \emph{contract}, which is an outcome-dependent payment scheme defining a payment from the principal to the agent for every possible outcome.

A classical example of principal-agent problem is that of a salesperson (agent) working for a company (principal).
The former has to decide on the level of effort (action) to put in selling some products on behalf of the company.
The latter can observe the total number of products sold by the salesperson (outcome), but it has no information about the actual level of effort undertaken by the salesperson.
This naturally fits hidden-action models.
Moreover, the company usually incentivizes an high level of effort by paying the salesperson on the basis of the number of products that have been actually sold, as it is the case in the classical definition of contract.
Besides this simple example, nowadays principal-agent problems are ubiquitous in digital economies, as they find application in several real-world scenarios, such as, \emph{e.g.}, crowdsourcing platforms~\citep{ho2016adaptive}, blockchain-based smart contracts~\citep{cong2019blockchain}, and healthcare~\citep{bastani2016analysis}.

Most of the computational works on principal-agent problems have focused on the basic setting in which the principal knows everything about the agent, \emph{i.e.}, they know both the probability distribution over outcomes and the cost associated with each agent's action.
Very recently, three concurrent works~\citep{guruganesh2020contracts,castiglioni2021bayesian,alon2021contracts} started the study of much more realistic \emph{Bayesian} principal-agent problems---also known as models with \emph{adversarial selection}---in which the agent has some \emph{private information} that is unknown to the principal.
As it is common in Bayesian models, these works assume that the agent's private information is collectively encoded by an unknown agent's \emph{type}, and that the latter is drawn from a probability distribution over a set of possible types, which is known to the principal.
For instance, in the company-salesperson example described above, the salesperson may have some private features (such as, \emph{e.g.}, experience gained with past works and/or advanced training courses) that determine how effectively the undertaken level of effort converts into sales.

The addition of private information in principal-agent problems establishes an intimate connection with \emph{mechanism design}.
While the latter has received a lot of attention on its own by computational economics---thanks to its widespread application in auction settings---, only the very recent works by~\citet{guruganesh2020contracts}~and~\citet{alon2021contracts} addressed the computational aspects of problems at the interface of the two fields.
Indeed, in Bayesian principal-agent problems, it may be the case that the principal is better off committing to a \emph{menu of contracts} specifying a contract for every possible agent's type, rather than committing to a single contract.
This induces a two-stage process that resembles interactions usually studied in mechanism design.
In particular, as a first stage after the principal's commitment, the agent \emph{reports a type} to the principal, possibly different from their true type.
Then, the interaction goes on as in a non-Bayesian principal-agent problem, with the principal selecting the contract in the menu that corresponds to the reported type.
The principal's goal is to commit to menus of contracts that incentivize the agent to report their true type, choosing an expected-utility-maximizing menu among them.
Notice that proposing menus of contracts is natural in many practical applications.
For instance, in the company-salesperson example, one may imagine the company proposing a portfolio of different payment regimes to the agent, with the latter selecting the preferred one based on their private information.
Intuitively, this could considerably boost revenues of the company with respect to proposing a single contract.

\subsection{Original Contributions}

In this paper, we investigate the computational complexity of finding an optimal menu of contracts for the principal, \emph{i.e.}, one maximizing the principal's expected utility among those that incentivize the agent to truthfully report their type.
In particular, we study general Bayesian settings where the agent's private information determines both the probability distributions and the costs of actions.

In Bayesian principal-agent settings, designing an optimal (single) contract is largely computationally intractable, with the exception of some specific cases~\citep{guruganesh2020contracts,castiglioni2021bayesian}.
\citet{guruganesh2020contracts} unsuccessfully tried to circumvent this issue with menus of contracts, showing that the problem of computing an optimal menu is \APX-hard even in instances with a constant number of actions.
These results are in sharp contrast with what happens in non-Bayesian settings, where an optimal contract can be designed efficiently~\citep{dutting2019simple}.
Crucially, the work by~\citet{guruganesh2020contracts} focuses on menus of \emph{deterministic} contracts, in which no randomization is involved.
The main result of our work is that, if one considers menus of \emph{randomized} contracts, then an ``almost-optimal'' one can indeed be computed in polynomial time in arbitrary Bayesian principal-agent problem instances.
Indeed, the problem of computing a principal-optimal menu of \emph{randomized} contracts may \emph{not} admit a maximum, but only a supremum.
Nevertheless, we show how to design a polynomial-time algorithm that guarantees the principal with an expected utility arbitrarily close to the supremum.
This is surprising, since randomized contracts generalize classical, deterministic ones by specifying probability distributions over payment vectors, so that, after the type-reporting stage, the distribution corresponding to the reported agent's type is employed by the principal to draw a contract that is communicated to the agent.

After introducing all the preliminary definitions that we need in Section~\ref{sec:preliminaries}, we start, in Section~\ref{sec:hardness}, by providing a strong negative result for the problem of computing an optimal menu of deterministic contracts.
This considerably generalizes the negative result by~\citet{guruganesh2020contracts}, as we prove that the problem cannot be approximated up to within any constant multiplicative factor and it does \emph{not} admit an additive FPTAS unless \Poly\ $=$ \NP, even in instances with a constant number of actions and only four outcomes.
We prove our result by resorting to a non-trivial and non-standard reduction from a suitably-defined promise problem related to finding maximal independent sets in undirected graphs with bounded degree.
Let us remark that our negative result is surprising, since an optimal (single) contract can be computed in polynomial time in Bayesian settings with a constant number of outcomes~\citep{guruganesh2020contracts,castiglioni2021bayesian}.

In Section~\ref{sec:ptas}, we close the gaps in the computational analysis of menus of deterministic contracts by showing that our hardness results are indeed tight.
In particular, we provide an additive PTAS for the problem of computing an optimal menu in instances with a constant number of outcomes.
Our approximation scheme works by finding an approximately-incentive-compatible menu of deterministic contracts (\emph{i.e.}, one that does \emph{not} perfectly incentivize the agent to report their true type, according to a suitable definition of approximation introduced for our purposes), which can be shown to provide a good additive approximation of the optimal principal's expected utility.
We prove that such an approximate menu of deterministic contracts can be found in polynomial time by restricting the attention to menus that only employ a ``small'' number of different contracts.
%
Finally, starting from the approximate menu, we show how to recover in polynomial time a menu of deterministic contracts that correctly incentivizes the agent to report their true type, only incurring in a small additional loss in terms of principal's expected utility.

Next, in Section~\ref{sec:simple}, we provide two additional positive results that complete our computational analysis.
In particular, we show that the problem of finding an optimal menu of deterministic contracts can be solved in polynomial time when either there are only two outcomes or there is a constant number of agent's types (and outcomes and actions can be an arbitrary number).

Finally, we conclude with Section~\ref{sec:randomized}, which provides our main result on menus of randomized contracts.
As a first step, we show that the problem of computing a principal-optimal menu of randomized contracts may \emph{not} admit a maximum, but only a supremum.
%
%
%
Then, we show how to design in polynomial time a menu of randomized contracts with principal's expected utility greater than or equal to the value of the supremum minus $\epsilon$, for any given $\epsilon<0$.
To do so, we first show that, for every $\epsilon>0$, there always exists a menu of randomized contracts that achieves principal's expected utility at most $\epsilon$ less than the supremum by using ``small'' payment values, which can be bounded above by $1/\epsilon$ and a suitably-defined exponential function of the instance size.
This is crucial to show that, in order to find the desired menu, we can restrict the attention to randomized contracts placing positive probability on a specific finite set of deterministic contracts, whose size is exponential in the instance size.
Given such a set, we can formulate the problem as a linear program with exponentially-many variables and polynomially-many constraints, whose dual can be solved in polynomial time by means of the \emph{ellipsoid algorithm} provided that a suitable separation oracle can be implemented in polynomial time.
Such an oracle can be formulated as an optimization problem over the finite set of deterministic contracts defined above, which we show that can be solved in time polynomial in the instance size and in $\log(1/\epsilon)$, proving our main result and concluding the paper.

All the proofs omitted from the main body of paper are in the Appendix.

\subsection{Related Works}

Hidden-action principal-agent problems have received considerable attention in the economic literature, as part of a broader subject called \emph{contract theory}~\citep{shavell1979risk,grossman1983analysis,rogerson1985repeated,holmstrom1991multitask} (see the books by~\citet{mas1995microeconomic},~\citet{bolton2005contract},~and~\citet{laffont2009theory} for a detailed treatment of the subject).

Interest on the computational aspects of contract theory have emerged only recently.
In the following, we survey the major computational works on hidden-action principal-agent problems.

\paragraph{Works on non-Bayesian Settings.}
Most of the computational works on principal-agent problems focus on non-Bayesian settings in which the principal knows everything about the agent.
Since in a classical non-Bayesian setting the principal's computational problem can be solved straightforwardly in polynomial time by means of linear programming, all the works on the topic introduced more complicated models.
\citet{babaioff2006combinatorial} study a model with multiple agents (see also its extended version~\citep{babaioff2012combinatorial} and its follow-ups~\citep{babaioff2009free,babaioff2010mixed}), focusing on how complex combinations of agents' actions influence the resulting outcome in presence of inter-agent externalities.
%
%
\citet{dutting2020complexity}~and~\citet{duetting2021combinatorial} study other non-Bayesian principal-agent models, whose underlying structure is combinatorial.
In particular, the former study the case in which the outcome space is defined implicitly through a suitably-defined succinct representation, while the latter address settings in which the agent can select a subset of actions (rather than a single one) out of a set of available actions.
Other works worth citing are~\citep{babaioff2014contract}, which studies contract complexity in terms of the number of different payments specified by contracts, and~\citep{ho2016adaptive}, which proposes an online learning model and solves it by means of multi-armed bandit techniques.
Another important line of work is that initiated by~\citet{dutting2019simple} and aimed at using the computational lens for the efficiency analysis (in terms of principal's expected utility) of \emph{linear contracts} with respect to general ones, where the former are simple, pure-commission contracts that pay the agent a given fraction of the principal's reward associated with the obtained outcome.
In particular,~\citet{dutting2019simple} show that, in non-Bayesian principal-agent settings, linear contracts perform well under reasonable assumptions.

\paragraph{Works on Bayesian Settings.}
There are three works that are arguably the most related to ours, namely~\citep{guruganesh2020contracts,alon2021contracts,castiglioni2021bayesian}.
These works concurrently introduced similar Bayesian principal-agent models, in order to study their computational properties.
In particular,~\citet{guruganesh2020contracts} focus on a model in which the unknown agent's type determines the probability distributions associated to agent's actions.
They analyze linear contracts by extending the work of~\citet{dutting2019simple} from non-Bayesian to Bayesian settings, showing how their efficiency is affected by problem parameters.
\citet{guruganesh2020contracts} also investigate the computational complexity of computing an optimal (single) contract and an optimal menu of (deterministic) contracts, showing that both problems are \APX-hard even in instances with a constant number of actions.
\citet{castiglioni2021bayesian} take a more computational-oriented approach than that of~\citep{guruganesh2020contracts}, by analyzing the efficiency of linear contracts in Bayesian settings with respect to the much more reasonable benchmark defined as the best among tractable contracts, \emph{i.e.}, those computable in polynomial time. 
Furthermore,~\citet{guruganesh2020contracts} only compare contracts in multiplicative terms, while~\citet{castiglioni2021bayesian} investigate bi-approximation (\emph{i.e.}, both multiplicative and additive) guarantees. 
Finally,~\citet{alon2021contracts} study a specific Bayesian principal-agent setting in which the agent's type is single-dimensional.
In particular, they show that, in their setting, an optimal menu of (deterministic) contracts can be computed in polynomial time when the number of actions is constant.
Moreover, in a following preprint~\citep{alon2021contractsarvix}, the same authors introduce menus of randomized contracts in their setting, showing an example in which randomization makes the principal better off by increasing their expected utility with respect to menus of deterministic contracts.

\section{Preliminaries}\label{sec:preliminaries}

In this section, we introduce all the elements needed in the rest of the paper.
Section~\ref{sec:preliminaries_problem} formally describes standard Bayesian principal-agent problems, wile Section~\ref{sec:preliminaries_contracts} introduces our setting in which the principal proposes to the agent a menu of randomized contracts to choose from.

\subsection{The Bayesian Principal-Agent Problem}\label{sec:preliminaries_problem}

An instance of the \emph{Bayesian principal-agent problem} is defined by a tuple $(\Theta,A,\Omega)$, where: $\Theta$ is a finite set of $\ell \coloneqq |\Theta|$ agent's types; $A$ is a finite set of $n \coloneqq |A|$ agents' actions; and $\Omega$ is a finite set of $m \coloneqq |\Omega|$ possible outcomes.\footnote{For the ease of presentation, we assume that all the agent's types share the same action set. All the results in this paper can be easily extended to the case in which each agent's type $\theta \in \Theta$ has their own action set $A_\theta$.}
The agent's type is drawn according to a fixed probability distribution known to the principal. 
We let $\mu \in \Delta_{\Theta}$ be such a distribution, with $\mu_\theta$ denoting the probability of type $\theta \in \Theta$ being selected.\footnote{Given a finite set $X$, we denote with $\Delta_X$ the set of all the probability distributions defined over $X$.}
For every type $\theta \in \Theta$, we denote by $F_{\theta, a} \in \Delta_\Omega$ the probability distribution over outcomes $ \Omega$ when an agent of type $\theta$ selects action $a \in A$, while $c_{\theta, a} \in [0,1]$ is the agent's cost for that action.\footnote{In the rest of this work, we assume that rewards and costs are in $[0,1]$. All the results in this paper can be easily generalized to the case of an arbitrary range of positive numbers, by applying a suitable normalization.}
For the ease of notation, we let $F_{\theta, a, \omega} $ be the probability that $F_{\theta, a}$ assigns to outcome $\omega \in \Omega$, so that $\sum_{\omega \in \Omega} F_{\theta, a, \omega}  =1$.
Each outcome $\omega \in \Omega$ has a reward $r_\omega \in [0,1]$ for the principal.
As a result, when an agent of type $\theta \in \Theta$ selects an action $a \in A$, then the principal achieves an expected reward of $ \sum_{\omega \in \Omega} F_{\theta, a, \omega} \, r_\omega$.

In the standard model, the principal commits to a contract maximizing their expected utility.
A \emph{contract} specifies payments from the principal to the agent, which are contingent on the actual outcome achieved with the agent's action.
We formally define a contract by a vector $p\in \mathbb{R}_+^m$, whose components $p_\omega$ represent payments associated to outcomes $\omega \in\Omega$.
The assumption that payments are non-negative (\emph{i.e.}, they can only be from the principal to the agent, and \emph{not} the other way around) is known as \emph{limited liability} and it is common in contract theory~\citep{carroll2015robustness}.
When an agent of type $\theta \in \Theta$ selects an action $a \in A$, then the expected payment to the agent is $\sum_{\omega \in \Omega} F_{\theta, a, \omega} \, p_\omega$, while their utility is $\sum_{\omega \in \Omega} F_{\theta, a, \omega} \, p_\omega- c_{\theta, a}$.
On the other hand, the principal's expected utility in that case is $\sum_{\omega \in \Omega} F_{\theta, a, \omega} \, r_\omega -\sum_{\omega \in \Omega} F_{\theta, a, \omega} \, p_\omega$.

Given a contract $p\in \mathbb{R}_+^m$, an agent of type $\theta \in \Theta$ plays a \emph{best response}, which is an action that is:
\begin{enumerate}
	\item \emph{incentive compatible} (IC), \emph{i.e.}, it maximizes their expected utility over actions in $A$; and
	\item \emph{individually rational} (IR), \emph{i.e.}, it has non-negative expected utility (if there is no IR action, then the agent abstains from playing so as to preserve the \emph{status quo}).
\end{enumerate}

In the rest of this work, we make the following w.l.o.g. common assumption guaranteeing that IR is always enforced~\citep{dutting2019simple}.
This allows us to focus on IC only.
Intuitively, the following assumption ensures that each agent's type has always an action providing them with a non-negative utility, thus ensuring IR of any IC action.
\begin{assumption}\label{ass:ir}
	There exists an action $a \in A$ such that $c_{\theta, a} = 0$ for all $\theta \in \Theta$.
\end{assumption}

Formally, we denote by $\mathcal{B}_p^\theta \coloneqq \argmax_{a \in A} \left\{ \sum_{\omega \in \Omega} F_{\theta, a, \omega } p_\omega - c_{\theta, a} \right\}$ the set of \emph{best responses} of an agent of type $\theta \in \Theta$ under a contract $p\in \mathbb{R}_+^m$, \emph{i.e.}, given Assumption~\ref{ass:ir}, the set of all the actions that are IC for an agent of type $\theta$ under contract $p$.
As it is common in the literature (see, \emph{e.g.},~\citep{dutting2019simple}), we assume that the agent breaks ties in favor of the principal, selecting a best response that maximizes the principal's expected utility.
In the following, we let $b^\theta: \mathbb{R}_+^m \to A$ be a function returning the best responses played by an agent of type $\theta \in \Theta$, where, for any contract $p\in \mathbb{R}_+^m$, we define $b^\theta(p) \in \argmax_{a \in \mathcal{B}_p^\theta} \left\{ \sum_{\omega \in \Omega} F_{\theta, a, \omega} r_\omega - \sum_{\omega \in \Omega} F_{\theta, a, \omega} \, p_\omega \right\}$.

\subsection{Menus of Randomized Contracts}\label{sec:preliminaries_contracts}

We study Bayesian principal-agent problems in which there is an additional \emph{type-reporting} stage in which the principal proposes to the agent a menu of randomized contracts to choose from.

A \emph{randomized contract} is a probability distribution $\gamma$ over $\mathbb{R}^m_+$, \emph{i.e.}, over the set of vectors $p \in \mathbb{R}^m_+$ representing all the possible contracts.
%
We use $p \sim \gamma$ to denote that the (random) contract $p$ is distributed according to $\gamma$, and write $\mathbb{E}_{p \sim \gamma} [\cdot]$ to indicate the expectation taken with respect to the randomness of $p$.
We denote by $\supp (\gamma)$ the support of $\gamma$.
When $\gamma$ has a finite support, \emph{i.e.}, $|\supp (\gamma)| < \infty$, we let $\gamma_p $ be the probability that $\gamma$ assigns to contract $p \in \mathbb{R}^m_+$.

A \emph{menu of randomized contracts} is defined by a tuple $\Gamma = \left( \gamma^\theta \right)_{\theta \in \Theta}$ specifying a probability distribution $\gamma^\theta$ over $\mathbb{R}^m_+$ for each agent's type $\theta \in \Theta$.

The interaction between the principal and an agent of type $\theta \in \Theta$ goes as follows:
\begin{itemize}
	\item[(i)] the principal publicly commits to a menu $\Gamma = \left( \gamma^\theta \right)_{\theta \in \Theta}$ of randomized contracts;
	\item[(ii)] the agent reports a type $\hat \theta \in \Theta$ to the principal, possibly different from the true type $\theta$;
	%
	\item[(iii)] the principal draws a contract $p \sim \gamma^{\hat \theta}$ and communicates it to the agent;
	\item[(iv)] the agent plays the best-response action $b^\theta(p)$.
\end{itemize}

The goal of the principal is to commit to a utility-maximizing menu of randomized contracts, selecting among those that are \emph{dominant-strategy incentive compatible} (DSIC).\footnote{Notice that, by a revelation-principle-style argument (see the book by~\citet{shoham2008multiagent} for some examples of these kind of arguments), focusing on DSIC menus of contracts is w.l.o.g. when looking for a principal-optimal menu.}
Formally, a menu $\Gamma = \left( \gamma^\theta \right)_{\theta \in \Theta}$ of randomized contracts is DSIC if the following holds:
\begin{align}\label{eq:dsic}
	\mathbb{E}_{p \sim \gamma^\theta} \left[  \sum_{\omega \in \Omega} F_{\theta, b^\theta(p), \omega} \, p_\omega - c_{\theta, b^\theta(p)} \right]  \geq \mathbb{E}_{p \sim \gamma^{\hat \theta}} \left[  \sum_{\omega \in \Omega} F_{\theta, b^\theta(p),\omega} \, p_\omega - c_{\theta, b^\theta(p)}  \right]  \nonumber \\
	 \forall \hat \theta \neq \theta \in \Theta.
\end{align}
Intuitively, the conditions above guarantee that the expected utility of an agent of type $\theta \in \Theta$ under the randomized contract $\gamma^\theta$ is greater than or equal to that obtained under $\gamma^{\hat \theta}$, for any $\hat \theta \neq \theta \in \Theta$.
This ensures that the agent is always better off reporting their true type to the principal.
Then, the principal's goal is to find a menu $\Gamma = \left( \gamma^\theta \right)_{\theta \in \Theta}$ that is optimal for the following problem:
\begin{align}\label{eq:problem}
	\max_{\Gamma = \left( \gamma^\theta \right)_{\theta \in \Theta}} & \quad \sum_{\theta \in \Theta} \mu_\theta \, \mathbb{E}_{p \sim \gamma^\theta} \left[  \sum_{\omega \in \Omega} F_{\theta, b^\theta(p), \omega} \, r_\omega -\sum_{\omega \in \Omega} F_{\theta, b^\theta(p), \omega} \, p_\omega \right] \quad \textnormal{s.t.} \\
	 & \textnormal{Equation~\eqref{eq:dsic}}, \nonumber
\end{align}
whose objective is the principal's expected utility for a DSIC menu of randomized contracts.

We also consider the case, already investigated by~\citet{guruganesh2020contracts}, in which the menu of contracts is made by non-randomized contracts.
Formally, we denote a \emph{menu of deterministic contracts} by a tuple $P = \left(  p^\theta \right)_{\theta \in \Theta}$, containing a contract $p^\theta \in \mathbb{R}^m_+$ for each agent's type $\theta \in \Theta$.
In this case, the principal does \emph{not} draw a contract in step (iii) of the interaction, but they simply adopt contract $p^{\hat \theta}$, where $\hat \theta \in \Theta$ is the type reported by the agent.
Thus, all the definitions above can be specialized to the case of menus of deterministic contracts by dropping expectations.

\section{A Striking Negative Result for Menus of Deterministic Contracts}\label{sec:hardness}

We begin by providing a negative result for the principal's computational problem of finding an optimal menu of deterministic contracts.
%
%
In particular, we prove that it is \NPHard\ to provide any constant multiplicative approximation to the optimal principal's expected utility, even if both the number of outcomes and that of agent's actions are constants.
This is in contrast to what happens with the problem of finding an optimal (single) contract in Bayesian instances, which can be solved in polynomial time when the number of outcomes is constant~\citep{guruganesh2020contracts,castiglioni2021bayesian}.
Moreover, our reduction also shows that the problem does \emph{not} admit an additive FPTAS, unless \Poly\ $=$ \NP.
Let us remark that our reduction considerably strengthens an already-known hardness result for menus of deterministic contracts, which is that by~\citet{guruganesh2020contracts}, who show that the problem is \APX-hard even with a constant number of actions.
Indeed, our result shows that the problem is \emph{not} in \APX, and that this holds even when both the number of outcomes and that of actions are constants.

Our reduction is from the following \emph{promise problem}, which is related to the $\textsf{INDEPENDENT-SET}$ problem on undirected graphs with bounded-degree nodes. 
\begin{definition}[$\textsf{GAP-BOUNDED-IS}_{\alpha,k}$]
	For every $\alpha \in [0,1]$ and $k \in \mathbb{N_+}$, we define \emph{$\textsf{GAP-BOUNDED-IS}_{\alpha,k}$} as the following promise problem:
	\begin{itemize}
		\item \textnormal{\texttt{Input:}} An undirected graph $G=(V,E)$ in which each vertex has degree at most $k$ and a constant $\eta \in \left[ \frac{1}{k},1 \right]$ such that either one of the following is true:\footnote{For $\eta<\frac{1}{k}$, the problem can be trivially solved since a maximal independent set has size at least $\frac{1}{k}$.}
		\begin{itemize}
			\item there exists an independent set (i.e., a subset of vertices such that there is no edge connecting two of them) of size at least $\eta|V|$;
			\item all the independent sets have size at most $\alpha \eta|V|$.
		\end{itemize}
		\item \textnormal{\texttt{Output:}} Determine which of the above two cases hold.
	\end{itemize}
\end{definition}

Notice that, for every constant $\alpha>0$, there exists a constant $k=k(\alpha)$ depending on $\alpha$ such that $\textsf{GAP-BOUNDED-IS}_{\alpha,k}$ is \NPHard~\citep{Alon1995IndependentSet,Trevisan2001Independent}.
This observation is exploited in order to prove the following theorem:

\begin{theorem}\label{thm:inaprox}
	In Bayesian principal-agent problems, the following two results hold even when there are only four outcomes and a constant number of actions.
	\begin{itemize}
		\item For every constant $\alpha>0$, it is \NPHard\ to approximate the principal's expected utility in an optimal DSIC menu of deterministic contracts up to within an $\alpha$ multiplicative factor.
		\item There is no additive \emph{FPTAS} for the problem of finding an DSIC optimal menu of deterministic contracts, unless \Poly\ $=$ \NP.
		\end{itemize}
\end{theorem}

\begin{proof}
	We reduce from $\textsf{GAP-BOUNDED-IS}_{\alpha,k}$.
	In the rest of the proof, given an undirected graph $G = (V,E)$, for the ease of presentation we label vertices with natural numbers, so that we can write $V=\{1,\dots, s\}$, where $s \coloneqq |V|$.
	Given an instance of $\textsf{GAP-BOUNDED-IS}_{\alpha,k}$, letting $l \coloneqq \left\lceil\frac{k}{\alpha}\right\rceil$ and $\rho \coloneqq s^{-3}$, we build an instance of Bayesian principal-agent problem such that:
	\begin{itemize}
		\item \emph{Completeness}: If there is an independent set of $G$ with size at least $\eta s$, than there is a DSIC menu of deterministic contracts with principal's expected utility at least $\frac{1}{2} \eta \rho l 2^{-l}$.
		\item \emph{Soundness}: If all the independent sets of $G$ have size at most $\alpha \eta s$, then the principal's expected utility in any DSIC menu of deterministic contracts is at most $ 4 \alpha \eta  \rho l 2^{-l}$
	\end{itemize}
	Since, for any $\alpha > 0$, there exists $k = k(\alpha)$ such that $\textsf{GAP-BOUNDED-IS}_{\alpha,k}$ is \NPHard, the two points above immediately provide our hardness result for multiplicative approximations. 
	Moreover, by noticing that the difference between principal's expected utilities in the completeness and the soundness cases depends polynomially on the instance size, we also have that there is no additive FPTAS for the problem of finding an optimal menu of deterministic contracts, unless \Poly\ $=$ \NP.

	\paragraph{Construction.}
	%
	%
	There are four outcomes, namely $\Omega = \{  \omega_{1},\omega_{2},\omega_{3},\omega_{4} \}$, with $r_{\omega_{3}}=1$ and all the other rewards equal to zero.
	There is an agent's type $\theta_v\in \Theta$ for each vertex $v \in V$ of $G$, so that $\ell = s$.
	All types are equally probable, and, thus, $\mu_{\theta_v}=\frac{1}{s}$ for every $\theta_v \in \Theta$.
	Each agent's type $\theta_v \in \Theta$ has an action $a_v \in A_{\theta_v}$, and an additional action $a_{ui} \in A_{\theta_v}$ for every vertex $u \in V$ adjacent to $v$, \emph{i.e.}, $(v,u) \in E$, and index $i \in [l-3]$, where $[x] \coloneqq \{1,\dots,x\}$.\footnote{For the ease of presentation, in the instances that we define in the reduction proving Theorem~\ref{thm:inaprox} each agent's type $\theta \in \Theta$ has a different set of actions $A_\theta$. It is easy to modify such instances by adding dummy actions so as to recover ``equivalent'' instances in which each type has the same set of actions available.}
	For every $\theta_v \in \Theta$, we let	
	\begin{itemize}
	\item $F_{\theta_v,a_v,\omega_1}=\frac{1}{4}\cos\left(\pi \frac{v}{2s}\right)$,
	\item $F_{\theta_v,a_v,\omega_2}=\frac{1}{4}\sin\left(\pi \frac{v}{2s}\right)$, 
	\item $F_{\theta_v,a_v,\omega_3}=\frac{1}{4}$, and
	\item $F_{\theta_v,a_v,\omega_4}=1-F_{\theta_v,a_v,\omega_1}-F_{\theta_v,a_v,\omega_2}-F_{\theta_v,a_v,\omega_3}$.
	\end{itemize}
	Moreover, for every pair of adjacent vertices $v \in V, u \in V$ such that $(v,u)\in E$ and index $i \in [l-3]$, an agent of type $\theta_v$ has the following probabilities associated to action $a_{ui}$:
	\begin{itemize}
		\item $F_{\theta_v,a_{ui},\omega_1}=\cos\left(\pi \frac{u}{2s}\right)2^{-i-2}$, 
		\item $F_{\theta_v,a_{ui},\omega_2}=\sin\left(\pi \frac{u}{2s}\right)2^{-i-2}$,
		\item $F_{\theta_v,a_{ui},\omega_3}=2^{-i-2}$, and
		\item $F_{\theta_v,a_{ui},\omega_4}=1-F_{\theta_v,a_{ui},\omega_1}-F_{\theta_v,a_{ui},\omega_2}-F_{\theta_v,a_{ui},\omega_3}$.
	\end{itemize}
	Finally, each agent's type $\theta_v \in \Theta$ has an additional action $\bar a \in A_{\theta_v}$ such that $F_{\theta_v,\bar a,\omega_4}=1$ and $F_{\theta_v,\bar a,\omega}=0$ for all $\omega \neq \omega_4 \in \Omega$.
	Notice that $|A_{\theta_v}| = O( k^2 / \alpha )$ for every $\theta_v \in \Theta$.
	In order to conclude the construction, we need to define action costs.
	In particular, for every agent's type $\theta_v \in \Theta$, we set $c_{\theta_v,a_v}=\frac{1}{4}-\rho l2^{-l}$ and $c_{\theta_v,a_{ui}}=2^{-i-2}-\rho (l-i)2^{-l}$ for every vertex $u \in V$ such that $(v,u) \in E$ and index $i \in [l-3]$.
	Furthermore, the cost of action $\bar a$ is $c_{\theta_v,\bar a}=0$ for every agent's type $\theta_v \in \Theta$.

	Let us remark that we use the $\sin$ and $\cos$ trigonometric functions in order to easily represent the vectors $F_{\theta_v, a}$ in the $\omega_1$--$\omega_2 $ Cartesian coordinate system (see Figure~\ref{fig:circles}).
	One of the main properties of our construction is that, for a type $\theta_v$, no vector  representing an action $a_v$ or $a_{ui}$ is dominated by a convex combination of vectors representing actions with the same costs.
	This is necessary, since an action represented by a dominated vector is never preferred over another action represented by an un-dominated vector, with their costs being the same.
	Notice that we use the functions $\sin$ and $\cos$ for the ease of exposition.
	It is easy to change our instances so that our results continue to hold even if we replace such functions with sufficiently good approximations.\footnote{Indeed, it is sufficient to use approximations that can be represented in memory with a number of bits that is polynomial in the instance size, and, thus, that can be computed in polynomial time.}

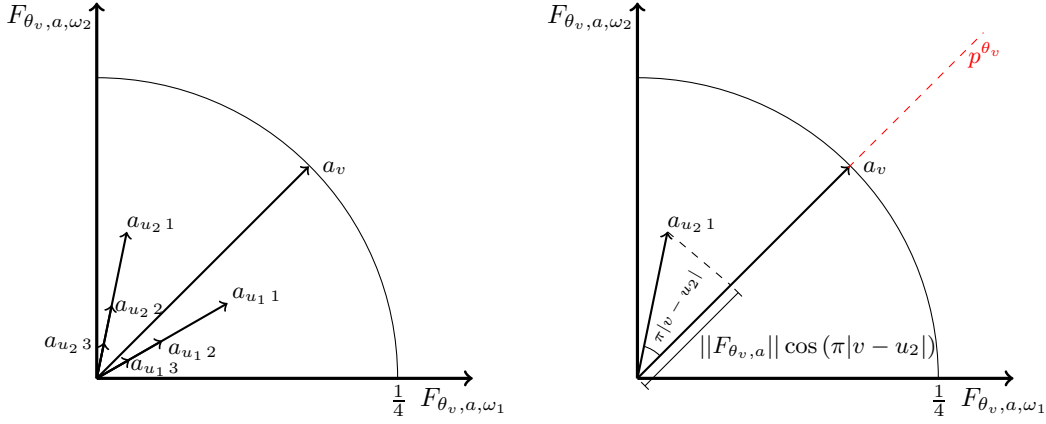
\begin{figure}
	\begin{minipage}{0.43\textwidth}
		\centering
		\begin{tikzpicture}[decoration={markings,
			mark=at position 0.5cm with {\arrow[line width=1pt]{>}},
			mark=at position 2cm with {\arrow[line width=1pt]{>}},
			mark=at position 7.85cm with {\arrow[line width=1pt]{>}},
			mark=at position 9cm with {\arrow[line width=1pt]{>}}}]
			\draw[->, very thick] (0,0) -- (5,0) coordinate (xaxis);
			\draw[->, very thick] (0,0) -- (0,5) coordinate (yaxis);
			\node[text width=0cm] at (4.3,-0.3) {$F_{\theta_v,a,\omega_1}$};
			\node[text width=0cm] at (-1.2,4.8) {$F_{\theta_v,a,\omega_2}$};
			\draw (4,0) arc[start angle=0, end angle=90, radius=4];
			\draw[->,thick] (0,0) -- (2.82,2.82) coordinate (xaxis);
			\node[text width=0cm] at (3.0,2.8) {\footnotesize$a_{v}$};
			\draw[->, thick] (0,0) -- (1.74,1) coordinate (xaxis);
			\node[text width=0cm] at (1.8,1.1) {\footnotesize$a_{u_1 \,1}$};
			\draw[ ->,thick] (0,0) -- (0.4,1.95) coordinate (xaxis);
			\node[text width=0cm] at (0.4,2.05) {\footnotesize $a_{u_2 \, 1}$};
			\draw[->, thick] (0,0) -- (0.87,0.5) coordinate (xaxis);
			\node[text width=0cm] at (0.95,0.35) {\footnotesize$a_{u_1 \,2}$};
			\draw[ ->,thick] (0,0) -- (0.2,0.975) coordinate (xaxis);
			\node[text width=0cm] at (0.24,0.925) {\footnotesize $a_{u_2 \, 2}$};
			\draw[->, thick] (0,0) -- (0.44,0.25) coordinate (xaxis);
			\node[text width=0cm] at (0.45,0.14) {\footnotesize$a_{u_1 \,3}$};
			\draw[ ->,thick] (0,0) -- (0.1,0.48) coordinate (xaxis);
			\node[text width=0cm] at (-0.7,0.45) {\footnotesize $a_{u_2 \, 3}$};
			\node[text width=0cm] at (3.9,-0.3) {$\frac{1}{4}$};
		\end{tikzpicture}
	\end{minipage}
	\begin{minipage}{0.43\textwidth}
		\centering
		\begin{tikzpicture}[decoration={markings,
			mark=at position 0.5cm with {\arrow[line width=1pt]{>}},
			mark=at position 2cm with {\arrow[line width=1pt]{>}},
			mark=at position 7.85cm with {\arrow[line width=1pt]{>}},
			mark=at position 9cm with {\arrow[line width=1pt]{>}}}]
			\draw[->, very thick] (0,0) -- (5,0) coordinate (xaxis);
			\draw[->, very thick] (0,0) -- (0,5) coordinate (yaxis);
			\node[text width=0cm] at (4.3,-0.3) {$F_{\theta_v,a,\omega_1}$};
			\node[text width=0cm] at (-1.2,4.8) {$F_{\theta_v,a,\omega_2}$};
			\draw (4,0) arc[start angle=0, end angle=90, radius=4];
			\draw[->,thick] (0,0) -- (2.82,2.82) coordinate (xaxis);
			\node[text width=0cm] at (3,2.8) {\footnotesize$a_{v}$};
			\draw[->,thick] (0,0) -- (0.4,1.95) coordinate (xaxis);
			\node[text width=0cm] at (0.4,2.05) {\footnotesize $a_{u_2 \, 1}$};
			\draw[color=red, dashed ] (2.8,2.8) -- (4.6,4.6) coordinate (xaxis);
			\node[text width=0cm] at (4.4,4.3) {\textcolor{red}{\footnotesize$p^{\theta_ v}$}};
			\draw (0.3,0.3) arc[start angle=45, end angle=75, radius=0.5];
			\node[text width=2cm,rotate=60] at (0.75,1.35) {\tiny$\pi|v-u_2|$};
			\draw[dashed] (0.4,1.95) -- (1.24,1.24) coordinate (xaxis);
			\node[text width=0cm] at (3.9,-0.3) {$\frac{1}{4}$};
			\draw[] (0.1,-0.1) -- (1.34,1.14) coordinate (xaxis);
			\draw[] (0.05,-0.05) -- (0.15,-0.15) coordinate (xaxis);
			\draw[] (1.29,1.19) -- (1.39,1.09) coordinate (xaxis);
			\node[text width=0cm] at (0.82,0.42) {\footnotesize $||F_{\theta_v,a}|| \cos \left( \pi |v-u_2|\right)$};
		\end{tikzpicture}
	\end{minipage}
	\caption{\textbf{Left:} Graphical representation of the probability vectors $F_{\theta_v, a}$ used in the proof of Theorem~\ref{thm:inaprox} in the $\omega_1$--$\omega_2 $ Cartesian coordinate system. In particular, node $v \in V$ has $2$ adjacent nodes $u_1, u_2 \in V$, and $l=6$. \textbf{Right:} Expected payments of contract $p^{\theta_v}$ for a vertex $v \in V^*$ in the completeness part of the proof of Theorem~\ref{thm:inaprox}.
	The dashed red line represents the direction of the vector encoding $p^{\theta_v}$. The expected payments of actions $a_v$ and $a_{u_2\,1}$ are given by the projections of their probability vectors on $p^{\theta_v}$, times $ ||p^{\theta_v}||$.}
	\label{fig:circles}
\end{figure}

	\paragraph{Completeness.}
	Suppose that there exists an independent set $V^\ast \subseteq V$ of $G$ with size at least $\eta s$.
	From that, we build a DSIC menu of deterministic contracts $P = \left( p^{\theta_v} \right)_{\theta_v \in \Theta}$ with principal's expected utility at least $\frac{1}{2} \eta \rho l 2^{-l}$.
	%
	For every agent's type $\theta_v \in \Theta$ such that $v \in V^*$, we define the contract $p^{\theta_v} \in \mathbb{R}_+^m$ so that $p^{\theta_v}_{\omega_1}=\cos\left(\pi\frac{v}{2s}\right) \left(1-\rho l2^{-l+1}\right)$, $p^{\theta_v}_{\omega_2}=\sin\left(\pi\frac{v}{2s}\right)\left(1- \rho l2^{-l+1}\right)$, and $p^{\theta_v}_{\omega_3} = p^{\theta_v}_{\omega_4} = 0$.
	Moreover, for every agent's type $\theta_v \in \Theta$ such that $v \notin V^*$, we define $p^{\theta_v} = p^{\theta_u}$ for some $u \in V^*$ such that \[\theta_u \in  \argmax_{\theta_u \in \Theta:u \in V^*} \left\{  \max_{a \in A_{\theta_v}} \sum_{\omega \in \Omega} F_{\theta_v,a,\omega} \, p^{\theta_u}_\omega-c_{\theta_v,a} \right\},\] where the term inside the $\argmax$ operator is the expected utility that a utility-maximizing action provides to an agent of type $\theta_v$ under a contract $p^{\theta_u}$ already defined above, since $u \in V^*$.
	Intuitively, contract $p^{\theta_v}$ represents a vector on the $\omega_{1}$--$\omega_{2}$ plane. An action $a_{ui}$ has a large expected payment if the angle between vector $p^{\theta_v}$ and the projection of vector $F_{\theta_v,a_{ui}}$ on the $\omega_{1}$--$\omega_{2}$ plane is small. In particular, the expected payment is $||p^{\theta_v}||_2 ||(F_{\theta_v,a_{ui},\omega_1}, F_{\theta_v,a_{ui},\omega_2})||_2 \cos(\beta)$, where $\beta$ is the angle between the two vectors.
	%
	
	Next, we show that the menu of contracts defined above is DSIC.
	First, we prove that an agent of type $\theta_v \in \Theta$ with $v \in V^*$ gets an expected utility of $\rho \frac{l}{2}2^{-l}$ under contract $p^{\theta_v}$, since the best response of the agent under contract $p^{\theta_v}$ is $a_v$, namely $a_v=b^{\theta_v}\left( p^{\theta_v} \right)$.
	Indeed, the agent's expected utility by selecting action $a_v$ under contract $p^{\theta_v}$ is
	\[
		\frac{1}{4} \cos^2 \hspace{-0.1cm} \left( \pi\frac{v}{2s} \right) \hspace{-0.1cm} \left(1-\rho l2^{-l+1}\right)+ \frac{1}{4} \sin^2 \hspace{-0.1cm} \left( \pi\frac{v}{2s} \right) \hspace{-0.1cm} \left( 1-\rho l2^{-l+1} \right)-\left(\frac{1}{4}-\rho l2^{-l} \right)=\rho \frac{l}{2}2^{-l},
	\]
	where we used $\cos^2 x+ \sin^2 x=1$.
	All the other actions $a_{ui}$, for any vertex $u \in V$ such that $(v,u) \in E$ and index $i \in [l-3]$, provide the agent with an expected utility of
	\begingroup
	\allowdisplaybreaks
		\begin{align*}
			\cos\left( \pi\frac{u}{2s} \right) & 2^{-i-2} \cos\left( \pi\frac{v}{2s} \right) \left(1-\rho l2^{-l+1}\right) \\
			& \hspace{-1cm}+\sin\left(\pi\frac{u}{2s}\right) 2^{-i-2}  \sin\left(\pi\frac{v}{2s}\right) \left(1-\rho l2^{-l+1}\right) - \left( 2^{-i-2}-\rho(l-i)2^{-l} \right)\\
			 & \hspace{-1cm} = \cos\left( \pi \frac{|v-u|}{2s} \right) 2^{-i-2} \left(1-\rho l2^{-l+1}\right)-\left(2^{-i-2}-\rho (l-i)2^{-l}\right) \\
			 & \hspace{-1cm} \le \left( 1-\frac{1}{\pi}\left( \pi \frac{|v-u|}{2s} \right)^2 \right) \left( 1-\rho l2^{-l+1} \right) 2^{-i-2} - \left( 2^{-i-2}-\rho (l-i)2^{-l} \right)\\
			 &\hspace{-1cm} \le \left( 1-\left(  \frac{|v-u|}{2s} \right)^2 \right) \left( 1-\rho l2^{-l+1} \right) 2^{-i-2} - \left( 2^{-i-2}-\rho (l-i)2^{-l} \right)\\
			 & \hspace{-1cm} \le 2^{-i-2}- \rho l 2^{-i-l-1} -  \frac{1}{4s^2} 2^{-i-2} + \frac{1}{4s^2} \rho l 2^{-i-l-1} -\left( 2^{-i-2}-\rho(l-i)2^{-l} \right)\\
			 & \hspace{-1.0cm} =- \rho l2^{-i-l-1} -  \frac{1}{4s^2} 2^{-i-2} + \frac{1}{4s^2} \rho l2^{-i-l-1}+\rho(l-i)2^{-l} \\
			&\hspace{-1.0cm}  =2^{-l} \left(- \rho l2^{-i-1} -  \frac{1}{4s^2} 2^{l-i-2} + \frac{1}{4s^2} \rho l2^{-i-1}+\rho(l-i) \right),
	\end{align*}
\endgroup
	where in the first equality we used $\sin x \sin y + \cos x \cos y\ = \cos|x-y|$, in the first inequality the fact that $\cos x\le 1-\frac{x^2}{\pi}$ for $x\in \left[ 0,\frac{\pi}{2} \right]$, and in the third inequality $|v-u|\ge 1$.
	For $s$ large enough, it holds that $\frac{1}{4s^2} 2^{l-i-2} \ge  \frac{1}{4s^2} \rho l2^{-i-1}+\rho(l-i)$ (recall that $\rho=s^{-3}$ and $l$ is a constant depending only on $k$ and $\alpha$), and, thus, the agent's expected utility by playing any action $a_{ui}$ is smaller than or equal to zero, proving that $a_v=b^{\theta_v}\left( p^{\theta_v} \right)$.
	%
	
	In order to prove that the menu is DSIC, it remains to show that each agent's type $\theta_v \in \Theta$ is \emph{not} better off by reporting a type different from their true type.
	Let us consider any contract $p^{\theta_{u}}$ with $\theta_u \in \Theta$ such that $u \in V^*$ (for contracts $p^{\theta_{u}}$ with $\theta_u \in \Theta$ such that $u \notin V^*$ the reasoning is analogous given how they are defined).
	Since $V^*$ is an independent set of $G$, an agent of type $\theta_v \in \Theta$ has no action $a_{ui}$ available, while, by selecting any action $a_{u'i}$ corresponding to a node $u'\neq u \in V$ and an index $i \in [l-3]$, under contract $p^{\theta_u}$ their expected utility is 
		\begin{align*}
		\cos\left( \pi\frac{u'}{2s} \right) & 2^{-i-2}  \cos\left( \pi\frac{u}{2s} \right) \left( 1-\rho l2^{-l+1} \right)\\
		&+\sin\left( \pi\frac{u'}{2s} \right) 2^{-i-2}  \sin\left( \pi\frac{u}{2s} \right) \left( 1-\rho l2^{-l+1} \right) 
		- \left( 2^{-i-2}-\rho (l-i)2^{-l} \right),
	  \end{align*}
   which is less than or equal to zero by arguments similar to those employed above.
   Finally, by playing action $a_v$ under contract $p^{\theta_{u}}$, the expected utility of an agent of type $\theta_v$ is 
   \begingroup
   \allowdisplaybreaks
		\begin{align*}
			\frac{1}{4}\cos\left( \pi\frac{v}{2s} \right) & \cos\left( \pi\frac{u}{2s} \right) \left(1-\rho l2^{-l+1}\right)\\
			&+ \frac{1}{4} \sin \left( \pi\frac{v}{2s} \right) \sin \left( \pi\frac{u}{2s} \right) \left( 1-\rho l2^{-l+1} \right)-\left( \frac{1}{4}-\rho l2^{-l} \right) \\
			&  = \frac{1}{4} \cos\left( \pi \frac{|v-u|}{2s} \right)  \left( 1-\rho l2^{-l+1} \right) -\left( \frac{1}{4}-\rho l2^{-l} \right) \\
			& \le \frac{1}{4}  \left( 1-\frac{1}{\pi}\left( \pi\frac{|v-u|}{2s}\right)^2 \right) \left(1-\rho l2^{-l+1} \right)-\left(\frac{1}{4} -\rho l2^{-l}\right)\\
			& \leq 2^{-l} \left(- \frac{l}{2} \rho -\frac{1}{16s^2}2^l+\frac{1}{8s^2} \rho l + \rho l \right),
		\end{align*}
	\endgroup
	which, as it is easy to check, is smaller than or equal to zero for $s$ large enough (recall that $\rho=s^{-3}$ and $l$ is a constant depending only on $k$ and $\alpha$).

	Now, we show that the principal's expected utility provided by the menu of deterministic contracts defined above is as desired.
	We do that by first proving that the contribution to the principal's expected utility due to each agent's type $\theta_v \in \Theta$ such that $v \in V^*$ is equal to $\frac{1}{2}\rho l 2^{-l}$.
	Since each agent reports their true type $\theta_v$ and plays action $a_v$, the principal's expected utility is
	\begin{align*}
		F_{\theta_v,a_v, \omega_3} r_{\omega_3} & -  F_{\theta_v,a_v, \omega_1} \, p^{\theta_v}_{\omega_1} - F_{\theta_v,a_v, \omega_2} \, p^{\theta_v}_{\omega_2}\\
		& =\frac{1}{4}-\frac{1}{4} \cos^2\left( \pi\frac{v}{2s} \right) \left( 1-\rho l2^{-l+1} \right)-\frac{1}{4}\sin^2 \left( \pi\frac{v}{2s} \right) \left( 1-\rho l2^{-l+1} \right)\\
		& =\frac{1}{4} -\frac{1}{4} \left( 1-\rho l2^{-l+1} \right) =\frac{1}{4} \rho l2^{-l+1}= \frac{1}{2}\rho l 2^{-l}.
	\end{align*}
	%
	Moreover, we show that the principal's expected utility resulting from each agent's type $\theta_v \in \Theta$ such that $v \notin V^*$ is at least zero.
	For each such a type $\theta_v$, let $u \in V^*$ be a vertex such that $p^{\theta_v}=p^{\theta_u}$.
	We need to consider four different cases.
	First, suppose that $b^{\theta_v}\left( p^{\theta_v} \right)=a_v$.
	Then, the principal's expected utility is equal to
	%
		\begin{align*}
			&\frac{1}{4}- \frac{1}{4} \cos\left( \pi\frac{v}{2s}\right) \cos \left( \pi\frac{u}{2s}\right) \left( 1-\rho l2^{-l+1} \right)\\
			&\hspace{1.8cm}-\frac{1}{4} \sin\left( \pi\frac{v}{2s} \right) \sin\left( \pi\frac{u}{2s} \right)  \left( 1-\rho l2^{-l+1} \right)  \ge \frac{1}{4}- \frac{1}{4}\left( 1-\rho l2^{-l+1} \right)\ge 0,
		\end{align*}
	where in the inequality we use $\sin x \sin y + \cos x \cos y\ \le 1 $.
	
	Second, assume that $b^{\theta_v}\left( p^{\theta_v} \right)=a_{ui}$ for some index $i \in [l-3]$.
	Then, the principal's utility is 
		\begin{align*}
			2^{i-2}-  \cos^2\left( \pi\frac{u}{2s} \right)& 2^{-i-2} \left( 1-\rho l2^{-l+1} \right)- \sin^2 \left( \pi\frac{u}{2s} \right) 2^{-i-2} \left( 1-\rho l2^{-l+1} \right) \\ & \ge2^{i-2}-\cos^2 \left( \pi\frac{u}{2s} \right) 2^{-i-2}- \sin^2 \left( \pi\frac{u}{2s} \right) 2^{-i-2} = 0.
	   \end{align*}
	%
	Third, suppose that $b^{\theta_v}\left( p^{\theta_v} \right)=a_{u'i}$ for some vertex $u' \neq u \in V$ and index $i \in [l-3]$, then the principal's expected utility is
		\begin{align*}
			2^{i-2}-\cos\left(\pi\frac{u}{2s}\right) 2^{-i-2} & \cos\left(\pi\frac{u'}{2s}\right) \left(1-\rho l2^{-l+1}\right)\\
			& - \sin\left( \pi\frac{u}{2s} \right) 2^{-i-2}  \sin\left(\pi\frac{u'}{2s}\right) \left(1-\rho l2^{-l+1}\right) \\
			&\hspace{-2cm} \ge2^{i-2} - \cos\left( \pi\frac{u}{2s} \right) 2^{-i-2} \cos\left( \pi\frac{u'}{2s} \right)\\
			&-\sin \left( \pi\frac{u}{2s} \right) 2^{-i-2}  \sin\left( \pi\frac{u'}{2s} \right)\\
			&\hspace{-2cm}  \ge 0.
		\end{align*}
	%
	Finally, whenever $b^{\theta_v}\left( p^{\theta_v} \right)=\bar a$, expected rewards and payments are zero.
	Hence, we can conclude that the principal's expected utility when the agent's type is $\theta_v \in \Theta$ with $v \notin V^*$ is at least zero.
	Thus, the overall principal's expected utility is at least 
	\[
		\sum_{\theta_{v} \in \Theta : v \in V^* } \mu_{\theta_v} \left(  \sum_{\omega \in \Omega} F_{\theta_v, b^{\theta_v}(p^{\theta_v}), \omega} \, r_\omega -\sum_{\omega \in \Omega} F_{\theta_v, b^{\theta_v}(p^{\theta_v}), \omega} \, p^{\theta_v}_\omega \right) \ge  \frac{1}{2} \eta \rho l2^{-l}.
	\]

	\paragraph{Soundness.}
	We prove that, if all the independent sets of $G$ have size at most $\alpha\eta s$, then the principal's expected utility is at most $ 4 \alpha \eta l \rho 2^{-l}$ in any DSIC menu of deterministic contracts $P = \left( p^\theta \right)_{\theta \in \Theta}$.
	First, we show that, when the agent plays an action $a_{ui}$, then the principal's expected utility is at most $3 \rho 2^{-l}$.
	If an agent of type $\theta_v \in \Theta$ is incentivized to play an action $a_{ui} \in A_{\theta_v}$ with $i=l-3$ by contract $p^{\theta_v}$, \emph{i.e.}, it holds $b^{\theta_v} \left( p^{\theta_v} \right) = a_{ui}$, then it is easy to see that the principal's expected utility is at most $F_{\theta_v,a_{ui},\omega_3} r_{\omega_3}- c_{\theta_v,a_{ui}}= 3 \rho 2^{-l}$.
	Moreover, suppose that the agent plays an action $a_{ui}$ with $i \in [l-4]$.
	Since such an action is incentivized by contract $p^{\theta_v}$ over action $a_{u \, i+1}$, it is easy to check that by IC conditions it must be the case that
	\begingroup
	\allowdisplaybreaks
		\begin{align*}
		p^{\theta_v}_{\omega_1} \sin\left( \pi\frac{u}{2s} \right) & 2^{-i-2}+p^{\theta_v}_{\omega_2} \cos\left( \pi\frac{u}{2s} \right) 2^{-i-2}  -\left( 2^{-i-2}- \rho (l-i)2^{-l} \right)\\
		& \hspace{-1.3cm}\geq  p^{\theta_v}_{\omega_1} \sin\left( \pi\frac{u}{2s} \right) 2^{-i-3}+p^{\theta_v}_{\omega_2} \cos\left( \pi\frac{u}{2s} \right) 2^{-i-3}  -\left( 2^{-i-3}- \rho (l-i-1)2^{-l} \right),
		\end{align*}
	\endgroup
	which implies that
	\begin{align*}
		&\left( p^{\theta_v}_{\omega_1} \sin\left(\pi\frac{u}{2s} \right) +p^{\theta_v}_{\omega_2} \cos\left( \pi\frac{u}{2s} \right) \right)  \left( 2^{-i-2}-2^{-i-3} \right)\\
		& \hspace{4cm}\ge \left( 2^{-i-2}-\rho (l-i)2^{-l} \right) - \left( 2^{-i-3}- \rho(l-i-1)2^{-l} \right).
	\end{align*}
	Hence, 
	\[
		p^{\theta_v}_{\omega_1} \sin\left( \pi\frac{u}{2s} \right) +p^{\theta_v}_{\omega_2} \cos\left( \pi\frac{u}{2s} \right)   \ge \left( 2^{-i-3}- \rho2^{-l} \right) 2^{i+3},
	\]
	 and the agent's expected payment is at least 
	 	\begin{align*}
			p^{\theta_v}_{\omega_1} \sin\left( \pi\frac{u}{2s} \right) 2^{-i-2}+p^{\theta_v}_{\omega_2} \cos\left( \pi\frac{u}{2s} \right) 2^{-i-2}&\ge 2^{-i-2} \left( 2^{-i-3}-\rho 2^{-l} \right)2^{i+3}\\
			&\ge 2^{-i-2}-\rho2^{-l+1} .
	   \end{align*}
 	%
	This implies that the principal's expected utility is at most
	\begin{align*}
		F_{\theta_v,a_{ui},\omega_3} r_{\omega_3}-  F_{\theta_v,a_{ui},\omega_3} \, p^{\theta_v}_{\omega_1} -  F_{\theta_v,a_{ui},\omega_3} \, p^{\theta_v}_{\omega_2} &\le 2^{-i-2} - \left( 2^{-i-2}-\rho2^{-l+1} \right) \\
		&\le \rho2^{-l+1} \\
		&\leq 3 \rho 2^{-l}.
	\end{align*}

	Next, we prove that, given two agent's types $\theta_v \in \Theta$ and $\theta_u \in \Theta$ such that $(v,u) \in E$, then the principal's expected utility due to an agent of type $\theta_v$ or $\theta_u$ is at most $3\rho2^{-l}$.
	As we have shown, when an agent of type $\theta_v$ plays an action different from $a_v$, the principal's expected utility is at most $3\rho2^{-l}$. Hence, the principal's expected utility is strictly greater than $3\rho2^{-l}$ for both the agent's types $\theta_v \in \Theta$ and $\theta_u \in \Theta$ only if they play $a_v$ and $a_u$, respectively. However, in the following we show that if $(v,u) \in E$ and the two agent's type play the actions $a_v$ and $a_u$, then the principal's expected utility is at most $3\rho2^{-l}$ for at least one of the two types.
	Suppose that $p^{\theta_v}$ incentivizes action $a_v$ and that $p^{\theta_{u}}$ incentivizes action $a_u$, namely $b^{\theta_v} \left( p^{\theta_v} \right) = a_{v}$ and $b^{\theta_u} \left( p^{\theta_u} \right) = a_{u}$.
	Then, since an agent of type $\theta_v$ prefers action $a_v$ over action $a_{ui}$ with $i=1$, by IC we have
		\begin{align*}
			p^{\theta_v}_{\omega_1} \frac{1}{4}\sin\left( \pi\frac{v}{2s} \right) & +p^{\theta_v}_{\omega_2} \frac{1}{4} \cos\left( \pi\frac{v}{2s}\right) - \left( \frac{1}{4}-\rho l2^{-l} \right) \\ 
			& \ge p^{\theta_u}_{\omega_1} \frac{1}{8}\sin\left( \pi\frac{u}{2s} \right) +p^{\theta_u}_{\omega_2} \frac{1}{8} \cos\left( \pi\frac{u}{2s}\right)  -\left( \frac{1}{8}-\rho(l-1)2^{-l} \right).
		 \end{align*}
	Similarly, since an agent of type $\theta_u$ prefers action $a_u$ over action $a_{vi}$ with $i=1$, we have
		\begin{align*}
				\frac{1}{4}p^{\theta_u}_{\omega_1} \sin\left( \pi\frac{u}{2s} \right) & +p^{\theta_u}_{\omega_2} \frac{1}{4} \cos\left( \pi\frac{u}{2s} \right) -\left( \frac{1}{4}-\rho l2^{-l}\right)  \\
				& \ge p^{\theta_v}_{\omega_1} \frac{1}{8} \sin\left( \pi\frac{v}{2s} \right) +p^{\theta_v}_{\omega_2} \frac{1}{8} \cos\left( \pi\frac{v}{2s} \right) - \left( \frac{1}{8}-\rho (l-1)2^{-l} \right).
		\end{align*}
	%
	Now, assume that
	\[
		p^{\theta_v}_{\omega_1} \sin\left( \pi\frac{v}{2s} \right) +p^{\theta_v}_{\omega_2} \cos\left( \pi\frac{v}{2s} \right)\ge p^{\theta_u}_{\omega_1} \sin\left( \pi\frac{u}{2s}\right) +p^{\theta_u}_{\omega_2} \cos\left( \pi\frac{u}{2s}\right).
	\]
	The other case is analogous.
	Then, by taking the difference between the two inequalities above, we get
	\begin{align*}
		p^{\theta_u}_{\omega_1} \frac{1}{8}\sin\left( \pi\frac{u}{2s} \right) +p^{\theta_u}_{\omega_2} \frac{1}{8} \cos\left( \pi\frac{u}{2s} \right)   &\ge  \left( \frac{1}{4}-\rho l2^{-l} \right)- \left( \frac{1}{8}-\rho (l-1)2^{-l} \right)\\
		&=\frac{1}{8}- \rho2^{-l}.
	\end{align*}
	Then, the expected payment for an agent of type $\theta_u$ is 
	\[
		p^{\theta_u}_{\omega_1} \frac{1}{4} \sin\left( \pi\frac{u}{2s} \right) +p^{\theta_u}_{\omega_2} \frac{1}{4} \cos\left( \pi\frac{u}{2s} \right) \ge \frac{1}{4}-  \rho 2^{-l+1},
	\]
	 and the principal's expected utility due to agent's type $\theta_u$ is at most $\rho 2^{-l+1} $.

	 As a last step, notice that, for every possible agent's type $\theta_v$ and menu of deterministic contracts $P = \left( p^\theta\right)_{\theta \in \Theta}$, the principal's expected utility is at most $\rho l 2^{-l}$. In particular, we have shown that, by incentivizing an action different from  $a_v$, the principal obtains at most $ 3\rho 2^{-l} $ expected utility. Moreover, by incentivizing action $a_v$, the principal gets at most $F_{\theta_v,a_v,\omega_3} r_{\omega_3}-c_{\theta_v,a_v}=\rho l 2^{-l}$.
	 
 	To conclude the proof, for every possible menu of deterministic contracts $P = \left( p^\theta \right)_{\theta \in \Theta}$, the set of agent's types $\theta_v \in \Theta$ that result in a principal's expected utility strictly greater than $3 \rho2^{-l}$ defines an independent set $V^* \subseteq V$ of $G$.
 	Moreover, for all the agent's types, the principal's expected utility is at most $\rho l 2^{-l}$.
	Hence, the overall principal's expected utility is:
		\begin{align*}
				&\sum_{\theta_v \in \Theta} \mu_{\theta_v} \left(  \sum_{\omega \in \Omega} F_{\theta, b^{\theta_v}(p^{\theta_v}), \omega} \, r_\omega -\sum_{\omega \in \Omega} F_{\theta, b^{\theta_v}(p^{\theta_v}), \omega} \, p^{\theta_v}_\omega  \right) \\
				& \hspace{4cm}\le\frac{1}{s} \left( \sum_{\theta_v \in \Theta: v \in V^*} \rho l 2^{-l} + \sum_{\theta_v \in \Theta: v \notin  V^*} 3\rho 2^{-l} \right)\\
			&  \hspace{4cm} \le \alpha \eta \rho l 2^{-l} + 3 \rho 2^{-l}   \\
			&  \hspace{4cm} \le 4 \alpha \eta \rho l 2^{-l},
		\end{align*}
	where the last inequality follows from $\alpha \eta l \ge \alpha \frac{1}{k} \frac{k}{\alpha}=1$ (recall that $\eta\ge \frac{1}{k}$).
\end{proof}

\section{Menus of Deterministic Contracts: An Additive PTAS with a Constant Number of Outcomes}\label{sec:ptas}

In the previous section, we showed that, even in Bayesian principal-agent instances with a constant number of outcomes (and actions), the problem of computing a utility-maximizing DSIC menu of randomized contracts does \emph{not} admit any multiplicative approximation that can be computed in polynomial time.
In this section, we study the domain of additive approximations, where we provide a PTAS for the problem that works in settings with a constant number of outcomes.
This result is tight, since by Theorem~\ref{thm:inaprox} the problem does \emph{not} admit an additive FPTAS unless \Poly\ $=$ \NP.

In order to define our PTAS, we first need to introduce some auxiliary results.
In the following, we will rely on an abstract definition of \emph{approximately-incentive-compatible} menu of deterministic contracts, which combines two levels of approximation, one related to the IC conditions on agent's actions, and the other to the DSIC constraints that the menu should satisfy.
In particular, approximate menus do \emph{not} only specify a contract for each agent's type, but also a tuple of action recommendations, one for each type, so that the action recommended to each agent's type is approximately IC under the contract corresponding to that type.
Formally:
\begin{definition}[$\epsilon$-Approximate Menu of Deterministic Contracts]\label{def:apx_menu}
	Given any Bayesian principal-agent instance $(\Theta, A, \Omega)$ and $\epsilon > 0$, an \emph{$\epsilon$-approximate menu of deterministic contracts} is a tuple $\left( p^\theta, a^\theta \right)_{\theta \in \Theta}$ of contract-action pairs, with $p^\theta \in \mathbb{R}_{+}^m$ and $a^\theta \in A$ for all $\theta \in \Theta$, such that:
	\begin{align}\label{eq:apxIc}
		&\sum_{\omega \in \Omega} F_{\theta, a^\theta, \omega} \, p^\theta_\omega - c_{\theta, a^\theta} \geq \sum_{\omega \in \Omega} F_{\theta, b^\theta(p^{\hat \theta}), \omega} \, p^{\hat \theta}_\omega - c_{\theta, b^\theta(p^{\hat \theta})} -\epsilon  & \forall \theta, \hat \theta \in \Theta.
	\end{align}
\end{definition}
Intuitively, an $\epsilon$-approximate menu of deterministic contracts defined by $\left( p^\theta, a^\theta \right)_{\theta \in \Theta}$ satisfies the following two conditions: each action $a^\theta$ is $\epsilon$-approximately IC for an agent of type $\theta \in \Theta$ under contract $p^\theta$ (see the case $\theta = \hat \theta$ in Equation~\eqref{eq:apxIc}), and, additionally, the menu is $\epsilon$-approximately DSIC assuming type $\theta$ plays $a^\theta$ when reporting their true type (see cases $\theta\neq\hat \theta$ in Equation~\eqref{eq:apxIc}).
In the following, when we refer to the principal's expected utility in an $\epsilon$-approximate menu of deterministic contracts $\left( p^\theta, a^\theta \right)_{\theta \in \Theta}$, we mean the expected utility that the principal gets if each agent's type $\theta \in \Theta$ truthfully reports their type and plays $a^\theta$ under contract $p^\theta$.
Formally, the principal's expected utility in $\left( p^\theta, a^\theta \right)_{\theta \in \Theta}$ can be written as follows:
\[
	\sum_{\theta \in \Theta} \mu_\theta \left(  \sum_{\omega \in \Omega} F_{\theta, a^\theta, \omega} \, p^\theta_\omega - c_{\theta, a^\theta} \right).
\]
%

Now, we show that, starting from an $\epsilon$-approximate menu of deterministic contracts with principal's expected utility equal to $APX$, we can recover in polynomial time a (non-approximate) DSIC menu of deterministic contracts providing the principal with an expected utility at least of $APX-2\sqrt{\epsilon}$.
This allows us to focus on the computation of $\epsilon$-approximate menu of deterministic contracts.
Formally, we prove the following lemma.
\begin{restatable}{lemma}{lemmaOne}\label{lm:toIC}
	Given a Bayesian principal-agent instance $(\Theta, A, \Omega)$ and $\epsilon > 0$, let $\left( p^\theta, a^\theta \right)_{\theta \in \Theta}$ be an $\epsilon$-approximate menu of deterministic contracts with principal's expected utility $APX$.
	Then, there exists a DSIC menu of deterministic contracts $\overline{P} = \left( \bar p^\theta \right)_{\theta \in \Theta}$ in which the principal's expected utility is at least 
	$APX-2\sqrt{\epsilon}$.
	Moreover, such a menu can be computed in polynomial time.
\end{restatable}

Next, we show that, in an optimal menu of deterministic contracts, large expected payments are only assigned to agent's types occurring with small probability, as it will be useful in Lemma~\ref{lm:smallSup}.
Given a constant $L\ge 1$ and a menu of deterministic contracts $P=\left( p^\theta \right)_{\theta \in \Theta}$, we let $\Theta(L,P) \subseteq \Theta$ be the set of agent's types such that $\sum_{\omega \in \Omega} F_{\theta,b^\theta(p^\theta),\omega} \, p^\theta_\omega\ge  L$.
Formally, we prove the following:
\begin{restatable}{lemma}{lemmaTwo}\label{lm:smallCosts}
	Given a Bayesian principal-agent instance $(\Theta, A, \Omega)$, let $P = \left( p^\theta \right)_{\theta \in \Theta}$ be an optimal menu of deterministic contracts.
	Then, for every constant $L>1$, it holds $\sum_{\theta \in \Theta(L,P)} \mu_\theta \le \frac{1}{L}$.
\end{restatable}

The next step is to show that there exists a menu of deterministic contracts that employs a small set of \emph{different} contracts, while providing almost optimal principal's expected utility.
This allows us to represent a menu of deterministic contracts in a more succinct way.
In particular, a menu that uses $k \in \mathbb{N}_+$ different contracts can be represented by a matrix $T \in \mathcal{T}\coloneqq \mathbb{R}_+^{k \times m}$ and a function $f^T :\Theta\rightarrow \{1,\dots,k\}$.
The matrix $T \in \mathcal{T}$ defines the $k$ different contracts used by the menu on its rows, with $T_i \in\mathbb{R}_+^{m}$ (\emph{i.e.}, the $i$-th row of $T$) denoting the vector encoding the $i$-th contract and $T_{i,\omega}$ (\emph{i.e.}, the element indexed by $i,\omega$) specifying the payment that the $i$-th contract associates to outcome $\omega \in \Omega$.
Additionally, the function $f^T:\Theta\rightarrow \{1,\dots,k\}$ assigns contracts to agent's types, with $f^T(\theta) \in \{1,\ldots,k\}$ defining the index of the contract proposed by menu $T \in \mathcal{T}$ when the agent reports type $\theta \in \Theta$ to the principal.

As a preliminary step, we prove the following crucial lemma.
\begin{restatable}{lemma}{lemmaThree} \label{lm:smallSup}
	Given a Bayesian principal-agent instance $(\Theta,A,\Omega)$ and a constant $\delta>0$, there exists a DSIC menu of deterministic contracts $P = \left( p^\theta \right)_{\theta \in \Theta}$ that employs at most $O\left( \left( \frac{m}{\delta^3} \log \frac{m}{\delta} \right)^m \right)$ different contracts and provides the principal with an expected utility at least of $OPT-\delta$, where $OPT$ is the principal's expected utility in an optimal DSIC menu of deterministic contracts.
\end{restatable}


The following theorem provides our PTAS.
By Lemma \ref{lm:smallSup}, in order to find an optimal DSIC menu of deterministic contracts, we can optimize over the set of menus of contracts represented by the matrices in  $\mathcal{T} =  \mathbb{R}_+^{k\times m}$, for a values of $k$ that only depends on the desired constant approximation $\delta>0$ and $m$.
As a first step, we show  that $\mathcal{T}$ can be partitioned into polyhedra such that the principal's expected utility is linear in each of them, and, thus, an optimal solution is on a vertex of one of these polyhedra.
Moreover, each polyhedron is defined by a subset of a common set of linear inequalities.
Since the dimension of the space $\mathcal{T}$ is $km$ and, thus, it is constant, each vertex is defined by the intersection of $km$ linearly independent hyperplanes. Finally, all the possible subsets of $km$ inequalities can be enumerated in polynomial time since $k$ and $m$ are constants.


\begin{restatable}{theorem}{theoremTwo}\label{thm:PTAS}
	There exists an additive PTAS for the problem of computing an optimal DSIC menu of deterministic contracts in Bayesian principal-agent instances with a constant number of outcomes.
\end{restatable}


Notice that our result cannot be extended to settings with an arbitrary number of outcomes.
Indeed, \citet{guruganesh2020contracts} show that, in Bayesian principal-agent instances with an arbitrary number of outcomes, there exists an $\epsilon>0$ such that it is \NP-hard to approximate the principal's expected utility in an optimal menu of deterministic contracts up to within an $\epsilon$ additive factor, unless \Poly\ $=$ \NP.

\section{Menus of Deterministic Contracts: Two Easy Cases} \label{sec:simple}

In this section, we analyze two cases in which an optimal menu of deterministic contracts can be computed in polynomial time. 
In particular, Section~\ref{subsec:two_out} studies instances with two outcomes, complementing the hardness result for the setting with four outcomes, while Section~\ref{subsec:const_types} addresses the case with a constant number of types.

\subsection{Case with Two Outcomes}\label{subsec:two_out}

Previously, we provide a PTAS for the problem of designing an optimal DSIC menu of deterministic contracts in Bayesian principal-agent instances with a constant number of outcomes.
With four or more outcomes, this result is tight by Theorem~\ref{thm:inaprox}.
In this section, we show that, with only two outcomes the problem can be solved in polynomial time.
We leave as an open problem the analysis of the case with three outcomes.
Our result generalizes that of~\citet{guruganesh2020contracts}.
In particular, they prove a result that is analogous to ours, but for the simpler case in which the action costs are the same for all the agent's types.

Our main result shows that, using menus of deterministic contracts does \emph{not} increase the principal's expected utility with respect to using a single contract.
\begin{restatable}{lemma}{lemmaBinary}
	In Bayesian principal-agent instances $(\Theta,A,\Omega)$ with $|\Omega| = 2$, there exists a contract having the same principal's expected utility as an optimal DSIC menu of deterministic contracts.
\end{restatable}

By the previous lemma, in order to compute an optimal DSIC menu of deterministic contracts, it is sufficient to compute an optimal single contract.
As shown by \citet{castiglioni2021bayesian} and \citet{guruganesh2020contracts}, an optimal single contract can be computed in polynomial time when the number of outcomes is constant. Hence, we obtain the following:

\begin{theorem}
	In Bayesian principal-agent instances $(\Theta,A,\Omega)$ with $|\Omega| = 2$, there exists a polynomial-time algorithm that computes an optimal DSIC menu of deterministic contracts.
\end{theorem}

\subsection{Case with a Constant Number of Types}\label{subsec:const_types}

As shown in Theorem~\ref{thm:inaprox}, when there is an arbitrary number of agent's types, the principal's problem cannot be approximated efficiently, even with a constant number of outcomes and actions.
We complement this result by showing that, when the number of types is constant, the problem can be solved in polynomial time.
Formally:

\begin{restatable}{theorem}{theoremTypes}
	In Bayesian principal-agent instances with a constant number of agent's types, there exists a polynomial-time algorithm that computes an optimal DSIC menu of deterministic contracts.
\end{restatable}

\section{How to Find an ``Almost-optimal'' Menu of Randomized Contracts Efficiently} \label{sec:randomized}

In this section, we show that adding randomization enables the polynomial-time computation of a DSIC menu of contracts providing the principal with an expected utility arbitrarily close to the best possible one.

As we show later in this section, the impossibility of designing an optimal menu of randomized contracts does \emph{not} stem from computational challenges associated with the problem, but it is rather due to the fact that an optimal menu of randomized contracts may \emph{not} exist.
Indeed, in general the problem of finding a DSIC menu of randomized contracts maximizing the principal's expected utility only admits a supremum, and it may \emph{not} admit a maximum. 
%
%
The supremum of the problem of finding a principal-utility-maximizing DSIC menus of randomized contracts can be computed by solving the following problem:
%
%
\begin{subequations}\label{prob:opt_rand}
	\begin{align}
		\sup_{\Gamma = \{ \gamma^\theta \}_{\theta \in \Theta} } & \,\, \sum_{\theta \in \Theta} \mu_\theta \mathbb{E}_{p \sim \gamma^\theta} \left[  \sum_{\omega \in \Omega} F_{\theta, b^\theta(p), \omega} r_\omega - \sum_{\omega \in \Omega} F_{\theta, b^\theta(p), \omega} p_\omega  \right] \quad \text{s.t.} \\
		& \hspace{-1cm} \mathbb{E}_{p \sim \gamma^\theta} \left[ \sum_{\omega \in \Omega} F_{\theta, b^\theta(p), \omega} p_\omega - c_{\theta, b^\theta(p)}  \right] \geq \mathbb{E}_{p \sim \gamma^{\hat \theta}} \left[ \sum_{\omega \in \Omega} F_{\theta, b^\theta(p), \omega} p_\omega - c_{\theta, b^\theta(p)}  \right] \nonumber\\
		 &\hspace{7.5cm} \forall \theta \neq \hat \theta \in \Theta. \label{eq:opt_rand_ic} 
	\end{align}
\end{subequations}

As a first step, we show that the supremum defined in Problem~\ref{prob:opt_rand} always assumes a finite value (\emph{i.e.}, it is never $+\infty$).
In the following, for ease of notation, we denote by \SUP \ the value of Problem~\ref{prob:opt_rand}.
Then, we can formally prove:
\begin{proposition}
	Problem~\ref{prob:opt_rand} has always value $\emph{\SUP} \in [0,1]$.
\end{proposition}

\begin{proof}
	It is sufficient to observe that the maximum principal's expected utility in a menu is $1$, since the rewards are in $[0,1]$ and the payments can only provide a negative utility contribution to the principal.
	Let $p^* \in \mathbb{R}^m_+$ be a contract such that $p^*_\omega=0$ for every $\omega \in \Omega$.
	Then, it is easy to check that $\SUP \ge 0$, since the menu of randomized contracts that set $\gamma^\theta_{p^*}=1$ for every $\theta \in \theta$ provides the principal with an expected utility of at least $0$.
	This proves the statement.
\end{proof}

%
%

Next, we show that Problem~\ref{prob:opt_rand} may \emph{not} admit a maximum, \emph{i.e.}, there is no contract for which the (finite) value of the supremum is attained.

Before proving such a result, we introduce some additional notation.
In particular, we let $\pa^{\theta,a} \coloneqq \left\{ p \in \mathbb{R}^m_+ \mid a \in \mathcal{B}_p^\theta  \right\} $ be the set of contracts such that action $a \in A$ is IC for an agent of type $\theta \in \Theta$.
Furthermore, we let $\hat\pa^{\theta, a} \coloneqq \left\{ p \in \mathbb{R}^m_+ \mid b^\theta(p) = a  \right\}$ be the set of contracts in which an agent of type $\theta $ plays action $a$.
Notice that the sets $\pa^{\theta,a}$ are closed polyhedra (as they can be defined by a system of linear inequalities), while sets $\hat\pa^{\theta,a}$ are nor open nor closed polyhedra (as they can be defined by a system of linear inequalities, some of which are strict due to the tie-breaking rule).
We also need to prove a preliminary result (Lemma~\ref{lem:finite_support}), which intuitively states that one can restrict the attention to randomized contracts placing positive probability on a finite number of contracts.
In particular, we prove that, for each type $\theta \in \Theta$, it is sufficient that the support of $\gamma^\theta$ contains at most one contract for each agent's action $a \in A$, with the latter being the action played by an agent of type $\theta$ in such a contract.
Formally:
\begin{restatable}{lemma}{lemmaRandOne}\label{lem:finite_support}
	Given any DSIC menu of randomized contract $\Gamma = \{ \gamma^\theta \}_{\theta \in \Theta}$, there always exists a DSIC menu of randomized contracts $\bar \Gamma = \{ \bar \gamma^\theta \}_{\theta \in \Theta}$ that provides the principal with at least the same expected utility as $\Gamma = \{ \gamma^\theta \}_{\theta \in \Theta}$ and such that, for every $\theta \in \Theta$, it holds $\left| \supp (\gamma^\theta) \cap \hat\pa^{\theta, a} \right| \leq 1$ for all $a \in A$.
\end{restatable}

Equipped with the result in Lemma~\ref{lem:finite_support}, we can now show that Problem~\ref{prob:opt_rand} does \emph{not} admit a maximum in general.
Formally:
\begin{theorem}
	There exist Bayesian principal-agent problem instances for which Problem~\ref{prob:opt_rand} does \emph{not} admit a maximum.
\end{theorem}

\begin{proof}
	Let us consider an instance defined as follows. There are three possible agent's types, namely $\Theta = \{ \theta_1$, $\theta_2$, $\theta_3\}$, with $\mu_{\theta_1}=\mu_{\theta_2}=\mu_{\theta_3}=\frac{1}{3}$.
	The agent has three actions available, namely $A = \{ a_1$, $a_2, a_3 \}$, while the set of possible outcomes is $\Omega = \{ \omega_1, \omega_2, \omega_3, \omega_4\}$.
	Type $\theta_1$ is such that $F_{\theta_1,a_1,\omega_1}=1$, $c_{\theta_1,a_1}=0$, $F_{\theta_1,a_2,\omega_3}=1$, $c_{\theta_1,a_2}=0$, $F_{\theta_1,a_3,\omega_3}=1$, and $c_{\theta_1,a_3}=0$.
	Type $\theta_2$ is such that $F_{\theta_2,a_1,\omega_1}=1$, $c_{\theta_2,a_1}=0$, $F_{\theta_2,a_2,\omega_2}=1$, $c_{\theta_2,a_2}=0$, $F_{\theta_2,a_3,\omega_4}=1$, and $c_{\theta_2,a_3}=0$.
	Type $\theta_3$ is such that $F_{\theta_3,a_1,\omega_2}=1$, $c_{\theta_3,a_1}=\frac{1}{4}$, $F_{\theta_3,a_2,\omega_3}=1$, $c_{\theta_3,a_2}=0$, $F_{\theta_3,a_3,\omega_3}=1$, and $c_{\theta_3,a_3}=0$.
	Finally, the principal's rewards are $r_{\omega_1}=1$, $r_{\omega_2}=\frac{3}{4}$, and $r_{\omega_3}=r_{\omega_4}=0$.
	
	As a first step, we show that, for every $\epsilon>0$, there exists a DSIC menu of randomized contracts with principal's expected utility at least $\frac{3}{4}-\epsilon$.
	Let $p^1 \in \mathbb{R}_+^m$ be a contract such that $p^1_\omega=0$ for all $\omega \in \Omega$, $p^2\in \mathbb{R}_+^m$ be such that $p^2_{\omega_4}=\frac{1}{12\epsilon}$ and $p^2_\omega=0$ for all $\omega \neq \omega_4$, and $p^3\in \mathbb{R}_+^m$ be such that $p^3_{\omega_2}=\frac{1}{4}$ and  $p^3_\omega=0$ for all $\omega \neq \omega_2$.
	Let us consider the menu $\Gamma=\{\gamma^\theta\}_{\theta \in \Theta}$ defined so that $\gamma^{\theta_1}_{p^1}=1$, $\gamma^{\theta_2}_{p^2}=3\epsilon$,  $\gamma^{\theta_2}_{p_1}=1-3\epsilon$, and $\gamma^{\theta_3}_{p_3}=1$.
	It Is easy to check that such a menu is DSIC.
	Moreover, the principal's expected utility is at least of
	\begin{align*}
		\frac{1}{3} \left[ \gamma^{\theta_1}_{p^1}+\gamma^{\theta_2}_{p^1} -\frac{1}{12\epsilon} \gamma^{\theta_2}_{p^2}+ \left( \frac{3}{4}-\frac{1}{4} \right) \gamma^{\theta_3}_{p^3} \right]&=\frac{1}{3}\left[ 1+1-3\epsilon-\frac{1}{4}+\frac{3}{4}-\frac{1}{4} \right]\\
		&=\frac{3}{4}- \epsilon .
	\end{align*}
	Hence, \SUP\ is at least $\frac{3}{4}$ for the considered instance.
	
	We conclude the proof by showing that any DSIC menu $\bar \Gamma=\{\bar \gamma^\theta\}_{\theta \in \Theta}$ results in a principal's expected utility that is strictly smaller than $\frac{3}{4}$.
	By contradiction, we assume that there exists a DSIC menu $\bar \Gamma=\{\bar \gamma^\theta\}_{\theta \in \Theta}$ of randomized contracts that provides the principal with expected utility greater than or equal to $\frac{3}{4}$.
	
	First, by Lemma~\ref{lem:finite_support}, we can focus on menus $ \bar \Gamma=\{ \bar \gamma^\theta\}_{\theta \in \Theta}$ such that for every $\theta \in \Theta$, it holds $\left| \supp (\bar \gamma^\theta) \cap \hat\pa^{\theta, a} \right| \leq 1$ for all $a \in A$.
	If $\left| \supp (\bar \gamma^{\theta_3}) \cap \hat\pa^{\theta_3, a_1} \right| = 1$, let $p^4 \in \mathbb{R}_+^m$ be the only contract in $\hat\pa^{\theta_3,a_1}$ such that $\bar \gamma^{\theta_3}_{p^4}>0$.
	Notice that $p^4_{\omega_2}\ge \frac{1}{4}$ by definition of $\hat\pa^{\theta_3, a_1}$.
	Instead, if $\left| \supp ( \bar\gamma^{\theta_3}) \cap \hat\pa^{\theta_3, a_1} \right| = 0$, let $p^4 \in \mathbb{R}_+^m$ be such that $p^4_\omega=0$ for all $\omega \in \Omega$.
	By the DSIC conditions for type $\theta_2$, we have that type $\theta_2$ gets an expected payment at least of $\frac{1}{4} \bar \gamma^{\theta_3}_{p^4} $.
	Moreover, if $\left| \supp (\bar \gamma^{\theta_2}) \cap \hat\pa^{\theta_2, a_1} \right| = 1$, let $p^5 \in \mathbb{R}_+^m$ be the only contract in $\hat\pa^{\theta_2,a_1}$ such that $\bar \gamma^{\theta_2}_{p^5}>0$.
	Otherwise, let $p^5 \in \mathbb{R}_+^m$ bu such that $p^5_\omega=0$ for all $\omega \in \Omega$.
	Finally, If $\left| \supp ( \bar \gamma^{\theta_2}) \cap \hat\pa^{\theta_2, a_2} \right| = 1$, let $p^6 \in \mathbb{R}_+^m$ be the only contract in $\hat\pa^{\theta_2,a_2}$ such that $\bar \gamma^{\theta_2}_{p^6}>0$; otherwise, let $p^6 \in \mathbb{R}_+^m$ be such that $p^6_\omega=0$ for every $\omega \in \Omega$.
	
	Now, suppose that action $a_1$ is incentivized for type $\theta_2$ with probability strictly smaller than $1$, \emph{i.e.}, $\bar \gamma^{\theta_2}_{p^5}<1$.
	Then, the overall principal's expected utility is at most
	\begin{align*}
		\frac{1}{3} \left[ 1+ \bar \gamma^{\theta_2}_{p^5}+\frac{3}{4}\bar \gamma^{\theta_2}_{p^6} -\frac{1}{4} \bar \gamma^{\theta_3}_{p^4}   +\left( \frac{3}{4}-\frac{1}{4} \right)  \bar \gamma^{\theta_3}_{p^4} \right] &=\frac{1}{3}\left[ 1+\bar \gamma^{\theta_2}_{p^5}+\frac{3}{4}\bar \gamma^{\theta_2}_{p^6} +\frac{1}{4}  \bar \gamma^{\theta_3}_{p^4} \right]\\
		& <\frac{1}{3}\, \frac{9}{4}= \frac{3}{4},
	\end{align*}
	which contradicts our initial assumption on $\bar \Gamma=\{ \bar \gamma^\theta\}_{\theta \in \Theta}$.
	
	As a result, it must hold $\bar \gamma^{\theta_2}_{p^5}=1$, and, by the DSIC conditions for type $\theta_2$, it must be the case that $p^5_{\omega_1}\ge \frac{1}{4}\bar \gamma^{\theta_3}_{p^4}$.
	Then, by the DSIC conditions for type $\theta_1$:
	\[
		\mathbb{E}_{p \sim \bar \gamma^{\theta_1}} \left[ \sum_{\omega \in \Omega} F_{\theta, b^{\theta_1}(p), \omega} p_\omega\right]\ge \bar \gamma^{\theta_2}_{p^5}p^5_{\omega_1}\ge\frac{1}{4}\bar \gamma^{\theta_3}_{p^4} .
	\]
	Thus, the overall principal's expected utility must be at most
	\[
		\frac{1}{3}[1-\frac{1}{4}\bar \gamma^{\theta_3}_{p^4} +  1 - \bar \gamma^{\theta_3}_{p^4} \frac{1}{4}  +  \bar \gamma^{\theta_3}_{p^4} (\frac{3}{4}-\frac{1}{4})=\frac{1}{3}[1+1 ] = \frac{2}{3},
	\]
	which contradicts our initial assumption on $\bar \Gamma$, concluding the proof.
\end{proof}

In the remaining part of this section, we show how to design an `` almost-optimal'' DSIC menu of randomized contracts, which is one providing principal' expected utility arbitrary close to $\SUP$.
First, we show that bounded payments are sufficient to provide an almost-optimal solution.
Formally:
\begin{restatable}{lemma}{lemmaRandTwo}\label{lem:bounded_payments}
	For every $\epsilon>0$, there always exists a DSIC menu of randomized contracts $\Gamma = \{ \gamma^\theta \}_{\theta \in \Theta}$ with principal's expected utility at least $\emph{\SUP}-\epsilon$ such that, for every $\theta \in \Theta$, it holds (i) $\left| \supp (\gamma^\theta) \cap \hat\pa^{\theta, a} \right| \leq 1$ for all $a \in A$, and (ii) $p_\omega \leq  C(I,\epsilon)$ for all $p \in \supp(\gamma^\theta)$ and $\omega \in \Omega$, where $C(I,\epsilon) \in O( \frac{1}{\epsilon} \cdot 2^{\poly(I)})$ and $I$ denotes the size of the problem instance.
\end{restatable}

%
%

Given a Bayesian principal-agent problem instance $(\Theta, A, \Omega)$ and an $\epsilon>0$, we let $C(I,\epsilon)$ be defined as in Lemma~\ref{lem:bounded_payments}.
Moreover, let $\pa^{\epsilon} \coloneqq [0,C(I,\epsilon)]^m$.
For every $\avec = (a_\theta)_{\theta \in \Theta} \in \bigtimes_{\theta \in \Theta} A$, where $\avec$ is a tuple specifying an action $a_\theta$ for each agent's type $\theta \in \Theta$, we let $\pa^{\avec,\epsilon} \coloneqq \pa^{\epsilon} \cap \Big( \bigcap_{\theta \in \Theta} \pa^{\theta, a_\theta} \Big)$ and, additionally, $\hat\pa^{\avec,\epsilon}\coloneqq \pa^{\epsilon} \cap \Big( \bigcap_{\theta \in \Theta} \hat\pa^{\theta, a_\theta} \Big)$.
We also define the set of all the vertexes of the closed polyhedra $\pa^{\avec,\epsilon}$ as $\pa^{*,\epsilon} \coloneqq \bigcup_{\avec \in \bigtimes_{\theta \in \Theta} A} V(\pa^{\avec,\epsilon})$, where $V(\cdot)$ denotes the set of vertexes of the polyhedron given as input.
Then, we can prove the following last lemma, which shows that for every $\epsilon>0$, in order to obtain a DSIC menu of randomized contracts with principal's expected utility at least $\SUP -\epsilon$, it is sufficient to restrict the attention to randomized contracts placing positive probability only on contracts in the (finite) set $\pa^{*,\epsilon}$.
\begin{restatable}{lemma}{lemmaRandThree}\label{lem:finite_support_star}
	For every $\epsilon>0$, there always exists a DSIC menu of randomized contracts supported on $\pa^{*,\epsilon}$ with principal's expected utility at least $\emph{\SUP}-\epsilon$.
\end{restatable}
%

%
%

Lemma~\ref{lem:finite_support_star} allows us to formulate the following LP which, given an $\epsilon>0$, has an optimal solution of value at least $\SUP-\epsilon$.
\begin{subequations}\label{prob:finite_lp}
	\begin{align}
		\max_{\Gamma = \{ \gamma^\theta \}_{\theta \in \Theta} } & \,\, \sum_{\theta \in \Theta} \mu_\theta \sum_{p \in \pa^{*,\epsilon}} \gamma_p^\theta \left[  \sum_{\omega \in \Omega} F_{\theta, b^\theta(p), \omega} r_\omega - \sum_{\omega \in \Omega} F_{\theta, b^\theta(p), \omega} p_\omega  \right] \quad \text{s.t.} \\
		& \hspace{-1.4cm} \sum_{p \in \pa^{*,\epsilon}} \gamma_p^\theta \left[ \sum_{\omega \in \Omega} F_{\theta, b^\theta(p), \omega} p_\omega - c_{\theta, b^\theta(p)}  \right] \geq \sum_{p \in \pa^{*,\epsilon}} \gamma_p^{\hat \theta} \left[ \sum_{\omega \in \Omega} F_{\theta, b^\theta(p), \omega} p_\omega - c_{\theta, b^\theta(p)}  \right] \nonumber\\ 
		& \hspace{7.7cm} \forall \theta \neq \hat \theta \in \Theta \label{eq:finite_lp_cons_1}\\
		& \hspace{-1.4cm} \sum_{p \in \pa^{*,\epsilon}} \gamma_p^\theta = 1 \hspace{7.8cm} \forall \theta \in \Theta.\label{eq:finite_lp_cons_2}
	\end{align}
\end{subequations}

Notice that LP~\ref{prob:finite_lp} has an exponential number of variables, since the probability distributions $\gamma^\theta$ are defined over contracts in $\pa^{*,\epsilon}$, and these may be exponentially many in the size of the problem instance.

Nevertheless, the LP has polynomially many constraints, and, thus, as we show next, we can solve it in polynomial time by applying the \emph{ellipsoid algorithm} to its dual program, which features polynomially-many variables and exponentially-many constraints.\footnote{An LP analogous to LP~\ref{prob:finite_lp} also arises when dealing with Bayesian persuasion problems in which the receiver can be of multiple types (see~\citep{castiglioni2022bayesian}). However, in that case, the counterparts of the terms $\gamma^\theta_{p}$ and $p_\omega$ (namely, for a receiver's type $k$, the probability $\gamma^k_\xi$ of a given posterior $\xi$ and the probability $\xi_\theta$ that the posterior assigns to a given state of nature $\theta$) always appear in a product $\gamma^\theta_p \, p_\omega$, which can be replaced by a suitably-defined new variable (subject to some consistency constraints). This property of the Bayesian persuasion setting allows to formulate the problem as an LP with polynomially-many variables and constraints. However, in our principal-agent setting, the variables $\gamma^\theta_{p}$ also appear without being multiplied by $p_\omega$, which prevents us from applying the same trick. As a result, in our setting, we had to resort to the ellipsoid algorithm.}

The dual of LP~\ref{prob:finite_lp} reads as follows:

\begin{subequations}\label{prob:finite_lp_dual}
	\begin{align}
	\min_{y \leq 0, t} & \,\, \sum_{\theta \in \Theta} t_\theta  \quad \text{s.t.} \\
	& \sum_{\hat\theta \in \Theta: \hat\theta \neq \theta} y_{\theta, \hat\theta} \left(  \sum_{\omega \in \Omega} F_{\theta, b^\theta(p), \omega} p_\omega - c_{\theta, b^\theta(p)}  \right) \nonumber \\
	& \quad\quad\quad - \sum_{\hat\theta \in \Theta: \hat\theta \neq \theta} y_{\theta, \hat\theta} \left(  \sum_{\omega \in \Omega} F_{\hat \theta, b^{\hat \theta}(p), \omega} p_\omega - c_{\hat \theta, b^{\hat \theta}(p)}  \right) + t_\theta \geq \nonumber\\
	& \quad\quad\quad   \mu_\theta \left(  \sum_{\omega \in \Omega} F_{\theta, b^\theta(p), \omega} r_\omega - \sum_{\omega \in \Omega} F_{\theta, b^\theta(p), \omega} p_\omega  \right)  \hspace{0.3cm}\forall \theta \in \Theta, \forall p\in \pa^{*,\epsilon}, \label{eq:finite_lp_dual_cons}
	\end{align}
\end{subequations}
where $y \in \mathbb{R}^{\ell (\ell - 1)}$ is a vector of dual variables whose components $y_{\theta, \hat\theta}$ for $\theta, \hat \theta \in \Theta: \theta \neq \hat \theta$ are the dual variables corresponding to Constraints~\eqref{eq:finite_lp_cons_1}, while $t \in \mathbb{R}^\ell$ is another vector of dual variables whose components $t_\theta$ for $\theta \in \Theta$ are the dual variables of Constraints~\eqref{eq:finite_lp_cons_2}.

The dual LP~\ref{prob:finite_lp_dual} has polynomially-many variables and exponentially-many constraints, and it can be optimally solved in polynomial time by means of the ellipsoid algorithm provided that a suitable polynomial-time \emph{separation oracle} is available.
In particular, given a pair $(y,t)$ assigning values to the variables of the dual LP, the separation oracle that we provide outputs a pair $\theta \in \Theta$, $ p \in \pa^{*,\epsilon}$ such that the corresponding inequality in Constraints~\eqref{eq:finite_lp_dual_cons} is violated for $(y,t)$, if any; otherwise, the oracle concludes that $(y,t)$ is feasible for the dual LP.
Formally, the separation oracle that we provide solves (in polynomial time) the following optimization problem for each $\theta \in \Theta$:
\begin{align}\label{prob:separation}
\max_{p \in \pa^{*,\epsilon}} \Bigg\{ & \mu_\theta \left(  \sum_{\omega \in \Omega} F_{\theta, b^\theta(p), \omega} r_\omega - \sum_{\omega \in \Omega} F_{\theta, b^\theta(p), \omega} p_\omega  \right) \\
& - \sum_{\hat\theta \in \Theta: \hat\theta \neq \theta} y_{\theta, \hat\theta} \left(  \sum_{\omega \in \Omega} F_{\theta, b^\theta(p), \omega} p_\omega - c_{\theta, b^\theta(p)}  \right)  \nonumber \\
& + \sum_{\hat\theta \in \Theta: \hat\theta \neq \theta} y_{\theta, \hat\theta} \left(  \sum_{\omega \in \Omega} F_{\hat \theta, b^{\hat \theta}(p), \omega} p_\omega - c_{\hat \theta, b^{\hat \theta}(p)}  \right) \Bigg\}. \nonumber
\end{align} 
Indeed, if the the value of the maximization above is greater than $t_\theta$ for some $\theta \in \Theta$, then the separation oracle outputs the pair $(\theta, p)$ with $p \in \pa^{*,\epsilon}$ being a contract for which the maximization is attained, since the constraint corresponding to $\theta$ and $p$ is violated.
Instead, if $t_\theta$ is less than or equal to the value of the maximization above for every $\theta \in \Theta$, then the oracle concludes that $(y,t)$ is feasible since no constraint is violated.

The following lemma shows that Problem~\ref{prob:separation} above can indeed be solved in polynomial time.
Notice that $\epsilon$ appears only in the coefficients of Problem~\ref{prob:finite_lp}. Indeed, such coefficients are at most $C(I,\epsilon)$. Hence, our algorithm runs in time polynomial in the number of bits needed to represent the coefficients, and, thus, polynomial in $\log(1/\epsilon)$.
Formally, we can state the following:
\begin{restatable}{lemma}{lemmaRandFour}
	For every $\epsilon>0$, there exists a separation oracle for LP~\ref{prob:finite_lp_dual}, that runs in time $\poly(I,\log(1/\epsilon))$, where $I$ is the size of the problem instance.
\end{restatable}

\begin{proof}
	Given a pair $(y,t) \in \mathbb{R}^{\ell (\ell - 1)} \times \mathbb{R}^{\ell}$ assigning values to the variables of LP~\ref{prob:finite_lp_dual}, the separation oracle solves Problem~\ref{prob:separation} for every $\theta \in \Theta$.
	Then, if there exists an agent's type $\theta \in \Theta$ such that $t_\theta$ is less than the value of an optimal solution $p \in \pa^{*,\epsilon}$ to Problem~\ref{prob:separation} for type $\theta$, the oracle outputs the pair $(\theta, p)$, which corresponds to a violated inequality in Constraints~\eqref{eq:finite_lp_dual_cons}.
	Instead, if $t_\theta$ is greater than or equal to the value of an optimal solution to Problem~\ref{prob:separation} for every $\theta \in \Theta$, the oracle concludes that $(y,t)$ is feasible, since all the inequalities in Constraints~\eqref{eq:finite_lp_dual_cons} are satisfied.
	
	Next, we show that, for every $\theta \in \Theta$, Problem~\ref{prob:separation} can be solved in polynomial time.
	In order to do that, we split the set $\pa^{*,\epsilon}$ into the subsets $\hat \pa^{\theta, a} \cap \pa^{*,\epsilon}$, defined for every agents' action $a \in A$.
	Notice that, by the definitions of $\hat\pa^{\theta, a}$ and $\pa^{*,\epsilon}$, each $p \in \pa^{*,\epsilon}$ belongs to exactly one subset $\hat \pa^{\theta, a} \cap \pa^{*,\epsilon}$.
	Thus, solving Problem~\ref{prob:separation} reduces to solving the following problem:
	\begin{align}\label{prob:separation_bis}
	\max_{a \in A} \max_{p \in \hat\pa^{\theta, a} \cap \pa^{*,\epsilon}} \Bigg\{ & \mu_\theta \left(  \sum_{\omega \in \Omega} F_{\theta, a, \omega} r_\omega - \sum_{\omega \in \Omega} F_{\theta, a, \omega} p_\omega  \right) \\
	& - \sum_{\hat\theta \in \Theta: \hat\theta \neq \theta} y_{\theta, \hat\theta} \left(  \sum_{\omega \in \Omega} F_{\theta, a, \omega} p_\omega - c_{\theta, a}  \right)  \nonumber \\
	&+ \sum_{\hat\theta \in \Theta: \hat\theta \neq \theta} y_{\theta, \hat\theta} \left(  \sum_{\omega \in \Omega} F_{\hat \theta, b^{\hat \theta}(p), \omega} \,p_\omega - c_{\hat \theta, b^{\hat \theta}(p)}  \right) \Bigg\}, \nonumber
	\end{align} 
	where we used the fact that $b^\theta(p) = a$ for $p \in \hat\pa^{\theta, a}$.

	As a first step, we prove that Problem~\ref{prob:separation_bis} is equivalent to the following one:
	\begin{align}\label{prob:separation_tris}
	\max_{a \in A} \max_{p \in \pa^{\theta, a} \cap \pa^{*,\epsilon}} \Bigg\{ & \mu_\theta \left(  \sum_{\omega \in \Omega} F_{\theta, a, \omega} r_\omega - \sum_{\omega \in \Omega} F_{\theta, a, \omega} p_\omega  \right) \\
	& - \sum_{\hat\theta \in \Theta: \hat\theta \neq \theta} y_{\theta, \hat\theta} \left(  \sum_{\omega \in \Omega} F_{\theta, a, \omega} p_\omega - c_{\theta, a}  \right) \nonumber  \\
	&+ \sum_{\hat\theta \in \Theta: \hat\theta \neq \theta} y_{\theta, \hat\theta} \left(  \sum_{\omega \in \Omega} F_{\hat \theta, b^{\hat \theta}(p), \omega}\, p_\omega - c_{\hat \theta, b^{\hat \theta}(p)}  \right) \Bigg\}, \nonumber
	\end{align}
	where we replace $\hat\pa^{\theta,a}$ with $\pa^{\theta,a}$ in the inner maximization.
	Indeed, by adding contracts $p \in \left(  \pa^{\theta,a} \setminus \hat\pa^{\theta,a} \right) \cap \pa^{*,\epsilon}$ to the domain of the inner maximizations for every $a \in A$ does \emph{not} change the optimal value of the overall maximization problem.
	In order to see this, notice that, for every $p \in \left(  \pa^{\theta,a} \setminus \hat\pa^{\theta,a} \right) \cap \pa^{*,\epsilon}$, by definition of best response it holds $\sum_{\omega \in \Omega} F_{\theta,b^\theta(p), \omega} p_\omega - c_{\theta,b^\theta(p)} = \sum_{\omega \in \Omega} F_{\theta,a, \omega} p_\omega - c_{\theta,a}$ and $\sum_{\omega \in \Omega} F_{\theta, b^\theta(p), \omega} r_\omega - \sum_{\omega \in \Omega} F_{\theta, b^\theta(p), \omega} p_\omega \geq \sum_{\omega \in \Omega} F_{\theta, a, \omega} r_\omega - \sum_{\omega \in \Omega} F_{\theta, a, \omega} p_\omega$ (tie-breaking rule).
	Thus, any pair $(a, p)$ with $p \in \left(  \pa^{\theta,a} \setminus \hat\pa^{\theta,a} \right) \cap \pa^{*,\epsilon}$ would result in an objective value smaller than that of the pair $(b^\theta(p),p)$.
	Recalling that $p \in \hat\pa^{\theta, b^\theta(p)}$ by definition, we can conclude that Problem~\ref{prob:separation_bis} and Problem~\ref{prob:separation_tris} are indeed equivalent.
	
	We are left to prove that Problem~\ref{prob:separation_tris} can be solved in polynomial time.
	First, notice that solving Problem~\ref{prob:separation_tris} is equivalent to solving the following problem for every $a\in A$:
	\begin{align}\label{prob:separation_noaction}
	\max_{p \in \pa^{\theta, a} \cap \pa^{*,\epsilon}} \Bigg\{ & \mu_\theta \left(  \sum_{\omega \in \Omega} F_{\theta, a, \omega} r_\omega - \sum_{\omega \in \Omega} F_{\theta, a, \omega} p_\omega  \right) \\
	& - \sum_{\hat\theta \in \Theta: \hat\theta \neq \theta} y_{\theta, \hat\theta} \left(  \sum_{\omega \in \Omega} F_{\theta, a, \omega} p_\omega - c_{\theta, a}  \right) \nonumber \\
	&+ \sum_{\hat\theta \in \Theta: \hat\theta \neq \theta} y_{\theta, \hat\theta} \left(  \sum_{\omega \in \Omega} F_{\hat \theta, b^{\hat \theta}(p), \omega} p_\omega - c_{\hat \theta, b^{\hat \theta}(p)}  \right) \Bigg\}. \nonumber
	\end{align}
	The first and second terms in the maximization above are linear in the payments $p_\omega$ for $\omega \in \Omega$.
	Moreover, it holds:
	\begin{align*}
	\sum_{\hat\theta \in \Theta: \hat\theta \neq \theta} y_{\theta, \hat\theta} & \left(  \sum_{\omega \in \Omega} F_{\hat \theta, b^{\hat \theta}(p), \omega} p_\omega - c_{\hat \theta, b^{\hat \theta}(p)}  \right) \\
	& = \sum_{\hat\theta \in \Theta: \hat\theta \neq \theta} y_{\theta, \hat\theta} \max_{a' \in A}\left\{   \sum_{\omega \in \Omega} F_{\hat \theta, a', \omega} p_\omega - c_{\hat \theta, a'}  \right\},
	\end{align*}
	and the latter is a concave function of the payments $p_\omega$ since all the $y_{\theta,\hat\theta}$ are negative and the $\max$ is convex.
	These observations allow us to solve Problem~\ref{prob:separation_tris} in polynomial time by first solving the following LP relaxation obtained by weakening the requirement $p \in \pa^{\theta,a} \cap \pa^{*,\epsilon}$ as $p \in \pa^{\theta,a} \cap \pa^{\epsilon}$:
	\begin{subequations}\label{prob:separation_lp_relaxed}
		\begin{align}
		\max_{p \in \pa^{\epsilon}} & \,\, \mu_\theta \left(  \sum_{\omega \in \Omega} F_{\theta, a, \omega} r_\omega - \sum_{\omega \in \Omega} F_{\theta, a, \omega} p_\omega  \right) \label{eq:separation_lp_relaxed_obj}\\
		& \,\,\,\, - \sum_{\hat\theta \in \Theta: \hat\theta \neq \theta} y_{\theta, \hat\theta} \left(  \sum_{\omega \in \Omega} F_{\theta, a, \omega} p_\omega - c_{\theta, a}  \right)  + \sum_{\hat\theta \in \Theta: \hat\theta \neq \theta} y_{\theta, \hat\theta} z_{\hat \theta} \hspace{1cm} \text{s.t.}  \nonumber\\
		& z_{\hat \theta} \geq \sum_{\omega \in \Omega} F_{\hat \theta, a', \omega} p_\omega - c_{\hat \theta, a'} \hspace{2.5cm} \forall \hat\theta \in \Theta : \hat\theta \neq \theta, \forall a' \in A \label{eq:separation_lp_relaxed_cons1}\\
		& \sum_{\omega \in \Omega} F_{ \theta, a, \omega} p_\omega - c_{ \theta, a} \geq \sum_{\omega \in \Omega} F_{ \theta, a', \omega} p_\omega - c_{ \theta, a'} \hspace{2.2cm} \forall a' \in A,\label{eq:separation_lp_relaxed_cons2}
		\end{align}
	\end{subequations}
	where Constraints~\eqref{eq:separation_lp_relaxed_cons1} ensure that $z_{\hat\theta} = \max_{a' \in A}\left\{   \sum_{\omega \in \Omega} F_{\hat \theta, a', \omega} p_\omega - c_{\hat \theta, a'}  \right\}$ for every $\hat \theta \in \Theta: \hat\theta \neq \theta$ (since in Objective~\eqref{eq:separation_lp_relaxed_obj} all the coefficients of the variables $z_{\hat \theta}$ are negative), while Constraints~\eqref{eq:separation_lp_relaxed_cons2} are equivalent to $p \in \pa^{\theta,a}$.
	Given an optimal solution to LP~\ref{prob:separation_lp_relaxed}, it is possible to recover in polynomial time a contract $p \in \pa^{\theta,a} \cap \pa^{*,\epsilon}$ that is an optimal solution to Problem~\ref{prob:separation_noaction}.
	Given a contract $ p^{\ast} \in \pa^{\theta, a} \cap \pa^{\epsilon}$ that is an optimal solution to LP~\ref{prob:separation_lp_relaxed}, let $\avec^* = (a^{p^\ast}_{\hat \theta})$, where $a^{p^\ast}_{\theta}=a$ and  $a^{p^\ast}_{\hat \theta}=b^{\hat \theta}(p^\ast)$ for each $\hat \theta \in \Theta: \hat \theta \neq \theta$.
	%
	%
	Then, since $a^{p^\ast}_\theta=a $, and for each contract $\hat p \in \pa^{\avec^*,\epsilon}$ and each $\hat \theta \in \Theta:\hat \theta \neq \theta$, $ b^{\hat\theta}(p^\ast) \in \mathcal{B}^{\hat \theta}_{\hat p}$, the following LP is equivalent to LP~\ref{prob:separation_lp_relaxed} restricted to the feasibility set $\pa^{\avec^{*,\epsilon}}\subseteq \pa^{\theta,a} \cap \pa^{\epsilon}$. Hence, the contract $p^\ast$ defined above is also an optimal solution to the following LP:
	\begin{subequations}\label{prob:separation_lp_relaxed_vertex}
		\begin{align}
		\max_{p \in \pa^{\epsilon}} & \,\, \mu_\theta \left(  \sum_{\omega \in \Omega} F_{\theta, a, \omega} r_\omega - \sum_{\omega \in \Omega} F_{\theta, a, \omega} p_\omega  \right) \\
		& \,\,\,\, - \sum_{\hat\theta \in \Theta: \hat\theta \neq \theta} y_{\theta, \hat\theta} \left(  \sum_{\omega \in \Omega} F_{\theta, a, \omega} p_\omega - c_{\theta, a}  \right) \nonumber \\
		& \,\,\,\, + \sum_{\hat\theta \in \Theta: \hat\theta \neq \theta} y_{\theta, \hat\theta} \left(  \sum_{\omega \in \Omega} F_{\hat \theta, b^{\hat \theta}(p^\ast), \omega} p_\omega - c_{\hat \theta, b^{\hat \theta}(p^\ast)}  \right) \hspace{2.5cm} \text{s.t.} \\
		& \sum_{\omega \in \Omega} F_{\hat \theta, a^{p^\ast}_{\hat \theta}, \omega} \, p_\omega - c_{\hat \theta, a^{p^\ast}_{\hat \theta}} \geq \sum_{\omega \in \Omega} F_{\hat \theta, a', \omega} \, p_\omega - c_{\hat \theta, a'} \quad \ \  \forall \hat\theta \in \Theta, \forall a' \in A.
		\end{align}
	\end{subequations}
	Thus, there always exists an optimal solution to LP~\ref{prob:separation_lp_relaxed_vertex} that is a vertex of its feasibility polytope $\pa^{\avec^\ast,\epsilon}$, \emph{i.e.}, a contract in $V(\pa^{\avec^\ast,\epsilon}) \subseteq \pa^{\theta, a} \cap \pa^{*,\epsilon}$, and such optimal solution can be found in polynomial time~\citep{ye1990recovering}.  
\end{proof}

By the previous lemma, it immediately follows the main result of the paper.

\begin{theorem}
	In Bayesian principal-agent problems, there exists an algorithm that, given any $\epsilon > 0$, runs in time $\poly(I, \log(1/\epsilon))$ and computes a DSIC menu of randomized contracts providing the principal with an expected utility at least of $\SUP-\epsilon$, where $I$ is the size of the problem instance.
\end{theorem}


\section*{Acknowledgment}
We are thankful to Jibang Wu and Haifeng Xu for reporting an error in the previous version of the paper.

\bibliographystyle{ACM-Reference-Format}
\bibliography{biblio}

\begin{thebibliography}{30}
\providecommand{\natexlab}[1]{#1}
\providecommand{\url}[1]{\texttt{#1}}
\providecommand{\href}[2]{#2}
\providecommand{\path}[1]{#1}
\providecommand{\eprint}[1]{\href{http://arxiv.org/abs/#1}{\path{#1}}}
\providecommand{\DOIprefix}{doi:}
\providecommand{\ArXivprefix}{arXiv:}
\providecommand{\URLprefix}{URL: }
\providecommand{\Pubmedprefix}{pmid:}
\providecommand{\doi}[1]{\href{http://dx.doi.org/#1}{\path{#1}}}
\providecommand{\Pubmed}[1]{\href{pmid:#1}{\path{#1}}}
\providecommand{\BIBand}{and}
\providecommand{\bibinfo}[2]{#2}
\ifx\xfnm\undefined \def\xfnm[#1]{\unskip,\space#1}\fi
\bibitem[{Castiglioni et~al.(2022{\natexlab{a}})Castiglioni, Marchesi and
  Gatti}]{CastiglioniEC22}
\bibinfo{author}{Castiglioni\xfnm[ M.]}, \bibinfo{author}{Marchesi\xfnm[ A.]},
  \bibinfo{author}{Gatti\xfnm[ N.]}.
\newblock \bibinfo{title}{Designing menus of contracts efficiently: The power
  of randomization}.
\newblock In: \bibinfo{booktitle}{Proceedings of the 23rd ACM Conference on
  Economics and Computation}. EC '22; \bibinfo{address}{New York, NY, USA}:
  \bibinfo{publisher}{Association for Computing Machinery}.
\newblock ISBN \bibinfo{isbn}{9781450391504};
  \bibinfo{year}{2022}{\natexlab{a}}, p. \bibinfo{pages}{705–735}.
\newblock \URLprefix \url{https://doi.org/10.1145/3490486.3538270}.
  \DOIprefix\doi{10.1145/3490486.3538270}.
\bibitem[{Ho et~al.(2016)Ho, Slivkins and Vaughan}]{ho2016adaptive}
\bibinfo{author}{Ho\xfnm[ C.J.]}, \bibinfo{author}{Slivkins\xfnm[ A.]},
  \bibinfo{author}{Vaughan\xfnm[ J.W.]}.
\newblock \bibinfo{title}{Adaptive contract design for crowdsourcing markets:
  Bandit algorithms for repeated principal-agent problems}.
\newblock \bibinfo{journal}{Journal of Artificial Intelligence Research}
  \bibinfo{year}{2016};\bibinfo{volume}{55}:\bibinfo{pages}{317--359}.
\bibitem[{Cong and He(2019)}]{cong2019blockchain}
\bibinfo{author}{Cong\xfnm[ L.W.]}, \bibinfo{author}{He\xfnm[ Z.]}.
\newblock \bibinfo{title}{Blockchain disruption and smart contracts}.
\newblock \bibinfo{journal}{The Review of Financial Studies}
  \bibinfo{year}{2019};\bibinfo{volume}{32}(\bibinfo{number}{5}):\bibinfo{pages}{1754--1797}.
\bibitem[{Bastani et~al.(2016)Bastani, Bayati, Braverman, Gummadi and
  Johari}]{bastani2016analysis}
\bibinfo{author}{Bastani\xfnm[ H.]}, \bibinfo{author}{Bayati\xfnm[ M.]},
  \bibinfo{author}{Braverman\xfnm[ M.]}, \bibinfo{author}{Gummadi\xfnm[ R.]},
  \bibinfo{author}{Johari\xfnm[ R.]}.
\newblock \bibinfo{title}{Analysis of medicare pay-for-performance contracts}.
\newblock \bibinfo{journal}{Available at SSRN 2839143} \bibinfo{year}{2016};.
\bibitem[{Guruganesh et~al.(2021)Guruganesh, Schneider and
  Wang}]{guruganesh2020contracts}
\bibinfo{author}{Guruganesh\xfnm[ G.]}, \bibinfo{author}{Schneider\xfnm[ J.]},
  \bibinfo{author}{Wang\xfnm[ J.R.]}.
\newblock \bibinfo{title}{Contracts under moral hazard and adverse selection}.
\newblock In: \bibinfo{booktitle}{Proceedings of the 22nd ACM Conference on
  Economics and Computation}. EC '21; \bibinfo{address}{New York, NY, USA}:
  \bibinfo{publisher}{Association for Computing Machinery}.
\newblock ISBN \bibinfo{isbn}{9781450385541}; \bibinfo{year}{2021}, p.
  \bibinfo{pages}{563–582}.
\newblock \URLprefix \url{https://doi.org/10.1145/3465456.3467637}.
  \DOIprefix\doi{10.1145/3465456.3467637}.
\bibitem[{Castiglioni et~al.(2021)Castiglioni, Marchesi and
  Gatti}]{castiglioni2021bayesian}
\bibinfo{author}{Castiglioni\xfnm[ M.]}, \bibinfo{author}{Marchesi\xfnm[ A.]},
  \bibinfo{author}{Gatti\xfnm[ N.]}.
\newblock \bibinfo{title}{Bayesian agency: Linear versus tractable contracts}.
\newblock \bibinfo{year}{2021}.
\newblock \href{http://arxiv.org/abs/2106.00319}{\tt arXiv:2106.00319}.
\bibitem[{Alon et~al.(2021{\natexlab{a}})Alon, D{\"u}tting and
  Talgam-Cohen}]{alon2021contracts}
\bibinfo{author}{Alon\xfnm[ T.]}, \bibinfo{author}{D{\"u}tting\xfnm[ P.]},
  \bibinfo{author}{Talgam-Cohen\xfnm[ I.]}.
\newblock \bibinfo{title}{Contracts with private cost per unit-of-effort}.
\newblock In: \bibinfo{booktitle}{Proceedings of the 22nd ACM Conference on
  Economics and Computation}. \bibinfo{year}{2021}{\natexlab{a}}, p.
  \bibinfo{pages}{52--69}.
\bibitem[{D{\"u}tting et~al.(2019)D{\"u}tting, Roughgarden and
  Talgam-Cohen}]{dutting2019simple}
\bibinfo{author}{D{\"u}tting\xfnm[ P.]}, \bibinfo{author}{Roughgarden\xfnm[
  T.]}, \bibinfo{author}{Talgam-Cohen\xfnm[ I.]}.
\newblock \bibinfo{title}{Simple versus optimal contracts}.
\newblock In: \bibinfo{booktitle}{Proceedings of the 2019 ACM Conference on
  Economics and Computation}. \bibinfo{year}{2019}, p.
  \bibinfo{pages}{369--387}.
\bibitem[{Shavell(1979)}]{shavell1979risk}
\bibinfo{author}{Shavell\xfnm[ S.]}.
\newblock \bibinfo{title}{Risk sharing and incentives in the principal and
  agent relationship}.
\newblock \bibinfo{journal}{The Bell Journal of Economics}
  \bibinfo{year}{1979};:\bibinfo{pages}{55--73}.
\bibitem[{Grossman and Hart(1983)}]{grossman1983analysis}
\bibinfo{author}{Grossman\xfnm[ S.J.]}, \bibinfo{author}{Hart\xfnm[ O.D.]}.
\newblock \bibinfo{title}{An analysis of the principal-agent problem}.
\newblock \bibinfo{journal}{Econometrica}
  \bibinfo{year}{1983};\bibinfo{volume}{51}(\bibinfo{number}{1}):\bibinfo{pages}{7--46}.
\bibitem[{Rogerson(1985)}]{rogerson1985repeated}
\bibinfo{author}{Rogerson\xfnm[ W.P.]}.
\newblock \bibinfo{title}{Repeated moral hazard}.
\newblock \bibinfo{journal}{Econometrica: Journal of the Econometric Society}
  \bibinfo{year}{1985};:\bibinfo{pages}{69--76}.
\bibitem[{Holmstrom and Milgrom(1991)}]{holmstrom1991multitask}
\bibinfo{author}{Holmstrom\xfnm[ B.]}, \bibinfo{author}{Milgrom\xfnm[ P.]}.
\newblock \bibinfo{title}{Multitask principal-agent analyses: Incentive
  contracts, asset ownership, and job design}.
\newblock \bibinfo{journal}{Journal of Law, Economics, \& Organization}
  \bibinfo{year}{1991};\bibinfo{volume}{7}:\bibinfo{pages}{24}.
\bibitem[{Mas-Colell et~al.(1995)Mas-Colell, Whinston, Green
  et~al.}]{mas1995microeconomic}
\bibinfo{author}{Mas-Colell\xfnm[ A.]}, \bibinfo{author}{Whinston\xfnm[ M.D.]},
  \bibinfo{author}{Green\xfnm[ J.R.]}, et~al.
\newblock \bibinfo{title}{Microeconomic theory}; vol.~\bibinfo{volume}{1}.
\newblock \bibinfo{publisher}{Oxford university press New York};
  \bibinfo{year}{1995}.
\bibitem[{Bolton et~al.(2005)Bolton, Dewatripont et~al.}]{bolton2005contract}
\bibinfo{author}{Bolton\xfnm[ P.]}, \bibinfo{author}{Dewatripont\xfnm[ M.]},
  et~al.
\newblock \bibinfo{title}{Contract theory}.
\newblock \bibinfo{publisher}{MIT press}; \bibinfo{year}{2005}.
\bibitem[{Laffont and Martimort(2009)}]{laffont2009theory}
\bibinfo{author}{Laffont\xfnm[ J.J.]}, \bibinfo{author}{Martimort\xfnm[ D.]}.
\newblock \bibinfo{title}{The theory of incentives: the principal-agent model}.
\newblock \bibinfo{publisher}{Princeton university press};
  \bibinfo{year}{2009}.
\bibitem[{Babaioff et~al.(2006)Babaioff, Feldman and
  Nisan}]{babaioff2006combinatorial}
\bibinfo{author}{Babaioff\xfnm[ M.]}, \bibinfo{author}{Feldman\xfnm[ M.]},
  \bibinfo{author}{Nisan\xfnm[ N.]}.
\newblock \bibinfo{title}{Combinatorial agency}.
\newblock In: \bibinfo{booktitle}{Proceedings of the 7th ACM Conference on
  Electronic Commerce}. \bibinfo{year}{2006}, p. \bibinfo{pages}{18--28}.
\bibitem[{Babaioff et~al.(2012)Babaioff, Feldman, Nisan and
  Winter}]{babaioff2012combinatorial}
\bibinfo{author}{Babaioff\xfnm[ M.]}, \bibinfo{author}{Feldman\xfnm[ M.]},
  \bibinfo{author}{Nisan\xfnm[ N.]}, \bibinfo{author}{Winter\xfnm[ E.]}.
\newblock \bibinfo{title}{Combinatorial agency}.
\newblock \bibinfo{journal}{Journal of Economic Theory}
  \bibinfo{year}{2012};\bibinfo{volume}{147}(\bibinfo{number}{3}):\bibinfo{pages}{999--1034}.
\bibitem[{Babaioff et~al.(2009)Babaioff, Feldman and Nisan}]{babaioff2009free}
\bibinfo{author}{Babaioff\xfnm[ M.]}, \bibinfo{author}{Feldman\xfnm[ M.]},
  \bibinfo{author}{Nisan\xfnm[ N.]}.
\newblock \bibinfo{title}{Free-riding and free-labor in combinatorial agency}.
\newblock In: \bibinfo{booktitle}{International Symposium on Algorithmic Game
  Theory}. \bibinfo{organization}{Springer}; \bibinfo{year}{2009}, p.
  \bibinfo{pages}{109--121}.
\bibitem[{Babaioff et~al.(2010)Babaioff, Feldman and Nisan}]{babaioff2010mixed}
\bibinfo{author}{Babaioff\xfnm[ M.]}, \bibinfo{author}{Feldman\xfnm[ M.]},
  \bibinfo{author}{Nisan\xfnm[ N.]}.
\newblock \bibinfo{title}{Mixed strategies in combinatorial agency}.
\newblock \bibinfo{journal}{Journal of Artificial Intelligence Research}
  \bibinfo{year}{2010};\bibinfo{volume}{38}:\bibinfo{pages}{339--369}.
\bibitem[{D{\"u}tting et~al.(2020)D{\"u}tting, Roughgarden and
  Cohen}]{dutting2020complexity}
\bibinfo{author}{D{\"u}tting\xfnm[ P.]}, \bibinfo{author}{Roughgarden\xfnm[
  T.]}, \bibinfo{author}{Cohen\xfnm[ I.T.]}.
\newblock \bibinfo{title}{The complexity of contracts}.
\newblock In: \bibinfo{booktitle}{Proceedings of the Fourteenth Annual ACM-SIAM
  Symposium on Discrete Algorithms}. \bibinfo{organization}{SIAM};
  \bibinfo{year}{2020}, p. \bibinfo{pages}{2688--2707}.
\bibitem[{Duetting et~al.(2021)Duetting, Ezra, Feldman and
  Kesselheim}]{duetting2021combinatorial}
\bibinfo{author}{Duetting\xfnm[ P.]}, \bibinfo{author}{Ezra\xfnm[ T.]},
  \bibinfo{author}{Feldman\xfnm[ M.]}, \bibinfo{author}{Kesselheim\xfnm[ T.]}.
\newblock \bibinfo{title}{Combinatorial contracts}.
\newblock \bibinfo{journal}{arXiv preprint arXiv:210914260}
  \bibinfo{year}{2021};.
\bibitem[{Babaioff and Winter(2014)}]{babaioff2014contract}
\bibinfo{author}{Babaioff\xfnm[ M.]}, \bibinfo{author}{Winter\xfnm[ E.]}.
\newblock \bibinfo{title}{Contract complexity.}
\newblock \bibinfo{journal}{EC}
  \bibinfo{year}{2014};\bibinfo{volume}{14}:\bibinfo{pages}{911}.
\bibitem[{Alon et~al.(2021{\natexlab{b}})Alon, D{\"u}tting and
  Talgam-Cohen}]{alon2021contractsarvix}
\bibinfo{author}{Alon\xfnm[ T.]}, \bibinfo{author}{D{\"u}tting\xfnm[ P.]},
  \bibinfo{author}{Talgam-Cohen\xfnm[ I.]}.
\newblock \bibinfo{title}{Contracts with private cost per unit-of-effort}.
\newblock \bibinfo{journal}{arXiv preprint arXiv:211109179}
  \bibinfo{year}{2021}{\natexlab{b}};.
\bibitem[{Carroll(2015)}]{carroll2015robustness}
\bibinfo{author}{Carroll\xfnm[ G.]}.
\newblock \bibinfo{title}{Robustness and linear contracts}.
\newblock \bibinfo{journal}{American Economic Review}
  \bibinfo{year}{2015};\bibinfo{volume}{105}(\bibinfo{number}{2}):\bibinfo{pages}{536--63}.
\bibitem[{Shoham and Leyton-Brown(2008)}]{shoham2008multiagent}
\bibinfo{author}{Shoham\xfnm[ Y.]}, \bibinfo{author}{Leyton-Brown\xfnm[ K.]}.
\newblock \bibinfo{title}{Multiagent systems: Algorithmic, game-theoretic, and
  logical foundations}.
\newblock \bibinfo{publisher}{Cambridge University Press};
  \bibinfo{year}{2008}.
\bibitem[{Alon et~al.(1995)Alon, Feige, Wigderson and
  Zuckerman}]{Alon1995IndependentSet}
\bibinfo{author}{Alon\xfnm[ N.]}, \bibinfo{author}{Feige\xfnm[ U.]},
  \bibinfo{author}{Wigderson\xfnm[ A.]}, \bibinfo{author}{Zuckerman\xfnm[ D.]}.
\newblock \bibinfo{title}{Derandomized graph products}.
\newblock \bibinfo{journal}{computational complexity}
  \bibinfo{year}{1995};\bibinfo{volume}{5}(\bibinfo{number}{1}):\bibinfo{pages}{60--75}.
\newblock \URLprefix \url{https://doi.org/10.1007/BF01277956}.
  \DOIprefix\doi{10.1007/BF01277956}.
\bibitem[{Trevisan(2001)}]{Trevisan2001Independent}
\bibinfo{author}{Trevisan\xfnm[ L.]}.
\newblock \bibinfo{title}{Non-approximability results for optimization problems
  on bounded degree instances}.
\newblock In: \bibinfo{booktitle}{Proceedings of the Thirty-Third Annual ACM
  Symposium on Theory of Computing}. STOC '01; \bibinfo{address}{New York, NY,
  USA}: \bibinfo{publisher}{Association for Computing Machinery}.
\newblock ISBN \bibinfo{isbn}{1581133499}; \bibinfo{year}{2001}, p.
  \bibinfo{pages}{453–461}.
\newblock \URLprefix \url{https://doi.org/10.1145/380752.380839}.
  \DOIprefix\doi{10.1145/380752.380839}.
\bibitem[{Castiglioni et~al.(2022{\natexlab{b}})Castiglioni, Marchesi and
  Gatti}]{castiglioni2022bayesian}
\bibinfo{author}{Castiglioni\xfnm[ M.]}, \bibinfo{author}{Marchesi\xfnm[ A.]},
  \bibinfo{author}{Gatti\xfnm[ N.]}.
\newblock \bibinfo{title}{Bayesian persuasion meets mechanism design: Going
  beyond intractability with type reporting}.
\newblock \bibinfo{journal}{arXiv preprint arXiv:220200605}
  \bibinfo{year}{2022}{\natexlab{b}};.
\bibitem[{Ye(1990)}]{ye1990recovering}
\bibinfo{author}{Ye\xfnm[ Y.]}.
\newblock \bibinfo{title}{Recovering optimal basic variables in karmarkar's
  polynomial algorithm for linear programming}.
\newblock \bibinfo{journal}{Mathematics of operations research}
  \bibinfo{year}{1990};\bibinfo{volume}{15}(\bibinfo{number}{3}):\bibinfo{pages}{564--572}.
\bibitem[{Bertsimas and Tsitsiklis(1997)}]{bertsimas1997introduction}
\bibinfo{author}{Bertsimas\xfnm[ D.]}, \bibinfo{author}{Tsitsiklis\xfnm[
  J.N.]}.
\newblock \bibinfo{title}{Introduction to linear optimization};
  vol.~\bibinfo{volume}{6}.
\newblock \bibinfo{publisher}{Athena Scientific Belmont, MA};
  \bibinfo{year}{1997}.

\end{thebibliography}



\begin{thebibliography}{30}


\ifx \showCODEN    \undefined \def \showCODEN     #1{\unskip}     \fi
\ifx \showDOI      \undefined \def \showDOI       #1{#1}\fi
\ifx \showISBNx    \undefined \def \showISBNx     #1{\unskip}     \fi
\ifx \showISBNxiii \undefined \def \showISBNxiii  #1{\unskip}     \fi
\ifx \showISSN     \undefined \def \showISSN      #1{\unskip}     \fi
\ifx \showLCCN     \undefined \def \showLCCN      #1{\unskip}     \fi
\ifx \shownote     \undefined \def \shownote      #1{#1}          \fi
\ifx \showarticletitle \undefined \def \showarticletitle #1{#1}   \fi
\ifx \showURL      \undefined \def \showURL       {\relax}        \fi
\providecommand\bibfield[2]{#2}
\providecommand\bibinfo[2]{#2}
\providecommand\natexlab[1]{#1}
\providecommand\showeprint[2][]{arXiv:#2}

\bibitem[\protect\citeauthoryear{Alon, Feige, Wigderson, and Zuckerman}{Alon
  et~al\mbox{.}}{1995}]%
        {Alon1995IndependentSet}
\bibfield{author}{\bibinfo{person}{Noga Alon}, \bibinfo{person}{Uriel Feige},
  \bibinfo{person}{Avi Wigderson}, {and} \bibinfo{person}{David Zuckerman}.}
  \bibinfo{year}{1995}\natexlab{}.
\newblock \showarticletitle{Derandomized graph products}.
\newblock \bibinfo{journal}{\emph{computational complexity}}
  \bibinfo{volume}{5}, \bibinfo{number}{1} (\bibinfo{year}{1995}),
  \bibinfo{pages}{60--75}.
\newblock
\showISBNx{1420-8954}
\urldef\tempurl%
\url{https://doi.org/10.1007/BF01277956}
\showDOI{\tempurl}


\bibitem[\protect\citeauthoryear{Alon, D{\"u}tting, and Talgam-Cohen}{Alon
  et~al\mbox{.}}{2021a}]%
        {alon2021contracts}
\bibfield{author}{\bibinfo{person}{Tal Alon}, \bibinfo{person}{Paul
  D{\"u}tting}, {and} \bibinfo{person}{Inbal Talgam-Cohen}.}
  \bibinfo{year}{2021}\natexlab{a}.
\newblock \showarticletitle{Contracts with Private Cost per Unit-of-Effort}. In
  \bibinfo{booktitle}{\emph{Proceedings of the 22nd ACM Conference on Economics
  and Computation}}. \bibinfo{pages}{52--69}.
\newblock


\bibitem[\protect\citeauthoryear{Alon, D{\"u}tting, and Talgam-Cohen}{Alon
  et~al\mbox{.}}{2021b}]%
        {alon2021contractsarvix}
\bibfield{author}{\bibinfo{person}{Tal Alon}, \bibinfo{person}{Paul
  D{\"u}tting}, {and} \bibinfo{person}{Inbal Talgam-Cohen}.}
  \bibinfo{year}{2021}\natexlab{b}.
\newblock \showarticletitle{Contracts with Private Cost per Unit-of-Effort}.
\newblock \bibinfo{journal}{\emph{arXiv preprint arXiv:2111.09179}}
  (\bibinfo{year}{2021}).
\newblock


\bibitem[\protect\citeauthoryear{Babaioff, Feldman, and Nisan}{Babaioff
  et~al\mbox{.}}{2006}]%
        {babaioff2006combinatorial}
\bibfield{author}{\bibinfo{person}{Moshe Babaioff}, \bibinfo{person}{Michal
  Feldman}, {and} \bibinfo{person}{Noam Nisan}.}
  \bibinfo{year}{2006}\natexlab{}.
\newblock \showarticletitle{Combinatorial agency}. In
  \bibinfo{booktitle}{\emph{Proceedings of the 7th ACM Conference on Electronic
  Commerce}}. \bibinfo{pages}{18--28}.
\newblock


\bibitem[\protect\citeauthoryear{Babaioff, Feldman, and Nisan}{Babaioff
  et~al\mbox{.}}{2009}]%
        {babaioff2009free}
\bibfield{author}{\bibinfo{person}{Moshe Babaioff}, \bibinfo{person}{Michal
  Feldman}, {and} \bibinfo{person}{Noam Nisan}.}
  \bibinfo{year}{2009}\natexlab{}.
\newblock \showarticletitle{Free-riding and free-labor in combinatorial
  agency}. In \bibinfo{booktitle}{\emph{International Symposium on Algorithmic
  Game Theory}}. Springer, \bibinfo{pages}{109--121}.
\newblock


\bibitem[\protect\citeauthoryear{Babaioff, Feldman, and Nisan}{Babaioff
  et~al\mbox{.}}{2010}]%
        {babaioff2010mixed}
\bibfield{author}{\bibinfo{person}{Moshe Babaioff}, \bibinfo{person}{Michal
  Feldman}, {and} \bibinfo{person}{Noam Nisan}.}
  \bibinfo{year}{2010}\natexlab{}.
\newblock \showarticletitle{Mixed strategies in combinatorial agency}.
\newblock \bibinfo{journal}{\emph{Journal of Artificial Intelligence Research}}
   \bibinfo{volume}{38} (\bibinfo{year}{2010}), \bibinfo{pages}{339--369}.
\newblock


\bibitem[\protect\citeauthoryear{Babaioff, Feldman, Nisan, and Winter}{Babaioff
  et~al\mbox{.}}{2012}]%
        {babaioff2012combinatorial}
\bibfield{author}{\bibinfo{person}{Moshe Babaioff}, \bibinfo{person}{Michal
  Feldman}, \bibinfo{person}{Noam Nisan}, {and} \bibinfo{person}{Eyal Winter}.}
  \bibinfo{year}{2012}\natexlab{}.
\newblock \showarticletitle{Combinatorial agency}.
\newblock \bibinfo{journal}{\emph{Journal of Economic Theory}}
  \bibinfo{volume}{147}, \bibinfo{number}{3} (\bibinfo{year}{2012}),
  \bibinfo{pages}{999--1034}.
\newblock


\bibitem[\protect\citeauthoryear{Babaioff and Winter}{Babaioff and
  Winter}{2014}]%
        {babaioff2014contract}
\bibfield{author}{\bibinfo{person}{Moshe Babaioff} {and} \bibinfo{person}{Eyal
  Winter}.} \bibinfo{year}{2014}\natexlab{}.
\newblock \showarticletitle{Contract complexity.}
\newblock \bibinfo{journal}{\emph{EC}}  \bibinfo{volume}{14}
  (\bibinfo{year}{2014}), \bibinfo{pages}{911}.
\newblock


\bibitem[\protect\citeauthoryear{Bastani, Bayati, Braverman, Gummadi, and
  Johari}{Bastani et~al\mbox{.}}{2016}]%
        {bastani2016analysis}
\bibfield{author}{\bibinfo{person}{Hamsa Bastani}, \bibinfo{person}{Mohsen
  Bayati}, \bibinfo{person}{Mark Braverman}, \bibinfo{person}{Ramki Gummadi},
  {and} \bibinfo{person}{Ramesh Johari}.} \bibinfo{year}{2016}\natexlab{}.
\newblock \showarticletitle{Analysis of medicare pay-for-performance
  contracts}.
\newblock \bibinfo{journal}{\emph{Available at SSRN 2839143}}
  (\bibinfo{year}{2016}).
\newblock


\bibitem[\protect\citeauthoryear{Bertsimas and Tsitsiklis}{Bertsimas and
  Tsitsiklis}{1997}]%
        {bertsimas1997introduction}
\bibfield{author}{\bibinfo{person}{Dimitris Bertsimas} {and}
  \bibinfo{person}{John~N Tsitsiklis}.} \bibinfo{year}{1997}\natexlab{}.
\newblock \bibinfo{booktitle}{\emph{Introduction to linear optimization}}.
  Vol.~\bibinfo{volume}{6}.
\newblock \bibinfo{publisher}{Athena Scientific Belmont, MA}.
\newblock


\bibitem[\protect\citeauthoryear{Bolton, Dewatripont, et~al\mbox{.}}{Bolton
  et~al\mbox{.}}{2005}]%
        {bolton2005contract}
\bibfield{author}{\bibinfo{person}{Patrick Bolton}, \bibinfo{person}{Mathias
  Dewatripont}, {et~al\mbox{.}}} \bibinfo{year}{2005}\natexlab{}.
\newblock \bibinfo{booktitle}{\emph{Contract theory}}.
\newblock \bibinfo{publisher}{MIT press}.
\newblock


\bibitem[\protect\citeauthoryear{Carroll}{Carroll}{2015}]%
        {carroll2015robustness}
\bibfield{author}{\bibinfo{person}{Gabriel Carroll}.}
  \bibinfo{year}{2015}\natexlab{}.
\newblock \showarticletitle{Robustness and linear contracts}.
\newblock \bibinfo{journal}{\emph{American Economic Review}}
  \bibinfo{volume}{105}, \bibinfo{number}{2} (\bibinfo{year}{2015}),
  \bibinfo{pages}{536--63}.
\newblock


\bibitem[\protect\citeauthoryear{Castiglioni, Marchesi, and Gatti}{Castiglioni
  et~al\mbox{.}}{2021}]%
        {castiglioni2021bayesian}
\bibfield{author}{\bibinfo{person}{Matteo Castiglioni},
  \bibinfo{person}{Alberto Marchesi}, {and} \bibinfo{person}{Nicola Gatti}.}
  \bibinfo{year}{2021}\natexlab{}.
\newblock \bibinfo{title}{Bayesian Agency: Linear versus Tractable Contracts}.
\newblock
\newblock
\showeprint[arxiv]{cs.GT/2106.00319}


\bibitem[\protect\citeauthoryear{Castiglioni, Marchesi, and Gatti}{Castiglioni
  et~al\mbox{.}}{2022a}]%
        {castiglioni2022bayesian}
\bibfield{author}{\bibinfo{person}{Matteo Castiglioni},
  \bibinfo{person}{Alberto Marchesi}, {and} \bibinfo{person}{Nicola Gatti}.}
  \bibinfo{year}{2022}\natexlab{a}.
\newblock \showarticletitle{Bayesian Persuasion Meets Mechanism Design: Going
  Beyond Intractability with Type Reporting}.
\newblock \bibinfo{journal}{\emph{arXiv preprint arXiv:2202.00605}}
  (\bibinfo{year}{2022}).
\newblock


\bibitem[\protect\citeauthoryear{Castiglioni, Marchesi, and Gatti}{Castiglioni
  et~al\mbox{.}}{2022b}]%
        {CastiglioniEC22}
\bibfield{author}{\bibinfo{person}{Matteo Castiglioni},
  \bibinfo{person}{Alberto Marchesi}, {and} \bibinfo{person}{Nicola Gatti}.}
  \bibinfo{year}{2022}\natexlab{b}.
\newblock \showarticletitle{Designing Menus of Contracts Efficiently: The Power
  of Randomization}. In \bibinfo{booktitle}{\emph{Proceedings of the 23rd ACM
  Conference on Economics and Computation}} \emph{(\bibinfo{series}{EC '22})}.
  \bibinfo{publisher}{Association for Computing Machinery},
  \bibinfo{address}{New York, NY, USA}, \bibinfo{pages}{705–735}.
\newblock
\showISBNx{9781450391504}
\urldef\tempurl%
\url{https://doi.org/10.1145/3490486.3538270}
\showDOI{\tempurl}


\bibitem[\protect\citeauthoryear{Cong and He}{Cong and He}{2019}]%
        {cong2019blockchain}
\bibfield{author}{\bibinfo{person}{Lin~William Cong} {and}
  \bibinfo{person}{Zhiguo He}.} \bibinfo{year}{2019}\natexlab{}.
\newblock \showarticletitle{Blockchain disruption and smart contracts}.
\newblock \bibinfo{journal}{\emph{The Review of Financial Studies}}
  \bibinfo{volume}{32}, \bibinfo{number}{5} (\bibinfo{year}{2019}),
  \bibinfo{pages}{1754--1797}.
\newblock


\bibitem[\protect\citeauthoryear{Duetting, Ezra, Feldman, and
  Kesselheim}{Duetting et~al\mbox{.}}{2021}]%
        {duetting2021combinatorial}
\bibfield{author}{\bibinfo{person}{Paul Duetting}, \bibinfo{person}{Tomer
  Ezra}, \bibinfo{person}{Michal Feldman}, {and} \bibinfo{person}{Thomas
  Kesselheim}.} \bibinfo{year}{2021}\natexlab{}.
\newblock \showarticletitle{Combinatorial Contracts}.
\newblock \bibinfo{journal}{\emph{arXiv preprint arXiv:2109.14260}}
  (\bibinfo{year}{2021}).
\newblock


\bibitem[\protect\citeauthoryear{D{\"u}tting, Roughgarden, and
  Cohen}{D{\"u}tting et~al\mbox{.}}{2020}]%
        {dutting2020complexity}
\bibfield{author}{\bibinfo{person}{Paul D{\"u}tting}, \bibinfo{person}{Tim
  Roughgarden}, {and} \bibinfo{person}{Inbal-Talgam Cohen}.}
  \bibinfo{year}{2020}\natexlab{}.
\newblock \showarticletitle{The complexity of contracts}. In
  \bibinfo{booktitle}{\emph{Proceedings of the Fourteenth Annual ACM-SIAM
  Symposium on Discrete Algorithms}}. SIAM, \bibinfo{pages}{2688--2707}.
\newblock


\bibitem[\protect\citeauthoryear{D{\"u}tting, Roughgarden, and
  Talgam-Cohen}{D{\"u}tting et~al\mbox{.}}{2019}]%
        {dutting2019simple}
\bibfield{author}{\bibinfo{person}{Paul D{\"u}tting}, \bibinfo{person}{Tim
  Roughgarden}, {and} \bibinfo{person}{Inbal Talgam-Cohen}.}
  \bibinfo{year}{2019}\natexlab{}.
\newblock \showarticletitle{Simple versus optimal contracts}. In
  \bibinfo{booktitle}{\emph{Proceedings of the 2019 ACM Conference on Economics
  and Computation}}. \bibinfo{pages}{369--387}.
\newblock


\bibitem[\protect\citeauthoryear{Grossman and Hart}{Grossman and Hart}{1983}]%
        {grossman1983analysis}
\bibfield{author}{\bibinfo{person}{Sanford~J Grossman} {and}
  \bibinfo{person}{Oliver~D Hart}.} \bibinfo{year}{1983}\natexlab{}.
\newblock \showarticletitle{An Analysis of the Principal-Agent Problem}.
\newblock \bibinfo{journal}{\emph{Econometrica}} \bibinfo{volume}{51},
  \bibinfo{number}{1} (\bibinfo{year}{1983}), \bibinfo{pages}{7--46}.
\newblock


\bibitem[\protect\citeauthoryear{Guruganesh, Schneider, and Wang}{Guruganesh
  et~al\mbox{.}}{2021}]%
        {guruganesh2020contracts}
\bibfield{author}{\bibinfo{person}{Guru Guruganesh}, \bibinfo{person}{Jon
  Schneider}, {and} \bibinfo{person}{Joshua~R. Wang}.}
  \bibinfo{year}{2021}\natexlab{}.
\newblock \showarticletitle{Contracts under Moral Hazard and Adverse
  Selection}. In \bibinfo{booktitle}{\emph{Proceedings of the 22nd ACM
  Conference on Economics and Computation}} \emph{(\bibinfo{series}{EC '21})}.
  \bibinfo{publisher}{Association for Computing Machinery},
  \bibinfo{address}{New York, NY, USA}, \bibinfo{pages}{563–582}.
\newblock
\showISBNx{9781450385541}
\urldef\tempurl%
\url{https://doi.org/10.1145/3465456.3467637}
\showDOI{\tempurl}


\bibitem[\protect\citeauthoryear{Ho, Slivkins, and Vaughan}{Ho
  et~al\mbox{.}}{2016}]%
        {ho2016adaptive}
\bibfield{author}{\bibinfo{person}{Chien-Ju Ho}, \bibinfo{person}{Aleksandrs
  Slivkins}, {and} \bibinfo{person}{Jennifer~Wortman Vaughan}.}
  \bibinfo{year}{2016}\natexlab{}.
\newblock \showarticletitle{Adaptive contract design for crowdsourcing markets:
  Bandit algorithms for repeated principal-agent problems}.
\newblock \bibinfo{journal}{\emph{Journal of Artificial Intelligence Research}}
   \bibinfo{volume}{55} (\bibinfo{year}{2016}), \bibinfo{pages}{317--359}.
\newblock


\bibitem[\protect\citeauthoryear{Holmstrom and Milgrom}{Holmstrom and
  Milgrom}{1991}]%
        {holmstrom1991multitask}
\bibfield{author}{\bibinfo{person}{Bengt Holmstrom} {and} \bibinfo{person}{Paul
  Milgrom}.} \bibinfo{year}{1991}\natexlab{}.
\newblock \showarticletitle{Multitask principal-agent analyses: Incentive
  contracts, asset ownership, and job design}.
\newblock \bibinfo{journal}{\emph{Journal of Law, Economics, \& Organization}}
  \bibinfo{volume}{7} (\bibinfo{year}{1991}), \bibinfo{pages}{24}.
\newblock


\bibitem[\protect\citeauthoryear{Laffont and Martimort}{Laffont and
  Martimort}{2009}]%
        {laffont2009theory}
\bibfield{author}{\bibinfo{person}{Jean-Jacques Laffont} {and}
  \bibinfo{person}{David Martimort}.} \bibinfo{year}{2009}\natexlab{}.
\newblock \bibinfo{booktitle}{\emph{The theory of incentives: the
  principal-agent model}}.
\newblock \bibinfo{publisher}{Princeton university press}.
\newblock


\bibitem[\protect\citeauthoryear{Mas-Colell, Whinston, Green,
  et~al\mbox{.}}{Mas-Colell et~al\mbox{.}}{1995}]%
        {mas1995microeconomic}
\bibfield{author}{\bibinfo{person}{Andreu Mas-Colell},
  \bibinfo{person}{Michael~Dennis Whinston}, \bibinfo{person}{Jerry~R Green},
  {et~al\mbox{.}}} \bibinfo{year}{1995}\natexlab{}.
\newblock \bibinfo{booktitle}{\emph{Microeconomic theory}}.
  Vol.~\bibinfo{volume}{1}.
\newblock \bibinfo{publisher}{Oxford university press New York}.
\newblock


\bibitem[\protect\citeauthoryear{Rogerson}{Rogerson}{1985}]%
        {rogerson1985repeated}
\bibfield{author}{\bibinfo{person}{William~P Rogerson}.}
  \bibinfo{year}{1985}\natexlab{}.
\newblock \showarticletitle{Repeated moral hazard}.
\newblock \bibinfo{journal}{\emph{Econometrica: Journal of the Econometric
  Society}} (\bibinfo{year}{1985}), \bibinfo{pages}{69--76}.
\newblock


\bibitem[\protect\citeauthoryear{Shavell}{Shavell}{1979}]%
        {shavell1979risk}
\bibfield{author}{\bibinfo{person}{Steven Shavell}.}
  \bibinfo{year}{1979}\natexlab{}.
\newblock \showarticletitle{Risk sharing and incentives in the principal and
  agent relationship}.
\newblock \bibinfo{journal}{\emph{The Bell Journal of Economics}}
  (\bibinfo{year}{1979}), \bibinfo{pages}{55--73}.
\newblock


\bibitem[\protect\citeauthoryear{Shoham and Leyton-Brown}{Shoham and
  Leyton-Brown}{2008}]%
        {shoham2008multiagent}
\bibfield{author}{\bibinfo{person}{Yoav Shoham} {and} \bibinfo{person}{Kevin
  Leyton-Brown}.} \bibinfo{year}{2008}\natexlab{}.
\newblock \bibinfo{booktitle}{\emph{Multiagent systems: Algorithmic,
  game-theoretic, and logical foundations}}.
\newblock \bibinfo{publisher}{Cambridge University Press}.
\newblock


\bibitem[\protect\citeauthoryear{Trevisan}{Trevisan}{2001}]%
        {Trevisan2001Independent}
\bibfield{author}{\bibinfo{person}{Luca Trevisan}.}
  \bibinfo{year}{2001}\natexlab{}.
\newblock \showarticletitle{Non-Approximability Results for Optimization
  Problems on Bounded Degree Instances}. In
  \bibinfo{booktitle}{\emph{Proceedings of the Thirty-Third Annual ACM
  Symposium on Theory of Computing}} \emph{(\bibinfo{series}{STOC '01})}.
  \bibinfo{publisher}{Association for Computing Machinery},
  \bibinfo{address}{New York, NY, USA}, \bibinfo{pages}{453–461}.
\newblock
\showISBNx{1581133499}
\urldef\tempurl%
\url{https://doi.org/10.1145/380752.380839}
\showDOI{\tempurl}


\bibitem[\protect\citeauthoryear{Ye}{Ye}{1990}]%
        {ye1990recovering}
\bibfield{author}{\bibinfo{person}{Yinyu Ye}.} \bibinfo{year}{1990}\natexlab{}.
\newblock \showarticletitle{Recovering optimal basic variables in Karmarkar's
  polynomial algorithm for linear programming}.
\newblock \bibinfo{journal}{\emph{Mathematics of operations research}}
  \bibinfo{volume}{15}, \bibinfo{number}{3} (\bibinfo{year}{1990}),
  \bibinfo{pages}{564--572}.
\newblock


\end{thebibliography}

\newpage
\appendix
\section{Proofs Omitted From Section~\ref{sec:ptas}}

\lemmaOne*

\begin{proof}
	Given $\left( p^\theta, a^\theta \right)_{\theta \in \Theta}$ as in Definition~\ref{def:apx_menu}, we let $\widehat P =\left( \hat p^\theta \right)_{\theta \in \Theta}$ be an auxiliary menu of deterministic contracts such that $\hat p^\theta_\omega = \left( 1-\sqrt{\epsilon} \right) \, p^\theta_\omega+\sqrt{\epsilon} r_\omega$ for every $\theta \in \Theta$ and $\omega \in \Omega$.
	Then, we define the menu of contracts $\overline{P} = \left( \bar p^\theta \right)_{\theta \in \Theta}$ as follows.
	For every agent's type $\theta\in \Theta$, we let $\bar p^\theta = \hat p^{\theta'}$ for some $\theta' \in \argmax_{\theta' \in \Theta } U^\theta \left( \hat p^{\theta'} \right)$, where, for the ease of presentation, we define $U^\theta \left( \hat p^{\theta'} \right) \coloneqq \max_{a \in A} \left\{ \sum_{\omega \in \Omega} F_{\theta,a,\omega} \, \hat p^{\theta'}_\omega-c_{\theta,a} \right\}$ as type $\theta$'s utility in any IC action under $\hat p^{\theta'} $.

	Notice that $\overline{P} = \left( \bar p^\theta \right)_{\theta \in \Theta}$ is DSIC by definition.
	Next, we show that, for each agent's type $\theta \in \Theta$, the principal's expected utility contribution due to that type under contract $\bar p^\theta$ decreases by a at most $2\sqrt{\epsilon}$ compared with that obtained when the agent plays $a^\theta$ under $p^\theta$, proving the lemma.
	%
	
	For every agent's type $\theta \in \Theta$, let $\theta' \in \Theta$ be the agent's type such that $\bar p^\theta=\hat p^{\theta'}$.
	Notice that it may be the case that $\theta = \theta'$ and/or $b^\theta\left( \bar p^\theta \right)= a^\theta$.
	%
	%
	First, we prove the following:
	\begin{align*}
	\sum_{\omega \in \Omega} F_{\theta,b^\theta( \bar p^\theta),\omega}  \left( (1-\sqrt{\epsilon}) \, p^{\theta'}_\omega+\sqrt{\epsilon} r_\omega \right)-c_{\theta,b^\theta( \bar p^\theta)} &= \sum_{\omega \in \Omega} F_{\theta,b^\theta( \bar p^\theta),\omega} \, \bar p^\theta -c_{\theta,b^\theta( \bar p^\theta)} \\
	&\ge \sum_{\omega \in \Omega} F_{\theta,a^\theta,\omega}  \hat p^\theta -c_{\theta,a^\theta} \\
	&= \sum_{\omega \in \Omega} F_{\theta,a^\theta,\omega}  \left( (1-\sqrt{\epsilon}) \, p^\theta+\sqrt{\epsilon} r_\omega \right) -c_{\theta,a^\theta} ,
	\end{align*}
	where the inequality follows from the fact that playing action $b^\theta\left( \bar p^\theta \right)$ under contract $\bar p^\theta$ is preferred by the agent over playing action $a^\theta$ under contract $\hat p^{\theta}$, by definition of contract $\bar p^\theta$.
	Moreover, by definition of the $\epsilon$-approximate menu of deterministic contracts $\left( p^\theta, a^\theta \right)_{\theta \in \Theta}$, it holds
	\[
	\sum_{\omega \in \Omega} F_{\theta,a^\theta,\omega} \, p^\theta -c_{\theta,a^\theta} \ge \sum_{\omega \in \Omega} F_{\theta,b^\theta(p^{\theta'}),\omega}  \, p^{\theta'} -c_{\theta,b^\theta(p^{\theta'})} -\epsilon\ge \sum_{\omega \in \Omega} F_{\theta,b^\theta(\bar p^\theta),\omega}  \, p^{\theta'} -c_{\theta,b^\theta(\bar p^\theta)} -\epsilon.
	\]
	By summing the two relations obtained above, we get
	\[
	\sum_{\omega \in \Omega}  F_{\theta,b^\theta( \bar p^\theta), \omega} \left( -\sqrt{\epsilon} \, p^{\theta'}_{\omega}+ \sqrt{\epsilon} r_\omega \right) \ge \sum_{\omega \in \Omega}  F_{\theta,a^\theta,\omega} \left(  -\sqrt{\epsilon} \, p^\theta_\omega+\sqrt{\epsilon} r_\omega \right)-\epsilon,
	\]
	which implies that 
	\[
	\sum_{\omega \in \Omega}  F_{\theta,b^\theta(\bar p^\theta)} \left(   r_\omega - p^{\theta'}_{\omega} \right) \ge \sum_{\omega \in \Omega}  F_{\theta,a^\theta,\omega} \left(   r_\omega- p^\theta_\omega \right)-\sqrt\epsilon.
	\]

	Then, the principal's expected utility when the agent is of type is $\theta \in \Theta$ is equal to
	%
	\begin{align*}
	\sum_{\omega \in \Omega} F_{\theta,b^\theta( \bar p^\theta),\omega}  \left( r_\omega-\bar p^\theta_\omega \right)&=
	\sum_{\omega \in \Omega} F_{\theta,b^\theta( \bar p^\theta),\omega}  \left( r_\omega -(1-\sqrt{\epsilon}) \, p^{\theta'} -\sqrt{\epsilon} r_\omega \right)\\
	&= \left( 1-\sqrt{\epsilon} \right) \sum_{\omega \in \Omega} F_{\theta,b^\theta( \bar p^\theta),\omega} \left( r_\omega - p^{\theta'} \right) \\
	&\ge \left( 1-\sqrt{\epsilon} \right) \left( \sum_{\omega \in \Omega}  F_{\theta,a^\theta,\omega} \left(  r_\omega- p^\theta_\omega \right)- \sqrt \epsilon \right)\\
	&\ge \sum_{\omega \in \Omega}  F_{\theta,a^\theta,\omega} \left(   r_\omega- p^\theta_\omega \right)-2 \sqrt{\epsilon}.
	\end{align*}
	%
	This concludes the proof.
\end{proof}

\lemmaTwo*

\begin{proof}
	By contradiction, we show that, if $\sum_{\theta \in \Theta(L,P)} \mu_\theta > \frac{1}{L}$, then the menu $P = \left( p^\theta \right)_{\theta \in \Theta}$ provides a negative expected utility to the principal.
	In particular, if $\sum_{\theta \in \Theta(L,P)} \mu_\theta > \frac{1}{L}$, then the expected payment due to agent's types in $\Theta(L,P)$ is greater than $1$, while the expected reward to the principal is at most $1$.
	Hence, since a menu of all-zero contract vectors provides the principal with an expected utility at least of zero, the menu $P$ cannot be optimal, reaching a contradiction.
	%
	%
\end{proof}

\lemmaThree*

\begin{proof}
	%
	The proof is organized in two steps.
	As a first step, we show that, from an optimal menu of deterministic contracts $P = \left( p^\theta \right)_{\theta \in \Theta}$, we can build a $\frac{\delta^2}{16}$--approximate menu of deterministic contracts $\left( \bar p^\theta, \bar a^\theta \right)_{\theta \in \Theta}$ that provides the principal with an expected utility of $OPT-\frac{\delta}{2}$ and such that it adopts a limited number of different payment values.
	In particular, all payment values $\bar p^\theta_\omega$, for every agent's type $\theta \in \Theta$ and outcome $\omega \in \Omega$, belong to a small, finite set.
	This implies that also the set of possible different contracts is small.
	Then, as a second step, we employ Lemma~\ref{lm:toIC} to build a DSIC menu of deterministic contracts as in the statement of the lemma, starting from $\left( \bar p^\theta, \bar a^\theta \right)_{\theta \in \Theta}$, and only incurring in a small loss in the principal's expected utility.
	
	Let $M(\omega) \coloneqq \max_{\theta \in \Theta}p^\theta_\omega$ for every $\omega \in \Omega$, and let $I \coloneqq \left\{ 0,1,\dots,\left\lceil\frac{\log \eta}{\log(1-\eta)}\right\rceil\right\}$ with $\eta \coloneqq \frac{\delta^3}{64m}$.
	We split $\Theta$ into two sets.
	Letting $L \coloneqq \frac{4}{\delta}$, the first set is $\Theta_1 \coloneqq \Theta(L,P)$, while the second one is $\Theta_2 \coloneqq \Theta\setminus \Theta_1$, which includes all the agent's types that are \emph{not} in $\Theta_1$.
	Moreover, we let $\hat \Theta\subseteq \Theta$ be the subset of agent's types $\theta \in \Theta$ such that there exists an action $a^\theta \in A$ and an outcome $\omega^\theta \in \Omega$ that simultaneously satisfy $p^\theta_{\omega^\theta}< \eta M(\omega^\theta)$ and $ F_{\theta,a^\theta,\omega^\theta} \, p^\theta_{\omega^{\theta}}\ge \frac{\delta^2}{8m}$.

	First, we prove that $\hat \Theta\subseteq \Theta_1$.
	For every $\theta \in \hat \Theta$, we have $ F_{\theta,a^\theta,\omega^\theta} \ge \frac{1}{p^\theta_{\omega^\theta}}\frac{\delta^2}{8m} >  \frac{\delta^2}{8m}\frac{1}{\eta M(\omega)}= \frac{8}{\delta M(\omega)} $.
	Then, by IC conditions and the definition of $M(\omega^\theta)$, it holds 
	\[\sum_{\omega \in \Omega} F_{\theta,b^{\theta}(p^\theta),\omega} \, p^\theta_\omega -c_{\theta,b^{\theta}(p^\theta)} \ge F_{\theta,a^\theta,\omega^\theta} M(\omega^\theta) - c_{\theta,a^\theta}, \]
	which implies that 
	\[\sum_{\omega \in \Omega} F_{\theta,b^{\theta}(p^\theta),\omega} \, p^\theta_\omega \ge F_{\theta,a^\theta,\omega^\theta} M(\omega^\theta)-1 \ge  \frac{8}{\delta}-1\ge \frac{4}{\delta}.\]

	Next, we build a $\frac{\delta^2}{16}$--approximate menu of deterministic contracts $\left( \bar p^\theta, \bar a^\theta \right)_{\theta \in \Theta}$, as follows.
	For each agent's type $\theta \in \Theta_2$ and outcome $\omega \in \Omega$, we set $\bar p^\theta_\omega=0$ if $ p^\theta_\omega < (1-\eta)^{\left\lceil\frac{\log \eta}{\log(1-\eta)}\right\rceil} M(\omega)\le \eta  M(\omega)$, while we set $\bar p^\theta_\omega=(1-\eta)^{i^*} M(\omega)$ otherwise, where $i^*$ is the smallest integer such that $(1-\eta)^{i^*} M(\omega)\le   p^\theta_\omega$.
	Notice that $i^*$ is at most $\left\lceil \frac{\log \eta}{\log(1-\eta)} \right\rceil$, and, thus, $i^*\in I$.
	Intuitively, contracts $\bar p^\theta$ are defined by lowering the payments of the corresponding contracts $p^\theta$ to the nearest value $(1-\eta)^i$ with $i \in I$, so that, in the second case above, it holds that $ \bar p^\theta_\omega \ge (1-\eta)  p^\theta_\omega$.
	Moreover, we set $\bar a^\theta=b^\theta \left( p^\theta \right)$ for all $\theta \in \Theta_2$.
	Finally, for each agent's type $\theta \in \Theta_1$, we set $\bar p^\theta=\bar p^{g(\theta)}$ and $\bar a^\theta=b^{\theta}\left( \bar p^{\theta} \right)$, where, for the ease of notation, we let $g(\theta) \in \argmax_{\hat \theta \in \Theta_2} \left\{ \sum_{\omega \in \Omega} F_{\theta,b^\theta(\bar p^{\hat\theta}) ,\omega} \, \bar p^{\hat \theta}_\omega-c_{\theta,b^\theta(\bar p^{\hat \theta})} \right\}$.
	
	First, we show that the menu $\left( \bar p^\theta, \bar a^\theta \right)_{\theta \in \Theta}$ is $\frac{\delta^2}{16}$--approximate.
	This holds by construction for all the agent's types in $\Theta_1$.
	For any type $\theta \in \Theta_2$, as a first step, we prove that
	\begin{align}\label{eq:smallPayment}
	& F_{\theta,\bar a^{\theta},\omega} \, \bar p^\theta_\omega \ge F_{\theta,\bar a^\theta,\omega}  \, p^\theta_\omega-\frac{\delta^2}{16m} & \forall \omega \in \Omega. 
	\end{align}
	%
	We consider two cases.
	If $F_{\theta,\bar a^{ \theta},\omega} \,  p^\theta_\omega< \frac{\delta^2}{16m}$, then $F_{\theta,\bar a^\theta,\omega}  \, p^\theta_\omega - \frac{\delta^2}{16m} < 0$, and, thus, the inequality above holds since $F_{\theta, \bar a^{\theta},\omega} \, \bar p^\theta_\omega\ge 0$.
	Otherwise, since $\theta\notin \hat \Theta \subseteq \Theta_1$, it holds $p^\theta_\omega \ge \eta M(\omega)$ and
	\[
	F_{\theta,\bar a^{\theta},\omega} \, \bar p^\theta_\omega \ge F_{\theta,\bar a^{\theta},\omega} (1-\eta) \, p^\theta_\omega \ge F_{\theta,\bar a^{\theta},\omega} \, p^\theta_\omega - \eta F_{\theta,\bar a^{\theta},\omega} \, p^\theta_\omega \ge F_{\theta,\bar a^{\theta},\omega} \, p^\theta_\omega - \eta \frac{4}{\delta}=F_{\theta,\bar a^{\theta},\omega} \, p^\theta_\omega-\frac{\delta^2}{16m},
	\]
	where the last inequality follows from $F_{\theta,\bar a^{\theta},\omega} \, p^\theta_\omega\le \sum_\omega F_{\theta,\bar a^{\theta},\omega} \, p^\theta_\omega \le \frac{4}{\delta}$ since $\theta\notin \Theta_1$.
	Then, we prove that $\left( \bar p^\theta, \bar a^\theta \right)_{\theta \in \Theta}$ satisfies Equation~\eqref{eq:apxIc} for every agent's type $\theta \in \Theta_2$.
	Notice that, since for each contract $\bar p^\theta$ with $\theta \in \Theta_1$ there exists a contract $\bar p^{g(\theta)}$ with $g(\theta) \in \Theta_2$ such that $\bar p^\theta = \bar p^{g(\theta)}$, the condition in Equation~\eqref{eq:apxIc} must be checked only for types $\hat \theta \in \Theta$.
	In particular, for every pair of agent's types $\theta, \hat \theta \in \Theta_2$, the following holds:
	\begin{align*}
	\sum_{\omega \in \Omega} F_{\theta,\bar a^{\theta},\omega} \, \bar p^\theta_\omega -c_{\theta, \bar a^{\theta}} & \ge \sum_{\omega \in \Omega} \left( F_{\theta,\bar a^{\theta},\omega} \, p^\theta_\omega -  \frac{\delta^2}{16m} \right) -c_{\theta,\bar a^{\theta}} \\
	&= \sum_{\omega \in \Omega} F_{\theta,b^\theta(p^\theta),\omega}   \, p^\theta_\omega  -c_{\theta,b^\theta(p^\theta)} -\frac{\delta^2}{16}\\
	&\ge\sum_{\omega \in \Omega} F_{\theta,b^{ \theta}(p^{\hat \theta}),\omega}  p^{\hat \theta}_\omega -c_{\theta,b^{ \theta}(p^{\hat \theta})} - \frac{\delta^2}{16} \\
	&\ge\sum_{\omega \in \Omega} F_{\theta,b^{ \theta}(\bar p^{\hat \theta}),\omega}  p^{\hat \theta}_\omega -c_{\theta,b^{ \theta}(\bar p^{\hat \theta})} - \frac{\delta^2}{16} \\
	&\ge \sum_{\omega \in \Omega} F_{\theta,b^{ \theta}(\bar p^{\hat \theta}),\omega} \, \bar p^{\hat \theta}_\omega -c_{\theta,b^{ \theta}(\bar p^{\hat \theta})} - \frac{\delta^2}{16} ,
	\end{align*}
	where the first inequality follows from Equation~\eqref{eq:smallPayment}, the second one holds by the DSIC property of $P = \left(p^\theta\right)_{\theta \in \Theta}$, the third one from the definition of best response, while the last inequality follows from the fact that $\bar p^{\theta'}_\omega\le  p^{\theta'}_\omega $ for every $\theta \in \Theta_2$ and $\omega \in \Omega$.
	
	Now, we prove that the principal's expected utility in the $\frac{\delta^2}{16}$--approximate menu of deterministic contracts $\left( \bar p^\theta, \bar a^\theta \right)_{\theta \in \Theta}$ is at least $OPT-\frac{\delta}{2}$.
	For every agent's type $\theta \in \Theta_1$, we have that
	\begin{align*}
	\sum_{\omega \in \Omega} F_{\theta,\bar a^\theta,\omega}  \, \bar p^\theta_\omega - c_{\theta,\bar a^\theta} &=  \sum_{\omega \in \Omega} F_{\theta,  b^\theta(\bar p^{g(\theta)}),\omega}  \, \bar p^{g(\theta)}_\omega - c_{\theta, b^\theta(\bar p^{g(\theta)})}  \\
	&\le \sum_{\omega \in \Omega} F_{\theta,  b^\theta(\bar p^{g(\theta)}),\omega}   \, p^{g(\theta)}_\omega - c_{\theta, b^\theta(\bar p^{g(\theta)})} \\
	&\le   \sum_{\omega \in \Omega} F_{\theta,b^\theta(p^\theta),\omega} \, p^\theta_\omega - c_{\theta,b^\theta(p^\theta)},
	\end{align*}
	%
	where the first inequality follows from the fact that $\bar p^\theta_\omega\le p^\theta_\omega$ for every $\theta \in \Theta_2$ and $\omega\in \Omega$, while the second one from the IC e DSIC properties of $P = \left(p^\theta\right)_{\theta \in \Theta}$.
	Thus, we can conclude that $\sum_{\theta \in \Theta} F_{\theta,\bar a^\theta,\omega} \, \bar  p^\theta_\omega \le \sum_{\theta \in \Theta} F_{\theta,b^\theta(p^\theta),\omega} \, p^\theta_\omega +1$.
	Moreover, the principal's expected reward of contract $p^\theta$ is $\sum_{\theta \in \Theta} F_{\theta,b^\theta(p^\theta),\omega} r_\omega \le 1$.
	Thus, the revenue that the principal extracts from an agent of type $\theta \in \Theta_1$ decreases by at most $2$, since, by changing the contract from $p^\theta$ to $\bar p^\theta$, the reward decreases by at most $1$ and the payment increases by at most $1$.
	On the other hand, for every agent's type $\theta \in \Theta_2$, the principal's expected utility does \emph{not} decrease, since the agent plays the same action and the payments decrease.
	Hence, the principal's expected utility decreases by at most $2 \sum_{\theta \in \Theta_1} \mu_\theta \le 2 \frac{\delta}{4}= \frac{\delta}{2}$, where the inequality follows from Lemma~\ref{lm:smallCosts}.

	Finally, by Lemma~\ref{lm:toIC}, starting from the $\frac{\delta^2}{16}$-approximate menu of deterministic contracts $\left( \bar p^\theta, \bar a^\theta \right)_{\theta \in \Theta}$, we can build a DSCI menu of deterministic contracts $\hat P = \left( \hat p^\theta \right)_{\theta \in \Theta}$ providing the principal with a utility at least of $APX-2\sqrt{\frac{\delta^2}{16}}=APX-\frac{\delta}{2}$, where $APX$ is the expected utility of $\left( \bar p^\theta, \bar a^\theta \right)_{\theta \in \Theta}$.
	Hence, the principal's expected utility is at least $APX-\frac{\delta}{2}=OPT-\delta$.
	Moreover, we also have that $\hat p^\theta_\omega \in \left\{ (1-\eta)^i M(\omega) \left( 1-\frac{\delta}{4} \right)+ \frac{\delta}{4} r_\omega \right\}_{i \in I } \cup \left\{ \frac{\delta}{4 } r_\omega\right\}$ for every $\theta \in \Theta$ and $\omega \in \Omega$.
	This holds since the menu of contracts built in Lemma~\ref{lm:toIC} uses contracts defined as $(1-\sqrt{\epsilon})\bar  p^\theta+\sqrt \epsilon r$ for $\theta \in \Theta$.
	Thus, there are at most
	\[
	|I+1|^m=\left( \frac{\log \frac{\delta^3}{64m}}{\log \left(1-\frac{\delta^3}{64m} \right)}+3 \right)^m\le \left( \frac{64m}{\delta^3} \log \frac{64m}{\delta^3} \right)^m = O\left(\left( \frac{m}{\delta^3} \log \frac{m}{\delta} \right)^m\right)
	\]
	different contracts, where in the first inequality we use $-\log(1-x)\ge x$ for $x\le 1$.
\end{proof}

\theoremTwo*

\begin{proof}
	Let $\delta > 0$ be a desired additive approximation factor for the PTAS.
	By Lemma~\ref{lm:smallSup}, in order to find an approximate DSIC menu of deterministic contracts, it is sufficient to optimize over all the menus using at most $k=k(\delta,m)$ different contracts.
	%
	%
	Recall that a menu that uses at most $k$ different contracts can be represented by a matrix $T \in \mathcal{T}	= \mathbb{R}_+^{k \times m}$ and a function $f:\Theta\rightarrow \{1,\dots,k\}$.
	
	The PTAS works by organizing the space of possible matrices representing menus of deterministic contracts, namely $\mathbb{R}_+^{k \times m}$, into polyhedra, depending on which function $f$ can be paired with the matrix to define a DSIC menu.
	In particular, given a function $f:\Theta\rightarrow \{1,\dots,k\}$, we define $\mathcal{Q}^f\subseteq \mathbb{R}_+^{k \times m}$ as the polyhedra defined by the following set of linear inequalities enforcing DSIC:
	\begin{align*}
	& \sum_{\omega \in \Omega}  F_{\theta,b^\theta(T_{f(\theta)}),\omega} \, T_{f(\theta),\omega} -c_{\theta,b^\theta(T_{f(\theta)})}\ge \sum_{\omega \in \Omega}   F_{\theta,a,\omega} \, T_{i,\omega}-c_{\theta,a} & \forall \theta \in \Theta, \forall a \in A, \forall i \in \{1,\dots,k\}.
	\end{align*}
	
	Once function $f$ is fixed, \emph{i.e.}, we restrict the attention to menus of deterministic contracts in $\mathcal{Q}^f$, an optimal menu can be obtained as optimal solution to the following LP:
	\begin{subequations} \label{lp:outcomes}
		\begin{align}
		\max_{T \in \mathbb{R}^{k \times m}} &\sum_{\theta \in \Theta} \mu_\theta \sum_{\omega \in \Omega} F_{\theta,b^\theta(T_{f(\theta)}),\omega} \left( r_\omega -T_{f(\theta),\omega} \right) \hspace{5cm} \text{s.t.} \\
		& \sum_{\omega \in \Omega}  F_{\theta,b^\theta(T_{f(\theta)})} \, T_{f(\theta),\omega} -c_{\theta,b^\theta(T_{f(\theta)})}\ge \sum_{\omega \in \Omega}   F_{\theta,a,\omega} \, T_{i,\omega}-c_{\theta,a} \nonumber \\
		& \hspace{6.7cm} \forall \theta \in \Theta, \forall a \in A, \forall i \in \{1,\dots,k\} \label{lp:outcomes1}\\
		&  T_{i,\omega}\ge0 \hspace{6.5cm} \forall i \in \{1,\dots,k\} , \forall \omega \in \Omega. \label{lp:outcomes2}
		\end{align}
	\end{subequations}
	
	Since the objective of the LP is a linear function with non-positive coefficients and the feasible region is $\mathbb{R}^{k\times m}_{+}$, an optimal solution to the LP is one of the vertexes of the polyhedron defined by Constraints~\eqref{lp:outcomes1}~and~\eqref{lp:outcomes2}.
	In order to conclude the proof, notice that, for every possible function $f$, the feasible region of the LP is defined by a subset of the same set of $\ell nk^2+km$ possible constraints.
	In particular, for every $ \theta \in \Theta$, $a \in A$, and $i \in \{1,\dots,k\}$, there exists $k$ possible Constraints~\eqref{lp:outcomes1} depending on the value of $f(\theta)$, while Constraints~\eqref{lp:outcomes1} does \emph{not} depend on the function $f$.
	Moreover, since each vertex is at the intersection of $km$ linearly independent hyperplanes representing the inequalities and there are at most $\ell nk^2+km$ different hyperplanes, there are at most $\binom{\ell nk^2+km}{km}$ possible vertexes.
	Finally, the vertexes can be enumerated in polynomial time.
\end{proof}
\section{Proofs Omitted From Section~\ref{sec:simple}}

\lemmaBinary*

\begin{proof}
	Let $\Omega = \left\{ \omega_0, \omega_1 \right\}$ be the set made by the two outcomes.
	Moreover, w.l.o.g., let $\omega_0$ be the outcome with smaller reward, \emph{i.e.}, $r_{\omega_0}\le r_{\omega_1}$.
	We show that, given a menu of deterministic contracts $P = \left( p^\theta \right)_{\theta \in \Theta}$, we can build a contract $\hat p \in \mathbb{R}_+^m$ providing the principal with at least the same principal's expected utility.
	Let $\hat \theta \in \argmax_{\theta \in \Theta} p^\theta_{\omega_1}$.
	Then, we define the contract $\hat p$ such that $\hat p_{\omega_1}=\min\left\{ p^{\hat \theta}_{\omega_1}, r_{\omega_1}-r_{\omega_0} \right\}$ and $\hat p_{\omega_0}=0$.
	We show that, by replacing the menu $P = \left( p^\theta \right)_{\theta \in \Theta}$ with the contract $\hat p$, for each agent's type $\theta \in \Theta$, the principal's expected utility does \emph{not} decrease.

	Take any type $\theta \in \Theta$.
	As a preliminarily step, notice that 
	\begin{align}\label{eq:bestSingle}
	F_{\theta,b^\theta( p^\theta),\omega_1} \, \hat p_{\omega_1}& \le
	F_{\theta,b^\theta( p^\theta),\omega_1} \, p^{\hat \theta}_{\omega_1}+ \left( 1-F_{\theta,b^\theta( p^\theta),\omega_1} \right) \, p^{\hat \theta}_{\omega_0} \nonumber \\
	& \le  F_{\theta,b^\theta( p^\theta),\omega_1} \, p^\theta_{\omega_1}+  \left(1-F_{\theta,b^\theta( p^\theta),\omega_1} \right) \, p^\theta_{\omega_0}, 
	\end{align}
	where the first inequality follows from the definition of $\hat p$ and the second one from DSIC.
	Then, we consider two cases.
	First, suppose that $b^\theta\left(\hat p\right)=b^\theta\left( p^\theta\right)$.
	Then, the expected payment under the contract $p$ for an agent of type $\theta$ is 
	\begin{equation*}
	F_{\theta,b^\theta(\hat p),\omega_1} \, \hat p_{\omega_1} \le  F_{\theta,b^\theta( p^\theta),\omega_1} \, p^\theta_{\omega_1}+  \left( 1-F_{\theta,b^\theta( p^\theta),\omega_1} \right) \, p^\theta_{\omega_0} 
	\end{equation*}
	where the inequality comes from Equation~\eqref{eq:bestSingle} and $b^\theta\left( \hat p \right)=b^\theta\left( p^\theta \right)$.
	Hence, the payment decreases while the revenue does not, implying that the principal's expected utility does \emph{not} decrease.
	
	Second, suppose that $b^\theta\left( \hat p \right)\neq b^\theta \left( p^\theta \right)$.
	By IC conditions, we have 
	\[
	F_{\theta,b^\theta(\hat p),\omega_1} \, \hat p_{\omega_1}-c_{\theta,b^\theta(\hat p)} \ge F_{\theta,b^\theta(p^\theta),\omega_1} \, \hat p_{\omega_1}-c_{\theta,b^\theta(p^\theta)}, 
	\]
	and
	\[
	F_{\theta,b^\theta(\hat p),\omega_1} \, p^{\theta}_{\omega_1} +  \left( 1-F_{\theta,b^\theta(\hat p),\omega_1} \right) \, p^\theta_{\omega_0} -c_{\theta,b^\theta(\hat p)} \le  F_{\theta,b^\theta(p^\theta),\omega_1}\, p^\theta_{\omega_1}+  \left( 1-F_{\theta,b^\theta(p^\theta),\omega_1} \right) \, p^\theta_{\omega_0} - c_{\theta,b^\theta(p^\theta)}.
	\]
	By summing the two inequalities, we get
	\begin{equation}\label{eq:largeF}
	F_{\theta,b^\theta(\hat p),\omega_1} \, \left( \hat p_{\omega_1} -  p^\theta_{\omega_1} +p^\theta_{\omega_0} \right) \ge F_{\theta,b^\theta(p^\theta),\omega_1} \, \left( \hat p_{\omega_1} -  p^\theta_{\omega_1} +p^\theta_{\omega_0} \right). 
	\end{equation}
	We consider two cases.
	%
	If $p^\theta_{\omega_1} \le r_{\omega_1}-r_{\omega_0}$, then $\hat p_{\omega_1} \ge p^\theta_{\omega_1}$, which implies $F_{\theta,b^\theta(\hat p),\omega_1} \ge F_{\theta,b^\theta( p^\theta),\omega_1}$ by Equation~\eqref{eq:largeF}.
	Then, the principal's expected utility when the agent plays action $b^\theta \left( \hat p \right)$ is 
	\begin{align*}
	F_{\theta,b^\theta(\hat p),\omega_1} \left( r_{\omega_1}-\hat p_{\omega_1} \right)+ & \left( 1-F_{\theta,b^\theta(\hat p),\omega_1} \right) r_{\omega_0} \ge  F_{\theta,b^\theta( p^\theta),\omega_1} \left( r_{\omega_1}-\hat p_{\omega_1} \right)+ \left( 1-F_{\theta,b^\theta( p^\theta),\omega_1} \right) r_{\omega_0} \\
	& \ge F_{\theta,b^\theta( p^\theta),\omega_1}  \left( r_{\omega_1}-p^\theta_{\omega_1} \right)+  \left( 1-F_{\theta,b^\theta( p^\theta),\omega_1} \right) \left( r_{\omega_0}-p^\theta_{\omega_0} \right),
	\end{align*}
	
	where the first inequality comes from $r_{\omega_1}-\hat p_{\omega_1} \ge r_{\omega_0}$ (by construction) and $F_{\theta,b^\theta(\hat p),\omega_1} \ge F_{\theta,b^\theta( p^\theta),\omega_1}$, while the second one comes from Equation \eqref{eq:bestSingle}.
	Hence, the principal's utility does not decrease.
	
	If $p^\theta_{\omega_1} \ge r_{\omega_1}-r_{\omega_0}$, then the expected reward is 
	\begin{subequations}
		\begin{align*}
		F_{\theta,b^\theta( \hat p),\omega_1} (r_{\omega_1}-\hat p_{\omega_1}) + F_{\theta,b^\theta( \hat p),\omega_0}  r_{\omega_0} &\ge F_{\theta,b^\theta( \hat p),\omega_1} r_{\omega_0} + F_{\theta,b^\theta( \hat p),\omega_0}  r_{\omega_0} \\
		&=r_{\omega_0} \ge F_{\theta,b^\theta(  p^\theta),\omega_1} (r_{\omega_1}-p^\theta_{\omega_1}) + F_{\theta,b^\theta(  p^\theta),\omega_0}  (r_{\omega_0}-p^\theta_{\omega_0})
		\end{align*}
	\end{subequations}
	where the first inquality holds since $\hat p_{\omega_1}\le r_{\omega_1}-r_{\omega_0}$ by construction and the last inequality holds since both $r_{\omega_0}-p^\theta_{\omega_0}$ and $r_{\omega_1}-p^\theta_{\omega_1}$ are smaller than $r_{\omega_0}$.
	Thus, the principal's utility does not decrease.
	This concludes the proof.
\end{proof}

\theoremTypes*

\begin{proof}
	Consider an optimal menu of contracts $(p^\theta)_{\theta \in \Theta}$.
	This optimal menu of contracts induces an action $a^\theta=b^\theta(p^\theta)$ for each possible receiver's type $\theta \in \Theta$.
	Once the action profile $(a^\theta)_{\theta \in \Theta}$ is determined, the optimal menu of contracts can be found solving the following LP of polynomial size.
	
	\begin{subequations}
		\begin{align} 
		\min_{p \in \mathbb{R}_+^{\ell \times m}} &\sum_\theta \mu_\theta \sum_{\omega \in \Omega} p_\omega F_{\theta,a^\theta,\omega}  \\
		&\text{s.t. } \sum_{\omega \in \Omega} p^{\theta}_\omega F_{\theta,a^\theta,\omega} -c_{\theta,a^\theta} \ge  \sum_{\omega \in \Omega} p^{\theta'}_\omega F_{\theta,a,\omega} -c_{\theta,a} \,\, \forall \theta, \theta' \in \Theta, a \in A\\ 
		\end{align}
	\end{subequations}
	We can enumerate over all the actions' profiles and, for each action profile, find an optimal menu of contract that incentivize this action profile.
	Since the are $O(n^\ell)$ action profile, the algorithm works in polynomial time if the number of types is fixed.
\end{proof}

\section{Proofs Omitted From Section~\ref{sec:randomized}}

\lemmaRandOne*

\begin{proof}
	Let $\Gamma = \{ \gamma^\theta \}_{\theta \in \Theta}$ be a DSIC menu of randomized contracts, \emph{i.e.}, a feasible solution to Problem~\ref{prob:opt_rand}.
	We show that it is always possible to build a new DSIC menu of randomized contracts $\tilde\Gamma = \{ \tilde\gamma^\theta \}_{\theta \in \Theta}$ that is still feasible to Problem~\ref{prob:opt_rand}, provides at least the same expected utility, and satisfies $\left| \supp (\tilde\gamma^\theta) \cap \hat\pa^{\theta, a} \right| \leq 1$ for all $\theta \in \Theta$ and $a \in A$.
	In particular, we let each probability distribution $\tilde\gamma^\theta$ be such that the support $\supp(\tilde\gamma^\theta)$ contains all and only the contracts defined by $p^{\theta,a}\coloneqq \mathbb{E}_{p \sim \gamma^\theta} \left[ p \mid p \in \hat\pa^{\theta,a} \right]$ for all the actions $a \in A$ such that $\supp (\gamma^\theta) \cap \hat\pa^{\theta, a} \neq \varnothing$. 
	Moreover, we define $\tilde\gamma^\theta_{p^{\theta,a}} \coloneqq \Pr_{p \sim \gamma^\theta} \left\{ p\in \hat\pa^{\theta,a}  \right\}$ for these contracts.
	%
	Clearly, the menu of randomized contracts $\tilde\Gamma = \{ \tilde\gamma^\theta \}_{\theta \in \Theta}$ satisfies the required conditions on the supports.

	It remains to show that $\tilde\Gamma = \{ \tilde\gamma^\theta \}_{\theta \in \Theta}$ defines a solution to Problem~\ref{prob:opt_rand}.
	First, we prove that the principal's expected utility (the objective of Problem~\ref{prob:opt_rand}) does \emph{not} decrease with respect to that of the menu $\Gamma = \{ \gamma^\theta \}_{\theta \in \Theta}$.
	Formally:
	\begin{subequations}\label{eq:lem_finite_support}
		\begin{align}
		\sum_{\theta \in \Theta} &\mu_\theta \mathbb{E}_{p \sim \gamma^\theta}  \left[  \sum_{\omega \in \Omega} F_{\theta, b^\theta(p), \omega} r_\omega - \sum_{\omega \in \Omega} F_{\theta, b^\theta(p), \omega} p_\omega  \right] \nonumber \\
		& \hspace{-1cm}=  \sum_{\theta \in \Theta} \mu_\theta \sum_{a \in A} \Pr_{p \sim \gamma^\theta} \left\{ p\in \hat\pa^{\theta,a}  \right\} \mathbb{E}_{p \sim \gamma^\theta} \left[  \sum_{\omega \in \Omega} F_{\theta, b^\theta(p), \omega} r_\omega - \sum_{\omega \in \Omega} F_{\theta, b^\theta(p), \omega} p_\omega  \mid  p \in \hat\pa^{\theta,a}  \right] \label{eq:lem_finite_support_1} \\
		& \hspace{-1cm} = \sum_{\theta \in \Theta} \mu_\theta \sum_{a \in A} \Pr_{p \sim \gamma^\theta} \left\{ p\in \hat\pa^{\theta,a}  \right\} \left( \sum_{\omega \in \Omega} F_{\theta, a, \omega} r_\omega - \sum_{\omega \in \Omega} F_{\theta, a, \omega} \mathbb{E}_{p \sim \gamma^\theta} \left[  p_\omega  \mid  p \in \hat\pa^{\theta,a} \right] \right) \label{eq:lem_finite_support_2} \\
		& \hspace{-1cm} = \sum_{\theta \in \Theta} \mu_\theta \sum_{p \in \supp(\tilde\gamma^\theta)} \tilde\gamma^\theta_p \left( \sum_{\omega \in \Omega} F_{\theta, b^\theta(p), \omega} r_\omega - \sum_{\omega \in \Omega} F_{\theta, b^\theta(p), \omega} p_\omega \right) \label{eq:lem_finite_support_3} \\
		& \hspace{-1cm} = \sum_{\theta \in \Theta} \mu_\theta \mathbb{E}_{p \sim \tilde\gamma^\theta} \left[  \sum_{\omega \in \Omega} F_{\theta, b^\theta(p), \omega} r_\omega - \sum_{\omega \in \Omega} F_{\theta, b^\theta(p), \omega} p_\omega  \right], \label{eq:lem_finite_support_4}
		\end{align}
	\end{subequations}
	where Equation~\eqref{eq:lem_finite_support_1} holds since the sets $\hat\pa_{\theta, a}$ define a partition of $\mathbb{R}_+^m$, Equation~\eqref{eq:lem_finite_support_2} holds since for contracts $p \in \hat\pa^{\theta, a}$ it is the case that $b^\theta (p) = a$, while Equation~\eqref{eq:lem_finite_support_3} holds by definition of $\tilde\Gamma = \{ \tilde\gamma^\theta \}_{\theta \in \Theta}$.

	Finally, we need to prove that $\tilde\Gamma = \{ \tilde\gamma^\theta \}_{\theta \in \Theta}$ satisfies Constraints~\eqref{eq:opt_rand_ic} of Problem~\ref{prob:opt_rand}.
	Formally, we have that for every $\theta \neq \hat\theta\in \Theta$ it holds:
	\begingroup
	\allowdisplaybreaks
	\begin{subequations}\label{eq:lem_finite_support_ic}
		\begin{align}
		\mathbb{E}_{p \sim \tilde\gamma^\theta} & \left[ \sum_{\omega \in \Omega} F_{\theta, b^\theta(p), \omega} p_\omega - \right.  \left. c_{\theta, b^\theta(p)}  \right]  = \sum_{p \in \supp(\tilde\gamma^\theta)} \tilde\gamma_p^\theta \left( \sum_{\omega \in \Omega} F_{\theta, b^\theta(p), \omega} p_\omega - c_{\theta, b^\theta(p)}  \right) \label{eq:lem_finite_support_ic_1}\\
		& \hspace{-1.1cm} = \sum_{a \in A}  \Pr_{p \sim \gamma^\theta} \left\{ p\in \hat\pa^{\theta,a}  \right\}  \left( \sum_{\omega \in \Omega} F_{\theta, a, \omega} \mathbb{E}_{p \sim \gamma^\theta} \left[  p_\omega  \mid  p \in \hat\pa^{\theta,a} \right] - c_{\theta, a}  \right) \label{eq:lem_finite_support_ic_2}\\
		& \hspace{-1.1cm}  = \sum_{a \in A}  \Pr_{p \sim \gamma^\theta} \left\{ p\in \hat\pa^{\theta,a}  \right\} \mathbb{E}_{p \sim \gamma^\theta} \left[ \sum_{\omega \in \Omega} F_{\theta, b^\theta(p), \omega}  p_\omega  - c_{\theta, b^\theta(p)} \mid  p \in \hat\pa^{\theta,a} \right] \label{eq:lem_finite_support_ic_3}\\
		& \hspace{-1.1cm}  = \mathbb{E}_{p \sim \gamma^\theta} \left[ \sum_{\omega \in \Omega} F_{\theta, b^\theta(p), \omega} p_\omega - c_{\theta, b^\theta(p)}  \right] \label{eq:lem_finite_support_ic_4} \\
		& \hspace{-1.1cm}  \geq \mathbb{E}_{p \sim \gamma^{\hat \theta}} \left[ \sum_{\omega \in \Omega} F_{\theta, b^\theta(p), \omega} p_\omega - c_{\theta, b^\theta(p)}  \right] \label{eq:lem_finite_support_ic_5}\\
		& \hspace{-1.1cm}  = \mathbb{E}_{p \sim \gamma^{\hat \theta}} \left[ \max_{a \in A} \left\{  \sum_{\omega \in \Omega} F_{\theta, a, \omega} p_\omega - c_{\theta, a}  \right\} \right] \label{eq:lem_finite_support_ic_6}\\
		& \hspace{-1.1cm} = \sum_{a' \in A}  \Pr_{p \sim \gamma^{\hat \theta}} \left\{ p\in \hat\pa^{\hat\theta,a'}  \right\} \mathbb{E}_{p \sim \gamma^{\hat \theta}} \left[ \max_{a \in A} \left\{ \sum_{\omega \in \Omega} F_{ \theta, a, \omega}    p_\omega - c_{ \theta, a}   \right\} \mid  p \in \hat\pa^{\hat \theta,a'} \right] \label{eq:lem_finite_support_ic_7}\\
		& \hspace{-1.1cm}  \geq \sum_{a' \in A}  \Pr_{p \sim \gamma^{\hat \theta}} \left\{ p\in \hat\pa^{\hat\theta,a'}  \right\} \max_{a \in A} \mathbb{E}_{p \sim \gamma^{\hat \theta}} \left[  \sum_{\omega \in \Omega} F_{ \theta, a, \omega}    p_\omega - c_{ \theta, a}  \mid  p \in \hat\pa^{\hat \theta,a'} \right] \label{eq:lem_finite_support_ic_8}\\
		& \hspace{-1.1cm}  = \sum_{a' \in A}  \Pr_{p \sim \gamma^{\hat \theta}} \left\{ p\in \hat\pa^{\hat\theta,a'}  \right\} \max_{a \in A} \left\{  \sum_{\omega \in \Omega} F_{ \theta, a, \omega}  \mathbb{E}_{p \sim \gamma^{\hat \theta}}  \left[  p_\omega \mid  p \in \hat\pa^{\hat \theta,a'} \right]- c_{ \theta, a}  \right\} \label{eq:lem_finite_support_ic_9}\\
		& \hspace{-1.1cm}  = \sum_{p \in \supp(\tilde\gamma^{\hat \theta})} \tilde\gamma_p^{\hat \theta} \max_{a \in A} \left\{  \sum_{\omega \in \Omega} F_{ \theta, a, \omega}  p_\omega - c_{ \theta, a}  \right\} \label{eq:lem_finite_support_ic_10}\\
		& \hspace{-1.1cm}  = \sum_{p \in \supp(\tilde\gamma^{\hat \theta})} \tilde\gamma_p^{\hat \theta}  \left( \sum_{\omega \in \Omega} F_{ \theta, b^\theta(p), \omega}  p_\omega - c_{ \theta, b^\theta(p)} \right) \label{eq:lem_finite_support_ic_11} \\
		&  \hspace{-1.1cm}  =  \mathbb{E}_{p \sim \tilde\gamma^{\hat \theta}} \left[ \sum_{\omega \in \Omega} F_{\theta, b^\theta(p), \omega} p_\omega - c_{\theta, b^\theta(p)}  \right] , \label{eq:lem_finite_support_ic_12}
		\end{align}
	\end{subequations}
\endgroup
	where Equation~\eqref{eq:lem_finite_support_ic_2} holds by definition of $\tilde\Gamma = \{ \tilde\gamma^\theta \}_{\theta \in \Theta}$, Equation~\eqref{eq:lem_finite_support_ic_3} holds by the fact that $b^\theta(p) = a$ for contracts $p \in \hat\pa^{\theta, a}$, Equation~\eqref{eq:lem_finite_support_ic_4} holds since the sets $\hat\pa_{\theta, a}$ define a partition of $\mathbb{R}_+^m$, Equation~\eqref{eq:lem_finite_support_ic_5} holds by the fact that $\Gamma = \{ \gamma^\theta \}_{\theta \in \Theta}$ satisfies Constraints~\eqref{eq:opt_rand_ic} in Problem~\ref{prob:opt_rand}, Equation~\eqref{eq:lem_finite_support_ic_6} holds since action $a^\theta(p)$ is IC given contract $p$, Equation~\eqref{eq:lem_finite_support_ic_7} holds since sets $\hat\pa_{\hat \theta, a}$ define a partition of $\mathbb{R}_+^m$, Equation~\eqref{eq:lem_finite_support_ic_8} holds by Jensen's inequality and the fact that $\max$ is convex, Equation~\eqref{eq:lem_finite_support_ic_9} holds since $F_{\theta, a, \omega}$ and $c_{\theta,a}$ do not depend on $p$, Equation~\eqref{eq:lem_finite_support_ic_10} holds by definition of $\tilde\Gamma = \{ \tilde\gamma^\theta \}_{\theta \in \Theta}$, while Equation~\eqref{eq:lem_finite_support_ic_11} holds by definition of $b^\theta(p)$.
	This concludes the proof.
\end{proof}

\lemmaRandTwo*

\begin{proof}
	We show that for each $\epsilon>0$, there exists a menu $ \Gamma=\{  \gamma^\theta \}_{\theta \in \Theta}$ such that the required condition are satisfied.
	By the definition of supremum for each $\eta>0$, there exists a menu $\bar \Gamma=\{ \bar \gamma^\theta \}_{\theta \in \Theta}$ that provides  expected principal utility at least $1-\eta$. 
	Then, exploiting Lemma~\ref{lem:finite_support}, we can derive a menu $\tilde \Gamma=\{ \tilde \gamma^\theta \}_{\theta \in \Theta}$ with at least the same utility and such that for every $\theta \in \Theta$, it holds $\left| \supp (\tilde \gamma^\theta) \cap \hat\pa^{\theta, a} \right| \leq 1$ for all $a \in A$.  	Let $\tilde \pa^\theta \coloneqq \supp(\tilde \gamma^\theta) $ for each $\theta \in \Theta$. Moreover, for each $\theta$ and $a$ such that $ \supp (\tilde \gamma^\theta) \cap \hat\pa^{\theta, a}  \neq \emptyset$, let  $\tilde p^{\theta,a} =\supp (\tilde \gamma^\theta) \cap \hat\pa^{\theta, a}  \neq \emptyset$. Notice that, for each $\theta$, $\tilde\gamma^\theta$ is supported on $\{\tilde p^{\theta,a}\}_{a \in A:\supp (\tilde \gamma^\theta) \cap \hat\pa^{\theta, a}  \neq \emptyset}$. 
	
	Let $F \coloneqq \min_{\theta \in \Theta, a \in A, \omega \in \Omega: F_{\theta,a,\omega}>0} F_{\theta,a,\omega}$ be the minimum value of the probabilities $F_{\theta,a,\omega}$. \footnote{We assume that each outcome $\omega\in \Omega$ occurs with positive probability at least for a type $\theta \in \Theta$ and an action $a \in A$. If it is not the case, we can safely remove outcome $\omega$ since it never occurs.}  Moreover, let $ Y= \min_{\theta \in \Theta} \mu_\theta$. \footnote{We assume that for each type $\theta \in \Theta$ it holds $\mu_\theta>0$. Otherwise, we can safely remove the type $\theta$ since it never occurs.} Notice that $F \geq 2^{-\poly(I)}$ and $Y \geq 2^{-\poly(I)}$, where $I$ is the instance size.
	Then, we show that for each $\theta \in \Theta$, $p \in \tilde \pa^\theta$, and $\omega \in \Omega$, it holds $\tilde \gamma^{\theta}_p p_\omega \le 4/FY$. 
	By contradiction, assume that there exists  $\tilde \theta \in \Theta$, $\tilde p \in \tilde \pa^\theta$, and $\tilde \omega \in \Omega$ such that
	$ \tilde \gamma^{\tilde \theta}_{\tilde p}\tilde p_{\tilde \omega} > \frac{4}{F  Y } $.
	%
	Then, by the definition of $F$, there exists an agent's type $\theta'$ and an action $a'$, such that $F_{\theta',a',\tilde \omega}\ge F$. The agent of type $\theta'$ reporting type $\tilde \theta$ to the principal would incur an expected payment 
	\begin{align*}
	\sum_{p \in \tilde \pa^{ \tilde \theta}} \tilde \gamma_p^{\tilde \theta} \left[  \sum_{\omega \in \Omega} F_{\theta', b^{\theta'}(p), \omega} p_\omega \right] &  \geq  \sum_{p \in \tilde \pa^{\tilde \theta}} \tilde \gamma_p^{\tilde \theta} \left[  \sum_{\omega \in \Omega} F_{\theta', a', \omega} p_\omega -c_{\theta',a'}\right] + \sum_{p \in \tilde \pa^{\tilde \theta}} \tilde \gamma_p^{\tilde  \theta}  c_{\theta', b^{\theta'}(p)} \\
	&  \geq  \sum_{p \in \tilde \pa^{\tilde \theta}} \tilde \gamma_p^{\tilde \theta}  \sum_{\omega \in \Omega} F_{\theta', a', \omega} p_\omega  -1 \\
	& \ge  \tilde  \gamma_{\tilde p}^{\tilde \theta}  F_{\theta', a', \tilde \omega} \tilde p_{\tilde \omega} -1 \\
	& \ge \tilde \gamma_{\tilde p}^{\tilde \theta}\tilde p_{\tilde \omega} F -1> \frac{4}{F Y}  \cdot F-1  \ge 3/Y.
	\end{align*}
	Thus, by the DSIC Constraints of type $\theta'$, it must be the case that the expected payment to agent's type $\theta'$, \emph{i.e.},
	$\sum_{p \in \tilde \pa^{\theta'}} \tilde \gamma_p^{\theta'} \left[  \sum_{\omega \in \Omega} F_{\theta', b^{\theta'}(p), \omega} p_\omega \right]$, under menu $\tilde \Gamma = \{ \tilde \gamma^\theta \}_{\theta \in \Theta}$ is at least $3/Y-1\ge 2/Y$, otherwise the constraints would be violated (recall that costs are in the range $[0,1]$).
	Hence, the expected payment of the principal (over the agent's types) is at least \[\sum_{\theta \in \Theta} \mu_\theta \sum_{p \in \tilde \pa^\theta} \tilde \gamma_p^\theta   \sum_{\omega \in \Omega} F_{\theta, b^\theta(p), \omega} p_\omega  \ge \mu_{\theta'}   \sum_{p \in \tilde \pa^{\theta'}} \tilde \gamma_p^{\theta'} \left[  \sum_{\omega \in \Omega} F_{\theta', b^{\theta'}(p), \omega} p_\omega \right]\ge 2. \]
	As a result, since the principal's expected reward is at most $1$, the principal's expected utility must be less than $-1$, contradicting that the expected utility of the principal with menu $\tilde \Gamma = \{ \tilde\gamma^\theta \}_{\theta \in \Theta}$ is at least $\SUP-\eta\ge-\eta>-1$ for $\eta <1$.

	Let $\mathcal{Z}$ the set of couples $(\theta,a)\in \Theta \times A$ such that there exists a $p \in  \hat\pa^{\theta, a} $ with $\tilde \gamma^\theta_p>0$. Notice that  $|\mathcal{Z}|\le n \ell$.
	Then, for each couple $(\theta,a) \in \mathcal{Z}$ there exists a payment $p^{\theta,a}$ in $\hat\pa^{\theta, a}$ such that $p^{\theta,a} \le D$, where $D=2^{\poly(I)}$. This holds since $\hat \pa^{\theta,a}$ is non-empty by the definition of $\mathcal{Z}$, while $p\le D$ since $\hat \pa^{\theta,a}$ is defined by linear inequalities such that the number of variables, the number of constraints, and the size of the binary representation of the coefficients can be bounded by a polynomial function of the instance size $I$~\citep{bertsimas1997introduction}. 
	
	Now we show how to build a menu $\Gamma=\{\gamma^\theta\}_{\theta \in \Theta}$ with utility at least $\SUP-\epsilon$ such that conditions i) and ii) hold for each $\theta$.
	Consider the DSIC menu of contracts $\hat \Gamma=\{\hat \gamma^\theta\}_{\theta \in \Theta}$ such that for each $\theta'\in \Theta$ it holds $\hat \gamma^{\theta'}_{p^{\theta,a}}=1/|\mathcal{Z}|$ for each  $(\theta,a) \in \mathcal{Z}$.
	Then, consider the menu $\mathring{ \Gamma}=\{\mathring{ \gamma}^\theta\}_{\theta \in \Theta}$ such that for each $p\in \mathbb{R}_+^m$ it holds $\mathring{ \gamma}_p=\tilde \gamma_p(1-q)+\hat \gamma_p q$, where $q>0$ will be defined in the following. By the linearity of the DSIC Constraints $\mathring{ \Gamma}$ is a DSIC menu of randomized contracts with principal's utility at least $ (1-q)(\SUP-\eta) - q \sum_{\theta \in \Theta} \mu_\theta \mathbb{E}_{p \sim \hat \gamma^\theta} \left[  \sum_{\omega \in \Omega} F_{\theta, b^\theta(p), \omega} r_\omega - \sum_{\omega \in \Omega} F_{\theta, b^\theta(p), \omega} p_\omega  \right]\ge \SUP -\eta-q-q\eta-qD$, where the inequalities follows from  $\mathbb{E}_{p \sim \hat \gamma^\theta} \left[  \sum_{\omega \in \Omega} F_{\theta, b^\theta(p), \omega} r_\omega - \sum_{\omega \in \Omega} F_{\theta, b^\theta(p), \omega} p_\omega  \right]\ge -D$ since distributions $\hat \gamma^\theta$ are supported on contracts with payments at most $D$.
	Hence, setting $\eta=q=\frac{\epsilon}{4D}$, we have that $\mathring{\Gamma}=\{\mathring{ \gamma}^\theta\}_{\theta \in \Theta}$ provides principal's utility at least $\SUP-\epsilon$.
	
	We conclude the proof exploiting an argument equivalent to the one in Lemma~\ref{lem:finite_support} to build a menu  $\Gamma=\{\gamma^\theta\}_{\theta \in \Theta}$ that satisfies the required conditions.
	As shown in the proof of Lemma~\ref{lem:finite_support}, given the menu $\mathring{\Gamma}$, it is possible to build a menu of randomized contracts $ \Gamma $ with at least the same principal's utility such that for every $\theta \in \Theta$, it holds $\left| \supp (\gamma^\theta) \cap \hat\pa^{\theta, a} \right| \leq 1$ for all $a \in A$. We need to prove that also  (ii)  holds. For each $\theta\in \Theta$ and $a \in A$ such that $\supp (\mathring{\gamma}^\theta) \cap \hat\pa^{\theta, a} \neq \varnothing$,  let $\bar p^{\theta,a}\coloneqq \mathbb{E}_{p \sim \mathring{\gamma}^\theta} \left[ p \mid p \in \hat\pa^{\theta,a} \right]$.
	By the definition of $\Gamma$ in Lemma~\ref{lem:finite_support} for each $\theta \in \Theta$ it holds $ \supp (\gamma^\theta)= \{\bar p^{\theta,a}\}_{a \in A:\supp (\mathring{\gamma}^\theta) \cap \hat\pa^{\theta, a} \neq \varnothing}$.
	To conclude the proof, it is sufficient to see that for each $\theta \in \Theta$ and $a \in A$ such that $\supp (\mathring{\gamma}^\theta) \cap \hat\pa^{\theta, a} \neq \varnothing$, it holds $(\theta,a) \in \mathcal{Z}$. Hence $\hat{\gamma}^\theta_{ \bar p^{\theta,a}}= \frac{1}{|\mathcal{Z}|}$ and $\mathring{\gamma}^\theta_{\bar p^{\theta,a}}\ge \frac{q}{|\mathcal{Z}|}$.
	For each $\theta$ and $a$, let $\mathcal{Z}^{\theta,a}$ be the set of couples $(\theta',a') \in \mathcal{Z}$ such that $p^{\theta',a'} \in \pa^{\theta,a}$. Notice that $(\theta,a)\in \mathcal{Z}^{\theta,a}$.
	Thus, for each $\theta \in \Theta$ and $a \in A$ such that $\supp (\mathring{\gamma}^\theta) \cap \hat\pa^{\theta, a} \neq \varnothing$ and $\omega \in \Omega$, it holds 
	\begin{align*}
	\bar p^{\theta,a}_{\omega}=  \mathbb{E}_{p \sim \mathring{\gamma}^\theta} \left[ p \mid p \in \hat\pa^{\theta,a} \right]&= \frac{ (1-q) \tilde \gamma^{\theta}_{\tilde p^{\theta,a}} \tilde p^{\theta,a}+\sum_{(\theta',a') \in \mathcal{Z}^{\theta,a}} \frac{q}{|\mathcal{Z}|}p^{\theta',a'}_\omega}{(1-q) \tilde \gamma^{\theta}_{\tilde p^{\theta,a}} +\sum_{(\theta',a') \in \mathcal{Z}^{\theta,a}} \frac{q}{|\mathcal{Z}|}}\\
	&\le \frac{\frac{4}{FY}+ q D}{\frac{q}{|\mathcal{Z}|}} \\
	&\le \frac{4 n \ell D}{\epsilon}  \left(\frac{4}{FY}+ D\right) \in O (\frac{D^2n \ell}{FY\epsilon}),
	\end{align*}
	where we assume $\tilde p^{\theta,a}$ that is an arbitrary contract in $\hat\pa^{\theta, a}$  when $\supp (\tilde \gamma^\theta) \cap \hat\pa^{\theta, a}  = \emptyset$ and $\tilde p^{\theta,a}$ is not been defined previously. 
	Hence, for each $\theta \in \Theta$, $p \in \supp(\gamma^\theta)$ and $\omega \in \Omega$, it holds $p_\omega =O( \frac{1}{\epsilon} \cdot 2^{\poly(I)})$.
\end{proof}

\lemmaRandThree*

\begin{proof}
	By Lemma~\ref{lem:bounded_payments}, there exists a menu of randomized contracts $\Gamma = \{ \gamma^\theta \}_{\theta \in \Theta}$ with principal's utility at least $\SUP-\epsilon$ such that, for every $\theta \in \Theta$, it holds $\left| \supp (\gamma^\theta) \cap \hat\pa^{\theta, a} \right| \leq 1$ for all $a \in A$ and $\supp(\gamma^\theta) \subseteq \pa^{\epsilon}$.
	We show that it is possible to build another optimal menu of randomized contracts $\tilde\Gamma = \{ \tilde\gamma^\theta \}_{\theta \in \Theta}$ such that $\supp(\tilde\gamma^\theta) \subseteq \pa^{*,\epsilon}$ for every $\theta \in \Theta$.
	
	In order to define $\tilde\Gamma = \{ \tilde \gamma^\theta \}_{\theta \in \Theta}$, for every $\theta \in \Theta$ and contract $p \in \supp(\gamma^\theta)$, let $\avec^p\defeq(a^p_{\theta'} )_{\theta' \in \Theta}$, where $a^p_{\theta'}\defeq b^{\theta'}(p)$.
	Then, given a $\theta \in \Theta$ and a $p \in \supp(\gamma^\theta)$, since $p \in \hat \pa^{\avec^p,\epsilon}$, $\hat \pa^{\avec^p,\epsilon}\subseteq \pa^{\avec^p,\epsilon}$, and $\pa^{\avec^p,\epsilon}$ is a closed polytope, we can apply Carathéodory's theorem to conclude that there exists a probability distribution $\gamma^{\theta, p} \in \Delta_{\pa^{*,\epsilon}}$ that is supported on the vertexes of the polytope $\pa^{\avec^p,\epsilon}$, \emph{i.e.}, it holds $\supp(\gamma^{\theta,p}) \subseteq V(\pa^{\avec^p,\epsilon})$, and such that $\mathbb{E}_{p' \sim \gamma^{\theta, p}} \left[ p' \right] = \sum_{p' \in \supp(\gamma^{\theta, p})} \gamma^{\theta, p}_{p'} \, p' = p$.
	%
	Then, for every $\theta \in \Theta$, we define the distribution $\tilde{\gamma}^\theta \in \Delta_{\pa^{*,\epsilon}}$ so that, for every $p \in \pa^{*,\epsilon}$, it holds:
	\[
	\tilde\gamma^\theta_p \coloneqq \sum_{p'  \in \supp(\gamma^{\theta})}  \gamma^\theta_{p'} \, \gamma^{\theta,p'}_p.
	\]
	
	It remains to show that $\tilde\Gamma = \{ \gamma^\theta \}_{\theta \in \Theta}$ is indeed a menu of randomized contracts with principal's utility at least $\SUP-\epsilon$.
	First, it is immediate to see that the $\tilde\gamma^\theta$ are valid probability distributions over $\pa^{*,\epsilon}$. Formally, for each $\theta \in \Theta$:
	\begin{align*}
	\sum_{p \in \pa^{*,\epsilon}} \tilde\gamma^\theta_p & = \sum_{p \in \pa^{*,\epsilon}} \sum_{p'  \in \supp(\gamma^{\theta})} \gamma^\theta_{p'} \,  \gamma^{\theta,p'}_p= \sum_{p'  \in \supp(\gamma^{\theta})} \gamma^\theta_{p'} \sum_{p \in \pa^{*,\epsilon}} \gamma^{\theta,p'}_p \\
	& = \sum_{p'  \in \supp(\gamma^{\theta})} \gamma^\theta_{p'} = 1.
	\end{align*}
	Next, we prove that $\tilde\Gamma = \{ \gamma^\theta \}_{\theta \in \Theta}$ satisfies Constraints~\eqref{eq:opt_rand_ic} in Problem~\ref{prob:opt_rand}.
	First, notice that, for every $\theta \in \Theta$, $\hat\theta \in \Theta$, and $p \in \supp(\gamma^\theta)$, it holds:
	\begin{subequations}\label{eq:lem_finite_support_star_ic}
		\begin{align}
		\sum_{p' \in \supp(\gamma^{\theta,p})} \gamma^{\theta, p}_{p'} & \left( \sum_{\omega \in \Omega} F_{\hat \theta, b^{\hat \theta}(p'), \omega} \, {p}'_\omega - c_{\hat \theta, b^{\hat \theta}(p')} \right) \\
		&  = \sum_{p' \in \supp(\gamma^{\theta,p})} \gamma^{\theta, p}_{p'} \left( \sum_{\omega \in \Omega} F_{\hat \theta, a^{p}_{\hat \theta}, \omega} \, {p}'_\omega - c_{\hat \theta, a^{p}_{\hat \theta}} \right) \label{eq:lem_finite_support_star_ic_1}\\
		& =  \sum_{\omega \in \Omega} F_{\hat \theta, a^{p}_{\hat \theta}, \omega} \sum_{p' \in \supp(\gamma^{\theta,p})} \gamma^{\theta, p}_{p'} \, {p}'_\omega - c_{\hat \theta, a^p_{\hat \theta}} \label{eq:lem_finite_support_star_ic_2}\\
		& = \sum_{\omega \in \Omega} F_{\hat \theta, a^p_{\hat \theta}, \omega} \, p_\omega- c_{\hat \theta, a^p_{\hat \theta}},\label{eq:lem_finite_support_star_ic_3}
		\end{align}
	\end{subequations}
	where Equation~\eqref{eq:lem_finite_support_star_ic_1} holds since $ a^p_{\hat\theta} \in \mathcal{B}^{\hat\theta}_{p'}$ for contracts $p' \in \supp(\gamma^{\theta,p}) \subseteq V(\pa^{\avec^p})$ and Equation~\eqref{eq:lem_finite_support_star_ic_3} holds by definition of $\gamma^{\theta,p}$.
	Moreover, for every $\theta \in \Theta$ and $\hat \theta \in \Theta$, it is the case that:
	\begingroup
	\allowdisplaybreaks
	\begin{subequations}\label{eq:lem_finite_support_star}
		\begin{align}
		\sum_{p \in \pa^{*,\epsilon}} \tilde\gamma^{\hat \theta}_p& \left(  \sum_{\omega \in \Omega}  F_{\theta, b^\theta(p), \omega} p_\omega - \right.  \left. c_{\theta, b^\theta(p)}   \right)  \\
		&= \sum_{p \in \pa^{*,\epsilon}}  \sum_{p'  \in \supp(\gamma^{\hat \theta})} \gamma^{\hat \theta}_{p'} \,  \gamma^{\hat \theta,p'}_p\left(  \sum_{\omega \in \Omega} F_{\theta, b^\theta(p), \omega} p_\omega - c_{\theta, b^\theta(p)}   \right) \label{eq:lem_finite_support_star_ic_bis_1} \\
		& =  \sum_{p'  \in \supp(\gamma^{\hat \theta})} \gamma^{\hat \theta}_{p'} \sum_{p \in \pa^{*,\epsilon}} \gamma^{\hat \theta,p'}_p \left(  \sum_{\omega \in \Omega} F_{\theta, b^\theta(p), \omega} p_\omega - c_{\theta, b^\theta(p)}   \right) \label{eq:lem_finite_support_star_ic_bis_2} \\
		& =  \sum_{p'  \in \supp(\gamma^{\hat \theta})} \gamma^{\hat \theta}_{p'}\sum_{p \in \supp(\gamma^{\hat \theta,p'})} \gamma^{\hat \theta,p'}_p \left(  \sum_{\omega \in \Omega} F_{\theta, b^\theta(p), \omega} p_\omega - c_{\theta, b^\theta(p)}   \right) \label{eq:lem_finite_support_star_ic_bis_3} \\
		& =  \sum_{p'  \in \supp(\gamma^{\hat \theta})} \gamma^{\hat \theta}_{p'} \left( \sum_{\omega \in \Omega} F_{\theta, a^{p'}_\theta, \omega} \, p'_\omega - c_{\theta, a^{p'}_\theta} \right) \label{eq:lem_finite_support_star_ic_bis_4} \\
		& =  \sum_{p \in \supp(\gamma^{\hat \theta})}  \gamma^{\hat \theta}_{p} \left( \sum_{\omega \in \Omega} F_{\theta, b^\theta(p), \omega} \, p_\omega - c_{\theta, b^\theta(p)} \right) \label{eq:lem_finite_support_star_ic_bis_5} \\
		& = \mathbb{E}_{p \sim \gamma^{\hat \theta}} \left[  \sum_{\omega \in \Omega} F_{\theta, b^\theta(p), \omega} \, p_\omega - c_{\theta, b^\theta(p)} \right], \label{eq:lem_finite_support_star_ic_bis_6} 
		\end{align}
	\end{subequations}
\endgroup
	where Equation~\eqref{eq:lem_finite_support_star_ic_bis_1} holds by definition of $\tilde\gamma^{\hat \theta}$, Equation~\eqref{eq:lem_finite_support_star_ic_bis_3}  holds since $\supp(\gamma^{\hat\theta, {p'}}) \subseteq V(\pa^{\avec^{p'}}) \subseteq \pa^{*,\epsilon}$, Equation~\eqref{eq:lem_finite_support_star_ic_bis_4} holds by Equation~\eqref{eq:lem_finite_support_star_ic}, while Equation~\eqref{eq:lem_finite_support_star_ic_bis_5} holds since $a^{p}_\theta = b^{\theta}(p)$.
	Thus, we can conclude that for every $\theta \in \Theta$ and $\hat \theta \in \Theta : \hat \theta \neq \theta$, it holds:
	\begin{align*}
	\sum_{p \in \pa^{*,\epsilon}} \tilde\gamma^{\theta}_p \left(  \sum_{\omega \in \Omega} F_{\theta, b^\theta(p), \omega} p_\omega - c_{\theta, b^\theta(p)}   \right)  & = \mathbb{E}_{p \sim \gamma^{\theta}} \left[  \sum_{\omega \in \Omega} F_{\theta, b^\theta(p), \omega} \, p_\omega - c_{\theta, b^\theta(p)} \right] \\
	& \geq \mathbb{E}_{p \sim \gamma^{\hat \theta}} \left[  \sum_{\omega \in \Omega} F_{\theta, b^\theta(p), \omega} \, p_\omega - c_{\theta, b^\theta(p)} \right] \\
	& = \sum_{p \in \pa^{*,\epsilon}} \tilde\gamma^{\hat \theta} \left(  \sum_{\omega \in \Omega} F_{\theta, b^\theta(p), \omega} p_\omega - c_{\theta, b^\theta(p)}   \right),
	\end{align*}
	where the equalities comes from Equation \eqref{eq:lem_finite_support_star} and the inequality from DSIC of $\Gamma$. 
	Finally, we need to show that the principal's expected utility does \emph{not} decrease.
	Formally, we have:
	\begin{subequations}\label{eq:lem_finite_support_star_ic_tris}
		\begin{align}
		\sum_{\theta \in \Theta}& \mu_\theta  \sum_{p \in \pa^{*,\epsilon}} \tilde\gamma_p^\theta  \left(  \sum_{\omega \in \Omega} F_{\theta, b^\theta(p), \omega} r_\omega - \sum_{\omega \in \Omega} F_{\theta, b^\theta(p), \omega} p_\omega  \right) \nonumber \\
		& = \sum_{\theta \in \Theta} \mu_\theta \sum_{p \in \pa^{*,\epsilon}} \sum_{p'  \in \supp(\gamma^{\theta})}  \gamma^\theta_{p'} \, \gamma^{\theta,p'}_p \left(  \sum_{\omega \in \Omega} F_{\theta, b^\theta(p), \omega} r_\omega - \sum_{\omega \in \Omega} F_{\theta, b^\theta(p), \omega} p_\omega  \right) \label{eq:lem_finite_support_star_ic_tris_1}  \\
		& = \sum_{\theta \in \Theta} \mu_\theta\sum_{p'  \in \supp(\gamma^{\theta})}  \gamma^\theta_{p'} \sum_{p \in \supp(\gamma^{\theta,p'})} \gamma^{\theta, p'}_p \left( \sum_{\omega \in \Omega} F_{\theta, b^\theta(p), \omega} r_\omega - \sum_{\omega \in \Omega} F_{\theta, b^\theta(p), \omega} p_\omega  \right) \label{eq:lem_finite_support_star_ic_tris_2} \\
		& \geq  \sum_{\theta \in \Theta} \mu_\theta\sum_{p'  \in \supp(\gamma^{\theta})}  \gamma^\theta_{p'} \sum_{p \in \supp(\gamma^{\theta,p'})} \gamma^{\theta, p'}_p \left( \sum_{\omega \in \Omega} F_{\theta, a^{p'}_\theta, \omega} r_\omega - \sum_{\omega \in \Omega} F_{\theta, a^{p'}_\theta, \omega} p_\omega  \right)  \label{eq:lem_finite_support_star_ic_tris_3} \\
		& = \sum_{\theta \in \Theta} \mu_\theta \sum_{p'  \in \supp(\gamma^{\theta})}  \gamma^\theta_{p'}\left( \sum_{\omega \in \Omega} F_{\theta, a^{p'}_\theta, \omega} r_\omega - \sum_{\omega \in \Omega} F_{\theta, a^{p'}_\theta, \omega} \sum_{p \in \supp(\gamma^{\theta,p'})} \gamma^{\theta, p'}_p \, p_\omega  \right) \label{eq:lem_finite_support_star_ic_tris_4}  \\
		& = \sum_{\theta \in \Theta} \mu_\theta \sum_{p'  \in \supp(\gamma^{\theta})}  \gamma^\theta_{p'}  \left( \sum_{\omega \in \Omega} F_{\theta, a^{p'}_\theta, \omega} r_\omega -\sum_{\omega \in \Omega} F_{\theta, a^{p'}_\theta, \omega} \, p'_\omega \right) \label{eq:lem_finite_support_star_ic_tris_5} \\
		& = \sum_{\theta \in \Theta} \mu_\theta \sum_{p \in \supp(\gamma^\theta)} \gamma^\theta_{p} \left(\sum_{\omega \in \Omega} F_{\theta, b^\theta(p), \omega} r_\omega - \sum_{\omega \in \Omega} F_{\theta, b^\theta(p), \omega} \, p_\omega \right)\label{eq:lem_finite_support_star_ic_tris_6}  \\
		& = \sum_{\theta \in \Theta} \mu_\theta \mathbb{E}_{p \sim \gamma^\theta} \left[  \sum_{\omega \in \Omega} F_{\theta, b^\theta(p), \omega} r_\omega - \sum_{\omega \in \Omega} F_{\theta, b^\theta(p), \omega} p_\omega  \right], \label{eq:lem_finite_support_star_ic_tris_7} 
		\end{align}
	\end{subequations}
	where Equation~\eqref{eq:lem_finite_support_star_ic_tris_1} holds by definition of $\tilde{\gamma}^\theta$, Equation~\eqref{eq:lem_finite_support_star_ic_tris_2} holds by re-arranging the summations and by $\supp(\gamma^{\theta, p'}) \subseteq V(\pa^{\avec^{p'}}) \subseteq \pa^{*,\epsilon}$, Equation~\eqref{eq:lem_finite_support_star_ic_tris_3} holds since $a^{p'}_\theta \in \mathcal{B}_{p}^\theta$ for each $p \in \supp(\gamma^{\theta,p'})$ and $b^\theta(p)$ breaks ties in favor of the principal, Equation~\eqref{eq:lem_finite_support_star_ic_tris_5} holds by definition of $\gamma^{\theta, p'}$, while Equation~\eqref{eq:lem_finite_support_star_ic_tris_6} holds since $a^p_\theta \coloneqq b^\theta(p)$.
	This concludes the proof.
\end{proof}


\end{document}